\documentclass[a4paper, aps, pra, reprint, superscriptaddress, nofootinbib]{revtex4-1}

\def\captionfigexampledes{
    \textbf{Abstract example of simulating a quantum protocol with discrete events}. When setting up the simulation, protocol actions are defined to occur when a specific event occurs, as in the setup
of a real system. Instead of performing a continuous time evolution, the simulation advances to the next event, and then automatically executes the actions that should occur when the event takes place.
Any action may again define future events to be triggered after a certain (stochastic) amount of time has elapsed. 
For concreteness a simplified quantum teleportation example is shown where a qubit, shown as an orange circle with arrow, is teleported between the quantum memories of Alice and Bob. Here, entanglement is produced using an abstract source sending two qubits, shown as blue circles with arrows, to Alice and Bob. Once the qubit has traversed the fibre and reaches Alice's lab, an event is triggered that invokes the simulation of Alice's Bell state measurement (BSM) apparatus. 
The simulation engine steps from event to events defined by the next action, which generally occur at irregular intervals.
This approach allows time-dependent physical non-idealities, such as quantum decoherence, to be accurately tracked.
}

\def\captionfigexamplerepeatdistil{
    \textbf{Illustrative example of a NetSquid use case.} Each sub-figure explains part of the modelling and simulation process. For greater clarity the figures are not based on real simulation data.
The scenario shown is a quantum repeater utilising entanglement distillation (see main text).
\textbf{a)} The setup of a quantum network using node and connection components.
\textbf{b)} A zoom in showing the subcomponents of the entangling connection component. The quantum channels are characterised using fibre delay and loss models. The quantum source samples from an entangled bipartite state sampler when externally triggered by the classical channel.
\textbf{c)} A zoom in of the quantum memory positions within a quantum processor illustrating their physical gate topology. The physical single-qubit instructions possible on each memory in this example are the Pauli ($X$, $Y$, $Z$), Hadamard ($H$), and $X$-rotation ($R_X$) gates, and measurement. The blue-dashed arrows show the positions and control direction (where applicable) for which the two-qubit instructions controlled-$X$ (CNOT) and swap are possible. Noise and error models for the memories and gates are also assigned.
\textbf{d)} Illustration of a single simulation run. Time progresses by discretely stepping from event to event, with new events generated as the simulation proceeds. Qubits are represented by circles, which are numbered according to the order they were generated. A star shows the moment of generation. The curved lines between qubits denote their entanglement with the colour indicating fidelity. The state of each qubit is updated as it is accessed during the simulation, for instance to apply time-dependent noise from waiting in memory.
\textbf{e)} A zoom in of the distillation protocol. The shared quantum states of the qubits are combined in an entangling step, which then shrinks as two of the qubits are measured. The output is randomly sampled, causing the simulation to choose one of two paths by announcing success or failure.
\textbf{f)} A plot illustrating the stochastic paths followed by multiple independent simulation runs over time, labeled by their final end-to-end fidelity $F_i$. The blue dashed line corresponds to the run shown in panel~(d). The runs are typically executed in parallel. Their results are statistically analysed to produce performance metrics such as the average outcome fidelity and run duration.
}

\def\captionfigswitchtopodiagram{
	    \textbf{A quantum switch in a star-shaped network topology as studied by Vardoyan et al.\cite{vardoyan2019stochastic}}.
	    The switch (central node) is connected to a set of users (leaf nodes) via an optical fibre link that distributes perfect Bell pairs at random times, following an exponential distribution with mean rate $\mu \propto e^{-\beta L}$, where $L$ denotes the distance of the link and $\beta$ the attenuation coefficient.
Associated with each link there is a buffer that can store $B$ qubits at each side of the link. 
	    As soon as $n$ Bell pairs with different leaves are available, the switch performs a measurement in the $n$-partite Greenberger-Horne-Zeilinger (GHZ) basis, which results in an $n$-partite GHZ state shared by the leaves.
The GHZ-basis measurement consists of: first, controlled-X gates with the same qubit as control; next, a Hadamard ($H$) gate on the control qubit; finally, measurement of all qubits individually.
	    The figure shows 4 leaf nodes, GHZ size $n=3$ and a buffer size $B=2$.
}

\def\captionfigswitch{
\textbf{Performance analysis of the quantum switch with 9 users using NetSquid.}
	\textbf{(a)} Capacity as a function of the buffer size (number of quantum memories that the switch has available per user) for either $2-$ or $4-$qubit Greenberger-Horne-Zeilinger (GHZ)-states.
    For each scenario, the generation rate $\mu$ of pairs varies per user. 
	For the blue scenario (2-partite entanglement, $\mu=[1.9, 1.9, 1.9, 1, 1, 1, 1, 1, 1]$ MHz), the capacity was determined analytically by Vardoyan et al. using Markov Chain methods~\cite[Figure 8]{vardoyan2019stochastic}.
    Here we extend this to 4-partite entanglement (orange scenario, same $\mu$s), for which Vardoyan et al. have found an upper bound (by assuming unbounded buffer and each $\mu=$ maximum of original rates $=1.9$ MHz) but no exact analytical expression.
The green scenario ($\mu=[15, 1.9, 1.9, 1, 1, 1, 1, 1, 1]$ MHz) does not satisfy the stability condition for the Markov chain for unbounded buffer size (each leaf's rate $<$ half of sum of all rates) so in that case steady-state capacity is not well-defined.
We note that regardless of buffer size, the switch has a single link to each user, which is the reason why the capacity does not scale linearly with buffer size.
\textbf{(b)} Average fidelity of the produced entanglement on the user nodes (no analytical results known) with unbounded buffer size.
    The fact that the green curve has lower fidelity than the blue one, while the former has higher rates, can be explained from the fact that the protocol prioritises entanglement which has the longest storage time (see Supplementary Note~\ref{app:switch}).
Each data point represents the average of 40 runs (each 0.1 ms in simulation).
Standard deviation is smaller than dot size.
}

\def\captionfigrepchainA{
\textbf{Performance of repeaters based on nitrogen-vacancy (NV) centres in diamond.}
    \textbf{(a)} Fidelity and entanglement distribution rate achieved with near-term and $10\times$ improved hardware (Supplementary Note \ref{app:nv-physical-modelling}) with the \swapasap~protocol.
    Dashed line represents classical fidelity threshold of $0.5$.
	We observe that for near-term hardware, the use of 3 repeaters yields worse performance in terms of fidelity than the no-repeater setup.
	For improved hardware we observe (i) that for approx. 0 - 750 kms, repeaters improve upon rate by orders of magnitude while still producing entanglement (fidelity$>0.5$), while (ii) for approx. 750 - 1500 kms, repeaters outperform in both rate and fidelity.
	\textbf{(b-c)} 
	Fidelity and rate achieved without and with repeaters (1, 3 or 7 repeaters) as function of a hardware improvement factor (Methods, section \nameref{sec:improvement-factor}) for two typical distances from both distance regime (i) and (ii), for two protocols \swapasap~and \withdistill.
    For the repeater case, only the best-performing number-of-repeaters \& protocol in terms of achieved fidelity is shown in (b), accompanied by its rate in (c).
	Each data point represents the average over (a) 200 and (b) 100 runs.
	Standard deviation is smaller than dot size.
}

\def\captionfigsensitivityanalysis{
\textbf{Sensitivity of fidelity in various hardware parameters for nitrogen-vacancy (NV) repeater chains.}
    The NV hardware model consists of \textapprox $15$ parameters and from those we focus on four parameters in this figure: (A) two-qubit gate fidelity, (B) detection probability, (C) induced storage qubit noise and (D) visibility.
    We start by improving all \textapprox $15$ parameters, including the four designated ones, using an improvement factor of 3 (Methods, section~\nameref{sec:improvement-factor}).
    Then, for each of the four parameters only, we individually decrease their improvement factor to 2, or increase it to 10 or 50.
    The figure shows the resulting fidelity (horizontal and vertical grid lines; dashed line indicates maximal fidelity which can be attained classically).
    Note that at an improvement factor of 3 (orange line), all $\textapprox 15$ parameters are improved by 3 times, resulting in a fidelity of 0.39.
    In addition, we vary the improvement factor for combinations of two of the four parameters (diagonal lines).
	The $3\times$ improved parameter values can be found in Supplementary Table~\ref{table:parameters}.
	The other values (at $2/10/50$$\times$) are approximately: two-qubit gate fidelity $F_{EC}$ ($0.985/0.997/0.9994$), detection probability $\pdetnofibre$ ($6.8\%/58\%/90\%$), induced storage qubit noise $N_{1/e}$ ($2800/14000/70000$), visibility $V$ ($95\%/99\%/99.8\%$).
    The fidelities shown are obtained by simulation of the \swapasap~protocol (3 repeaters) with a total spanned distance of 500 km.
	Each data point represents the average of 1000 runs (standard deviation on fidelity $< 0.002$).
}

\def\captionfigafcmemorycomparison{
    \textbf{Performance comparison of a single quantum repeater with atomic frequency comb (AFC) or electronically induced transparency (EIT) quantum memories}.
    Shown are: \textbf{(a)} the secret key rate in secret bits per entanglement generation attempt, \textbf{(b)} the quantum bit error rate (QBER) in the X and Z bases \textbf{(c)} the average number of attempts necessary for one successful end-to-end entanglement generation.
    Each data point is obtained using 10.000 (EIT) or 30.000 (AFC) successful attempts at generating entanglement between the end nodes.
	Solid lines are fits.
	Note that for the secret key plot we use logarithmic scale with added 0 at the origin of the axes.
	Error bars denote standard deviation and are symmetrical.
}

\def\captionqubitbenchmark{
        \textbf{Runtime comparison of NetSquid's quantum state formalisms.}\
        Runtime comparisons of the available quantum state formalisms in NetSquid as well as ProjectQ ket vector for two benchmark use cases. The KET, DM, STAB and GSLC formalisms refer to the use of ket vectors, density matrices, stabiliser tableaus and graph states with local Cliffords, respectively.
        \textbf{(a)} Generating a Greenberger-Horne-Zeilinger (GHZ) state. Qubits are \textit{split} off from the shared quantum state after a measurement. For the KET formalism the effect of turning off memoization (dotted line) is also shown.
        \textbf{(b)} Quantum computation involved in a repeater chain. Each formalism is shown with qubits split (dotted lines) versus being kept \textit{in-place} (solid lines) after measurement.
}

\def\captionrepchainbenchmark{
    \textbf{Runtime profile of a repeater chain simulation using Netsquid.} Runtime profile for a repeater chain simulation with a varying number of nodes in the chain. The maximum quantum state size is four qubits. The total time spent in the functions of each NetSquid subpackage and its main package dependencies (in italics) is shown. The dark hatched bands show the largest contribution from a single function in each NetSquid sub-package, as well as in NumPy and uncategorised (\textit{other}) functions.
The sub-packages are stacked in the same order as they are listed in the legend.
}

\def\captionfignetsquidarch{
        \textbf{Overview of NetSquid's software architecture.} The sub-packages that make up the NetSquid package are shown stacked in relation to each other
    and the PyDynAA package dependency.
    The main classes in each (sub-)package are highlighted, and their relationships in terms of inheritance, composition and aggregation are shown.
    Also shown are the key modules users interact with, which are described in the main text.
    In this paper NetSquid version 0.10 is described.
}

\usepackage{xpatch}
\makeatletter
\xpatchcmd{\@ssect@ltx}{\@xsect}{\protected@edef\@currentlabelname{#8}\@xsect}{}{}
\xpatchcmd{\@sect@ltx}{\@xsect}{\protected@edef\@currentlabelname{#8}\@xsect}{}{}
\makeatother
\usepackage[cm]{fullpage}
\usepackage[T1]{fontenc}
\usepackage{amssymb, amsmath, amsthm}
\usepackage{bbold}
\usepackage{csvsimple}
\usepackage{tikz}
\usepackage{verbatim}
\usepackage{subcaption}
\usepackage{algorithm2e}
\usepackage{todonotes}
\usepackage{amsmath}
\usepackage{amssymb}
\usepackage{amsthm}
\usepackage{mathrsfs}
\usepackage{mathtools} 
\usepackage[english]{babel}
\usepackage{hhline}
\usepackage{fullpage}
\usepackage{verbatim}
\usepackage{enumerate}
\usepackage{ulem}
\usepackage{caption}
\usepackage{graphicx}
\usepackage[margin=0.9in]{geometry}
\usepackage{wrapfig}
\usepackage{color}
\usepackage{calc}
\usepackage{subfiles}
\usepackage[bookmarksopen,bookmarksdepth=3]{hyperref}
\usepackage{nameref}
\usepackage{physics}
\usepackage{qcircuit}
\usepackage{array}
\usepackage{inconsolata}
\usepackage{xr-hyper}
\usepackage{rotating} 
\usepackage{adjustbox}  
\definecolor{darkblue}{RGB}{0, 0, 130}
\definecolor{darkred}{RGB}{130, 0, 0}
\definecolor{darkgreen}{RGB}{10, 70, 0}
\hypersetup{colorlinks, breaklinks, linkcolor=darkred, urlcolor=darkblue, citecolor=darkgreen}
\newcommand{\swapasap}{\textsc{swap-asap}}
\newcommand{\withdistill}{\textsc{nested-with-distill}}
\newcommand{\textapprox}{\raisebox{0.5ex}{\texttildelow}}
\newcommand{\netsquidpapertitle}{NetSquid, a NETwork Simulator for QUantum Information\\ using Discrete events}

\addto\captionsenglish{%

}

\newcommand{\move}{\textsc{store}}
\newcommand{\unmove}{\textsc{retrieve}}

\newcommand{\pdc}{p_{\text{dc}}}
\newcommand{\pde}{p_{\text{dexc}}}
\newcommand{\pdet}{p_{\text{det}}}
\newcommand{\pdetnofibre}{p_{\text{det}}^{\text{nofibre}}}

\newcommand{\poslose}{\pi_{\textnormal{lose}}}
\newcommand{\poskeep}{\pi_{\textnormal{keep}}}
\newcommand{\Plose}{P_{\textnormal{lose}}}
\newcommand{\Pkeep}{P_{\textnormal{keep}}}
\newcommand{\unit}{1\!\!1}
\newcommand{\X}{\mathsf{X}}
\newcommand{\Y}{\mathsf{Y}}
\newcommand{\Z}{\mathsf{Z}}

\newtheorem{prop}{Proposition}

\newcommand{\Ypitwo}{\gate{R_Y(\frac{\pi}{2})}}
\newcommand{\Yminuspitwo}{\gate{R_Y(-\frac{\pi}{2})}}
\newcommand{\Zpitwo}{\gate{R_Z(\frac{\pi}{2})}}
\newcommand{\Zminuspitwo}{\gate{R_Z(-\frac{\pi}{2})}}
\newcommand{\CXDIR}{\gate{R_X(\pm \frac{\pi}{2})}}
\newcommand{\CXDIRminus}{\gate{R_X(\mp \frac{\pi}{2})}}
\newcommand{\entgen}{\textsc{entgen}}
\newcommand{\distill}{\textsc{distill}}
\newcommand{\swap}{\textsc{swap}}

\def\netsquidpaperauthors{
\author{Tim Coopmans}
\thanks{These authors contributed equally.}
\affiliation{QuTech, Delft University of Technology and TNO, Lorentzweg 1, 2628 CJ Delft, The Netherlands}
\affiliation{Kavli Institute of Nanoscience, Delft, The Netherlands}
\author{Robert Knegjens}
\thanks{These authors contributed equally.}
\affiliation{QuTech, Delft University of Technology and TNO, Lorentzweg 1, 2628 CJ Delft, The Netherlands}
\author{Axel Dahlberg}
\affiliation{QuTech, Delft University of Technology and TNO, Lorentzweg 1, 2628 CJ Delft, The Netherlands}
\affiliation{Kavli Institute of Nanoscience, Delft, The Netherlands}
\author{David Maier}
\affiliation{QuTech, Delft University of Technology and TNO, Lorentzweg 1, 2628 CJ Delft, The Netherlands}
\affiliation{Kavli Institute of Nanoscience, Delft, The Netherlands}
\author{Loek Nijsten}
\affiliation{QuTech, Delft University of Technology and TNO, Lorentzweg 1, 2628 CJ Delft, The Netherlands}
\author{Julio de Oliveira Filho}
\affiliation{QuTech, Delft University of Technology and TNO, Lorentzweg 1, 2628 CJ Delft, The Netherlands}
\author{Martijn Papendrecht}
\affiliation{QuTech, Delft University of Technology and TNO, Lorentzweg 1, 2628 CJ Delft, The Netherlands}
\author{Julian Rabbie}
\affiliation{QuTech, Delft University of Technology and TNO, Lorentzweg 1, 2628 CJ Delft, The Netherlands}
\affiliation{Kavli Institute of Nanoscience, Delft, The Netherlands}
\author{Filip Rozp\k{e}dek}
\affiliation{QuTech, Delft University of Technology and TNO, Lorentzweg 1, 2628 CJ Delft, The Netherlands}
\affiliation{Kavli Institute of Nanoscience, Delft, The Netherlands}
\affiliation{Pritzker School of Molecular Engineering, University of Chicago, Chicago, IL 60637, USA}
\author{Matthew Skrzypczyk}
\affiliation{QuTech, Delft University of Technology and TNO, Lorentzweg 1, 2628 CJ Delft, The Netherlands}
\affiliation{Kavli Institute of Nanoscience, Delft, The Netherlands}
\author{Leon Wubben}
\affiliation{QuTech, Delft University of Technology and TNO, Lorentzweg 1, 2628 CJ Delft, The Netherlands}
\author{Walter de Jong}
\affiliation{SURF, P.O. Box 94613, 1090 GP Amsterdam, The Netherlands}
\author{Damian Podareanu}
\affiliation{SURF, P.O. Box 94613, 1090 GP Amsterdam, The Netherlands}
\author{Ariana Torres-Knoop}
\affiliation{SURF, P.O. Box 94613, 1090 GP Amsterdam, The Netherlands}
\author{David Elkouss}
\thanks{These authors jointly supervised this work.}
\affiliation{QuTech, Delft University of Technology and TNO, Lorentzweg 1, 2628 CJ Delft, The Netherlands}
\author{Stephanie Wehner}
\thanks{These authors jointly supervised this work.}
\affiliation{QuTech, Delft University of Technology and TNO, Lorentzweg 1, 2628 CJ Delft, The Netherlands}
\affiliation{Kavli Institute of Nanoscience, Delft, The Netherlands}
\affiliation{Corresponding author: \href{mailto:s.d.c.wehner@tudelft.nl}{s.d.c.wehner@tudelft.nl}}
}

\begin{document}

\title{\netsquidpapertitle}

\netsquidpaperauthors

\date{\today}

\begin{abstract}
    \section*{Abstract}
In order to bring quantum networks into the real world,
we would like to determine the requirements of quantum network protocols including the underlying quantum hardware.
Because detailed architecture proposals are generally too complex for mathematical analysis, it is natural to employ numerical simulation.
Here we introduce NetSquid, the NETwork Simulator for QUantum Information using Discrete events, a discrete-event based platform for simulating all aspects of quantum networks and modular quantum computing systems, ranging from the physical layer and its control plane up to the application level.
We study several use cases to showcase NetSquid's power,
including detailed physical layer simulations of repeater chains based on nitrogen vacancy centres in diamond as well as atomic ensembles.
We also study the control plane of a quantum switch beyond its analytically known regime,
and showcase NetSquid's ability to investigate large networks by simulating entanglement distribution over a chain of up to one thousand nodes.
\end{abstract}

\maketitle

\section{Introduction}
Quantum communication can be used to connect distant quantum devices into a quantum network. At short distances, networking quantum devices provides a path towards a scalable distributed quantum computer \cite{vanmeter2016scalable}. At larger distances, interconnected quantum networks allow for communication tasks between distant users on a quantum internet. Both types of quantum networks have the potential for large societal impact. First, analogous to classical computers, it is likely that any approach for scaling up a quantum computer so that it can solve real world problems impractical to treat on a classical computer, will require the interconnection of different modules~\cite{lekitsch2017blueprint, monroe2014largescale, stephens2008deterministic}.
Furthermore, quantum communication networks enable a host of tasks that are impossible using classical communication~\cite{wehner2018quantum}.

For both types of networks, many challenges must be overcome before they can fulfil their promise. The exact extent of these challenges remains unknown,
and precise requirements to guide the construction of large-scale quantum networks are missing. At the physical layer, many proposals exist for quantum repeaters that can carry qubits over long distances (see e.g.~\cite{munro2015inside,muralidharan2016optimal,gisin2007quantum} for an overview). Using analytical methods~\cite{ briegel1998quantum, duer1999quantum, duan2001long, amirloo2014quantum, asadi2018quantum, bernardes2011rate, borregaard2015heralded, bruschi2014repeat, chen2007fault, collins2007multiplexed, guha2015rate, hartmann2007role, jiang2008quantum, nemoto2016photonic, razavi2009quantum, razavi2006long, simon2007quantum, vinay2017practical, wu2020nearterm, sangouard2007long-distance, sangouard2008robust, sangouard2009quantum} and ad-hoc simulations~\cite{abruzzo2013quantum, brask2008memory, muralidharan2014ultrafast, pant2017ratedistance, ladd2006hybrid, loock2006hybrid, zwerger2017quantum, jiang2007fast} rough estimates for the requirements of such hardware proposals have been obtained. Yet, while greatly valuable to set minimal requirements, these studies still provide limited detail. Even for a small-scale quantum network, the intricate interplay between many communicating devices, and the resulting timing dependencies makes a precise analysis challenging. To go beyond, we would like a tool that can incorporate not only a detailed physical modelling, but also account for the implications of time-dependent behaviour.

Quantum networks cannot be built from quantum hardware alone; in order to scale they need a tightly integrated classical control plane, i.e. protocols responsible for orchestrating quantum network devices to bring entanglement to users. Fundamental differences between quantum and classical information demand an entirely new network stack in order to create entanglement, and run useful applications on future quantum networks~\cite{dahlberg2019linklayer, vanmeter2012quantumnetworking, aparicio2011protocol, vanmeter2013designing, vanmeter2007system, pirker2019stack}. 
The design of such a control stack is furthermore made challenging by numerous technological limitations of quantum devices. A good example is given by the limited lifetimes of quantum memories, due to which delays in the exchange of classical control messages have a direct impact on the performance of the network.
To succeed, we hence need to understand how possible classical control strategies do perform on specific quantum hardware. 
Finally, to guide overall development, we need to understand the requirements of quantum network applications themselves. 
Apart from quantum key distribution (QKD) \cite{acin2007device, branciard2012onesided, scarani2009securityqkd, xu2020secure, pirandola2019advances} and a few select applications~\cite{barz2012demonstration, nickerson2014freely, lipinska2018anonymous, khabiboulline2019optical}, little is known about the requirements of quantum applications~\cite{wehner2018quantum} on imperfect hardware. 

Analytical tools offer only a limited solution for our needs. Statistical tools (see e.g.~\cite{shchukin2017waitingPRA, vinay2019statistical, vardoyan2019stochastic, razavi2009physical}) have been used to make predictions about certain aspects of large regular networks using simplified models, but are of limited use for more detailed studies.
Information theory~\cite{wilde2013quantum} can be used to benchmark implementations against the ideal performance.
However, it gives no information about how well a specific proposal will perform. 
As a consequence, numerical methods are of great use to go beyond what is feasible using an analytical study. Ad-hoc simulations of quantum networks have indeed been used to provide further insights on specific aspects of quantum networks (see e.g.~\cite{abruzzo2013quantum, brask2008memory, muralidharan2014ultrafast, pant2017ratedistance, ladd2006hybrid, loock2006hybrid, zwerger2017quantum, jiang2007fast, pant2017routing, kuzmin2019scalable, khatri2019practical}).
However, while greatly informative, setting up ad-hoc simulations for each possible networking scenario to a level of detail that might allow the determination of more precise requirements is cumbersome, and does not straightforwardly lend itself to extensive explorations of new possibilities.

We would hence like a simulation platform that satisfies at least the following three features: First, accuracy: the tool should allow modelling the relevant physics. This includes the ability to model time-dependent noise 
and network behaviour. Second, modularity: it should allow stacking protocols and models together in order to construct complicated network
simulations out of simple components. This includes the ability to investigate not only the physical layer hardware, but the entirety of the quantum network system including how different control protocols behave on a given hardware setup. Third, scalability: it should allow us to investigate large networks. 

Evaluating the performance of large classical network systems, including their time-dependent behaviour is the essence of classical network analysis. Yet, even for classical networks, the predictive power of analytical methods is limited due to complex emergent behaviour arising from the interplay 
between many network devices. Consequently, a crucial tool in the design of such networks are network simulators, which form a tool to test new ideas, 
and many such simulators exist for purely classical networks~\cite{omnet, riley2010ns-3, lantz2010network}.
However, such simulators do not allow the simulation of quantum behaviour. 

In the quantum domain, many simulators are known for the simulation of quantum computers (see e.g.~\cite{fingerhuth2018opensource}). However, the task of simulating a quantum network differs greatly from simulating the execution of one monolithic quantum system. In the network, many devices are communicating with each other both quantumly and classically, leading to complex stochastic behaviour, and asynchronous and time-dependent events.
From the perspective of building a simulation engine, there is also an important difference that allows for gains in the efficiency of the simulation. A simulator for a quantum computation is optimised to track large entangled states.
In contrast, in a quantum network the state space grows and shrinks dynamically as qubits get measured or entangled, but for many protocols, at any moment in time the state space describing the quantum state of the network is small. We would thus like a simulator capable of exploiting this advantage.

In this paper we introduce the quantum network simulator NetSquid: the NETwork Simulator for QUantum Information using Discrete events.
NetSquid is a software tool (available as a package for Python and previously made freely available online~\cite{netsquid-website})
for accurately simulating quantum networking and modular computing systems that are subject to physical non-idealities.
It achieves this by integrating several key technologies:
a discrete-event simulation engine, a specialised quantum computing library, a modular framework for modelling quantum hardware devices,
and an asynchronous programming framework for describing quantum protocols.
We showcase the utility of this tool for a range of applications by studying several use cases: the analysis of a control plane protocol beyond its analytically accessible regime, predicting the performance of protocols on realistic near-term hardware, and benchmarking different quantum devices.
These use cases, in combination with a scalability analysis,
demonstrate that NetSquid achieves all three features set forth above.
Furthermore, they show its potential as a general and versatile design tool for quantum networks, as well as for modular quantum computing architectures.

\section{Results and Discussion}
\label{sec:results}

\subsection{NetSquid in a nutshell}  
\label{sec:netnut}

\begin{figure*}[ht]
\centering
\includegraphics[width=1.0\textwidth]{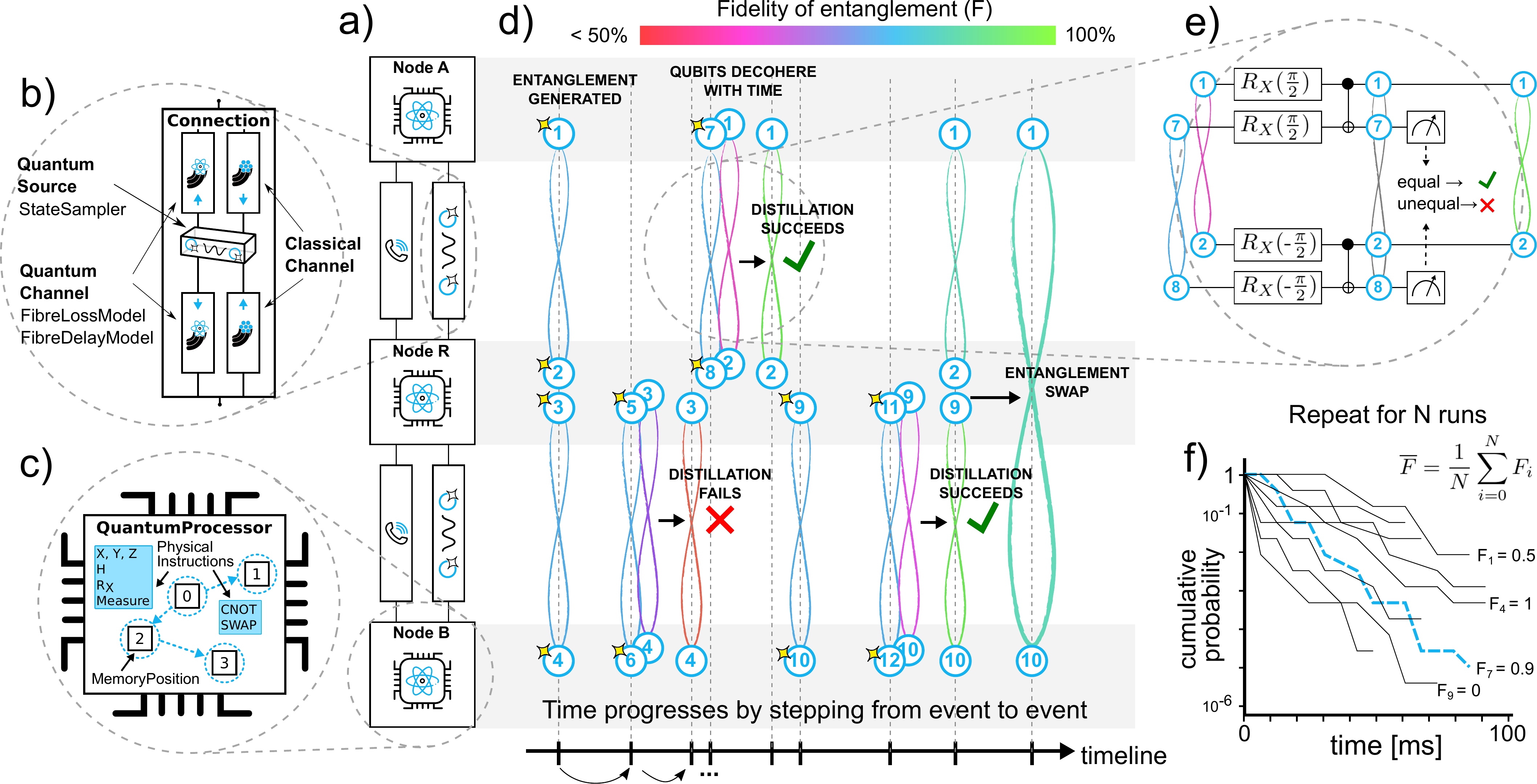}
\caption{\captionfigexamplerepeatdistil}
\label{fig:netsquid_distil_repeat}
\end{figure*}

Simulating a quantum network with NetSquid is generally performed in three steps.
Firstly, the network is modelled using a modular framework of components and physical models.
Next, protocols are assigned to network nodes to describe the intended behaviour.
Finally, the simulation is executed for a typically large number of independent runs to collect statistics with which to determine the performance of the network.
To explain these steps and the features involved further, we consider a simple use case for illustration.
For a more detailed presentation of the available functionality and design of the NetSquid framework see section \nameref{sec:methods-netsquid} of the Methods.

The scenario we will consider is the analysis of an entanglement distribution protocol over a quantum repeater chain with three nodes.
The goal of the analysis is to estimate the average output fidelity of the distributed entangled pairs.
The entanglement distribution protocol is depicted in Figure~\ref{fig:netsquid_distil_repeat}(d-e).
It works as follows. 
First, the intermediate node generates two entangled pairs with each of its adjacent neighbours. Entanglement generation is modelled as a stochastic process that succeeds with a certain probability at every attempt. 
When two pairs are ready at one of the links, the DEJMPS entanglement distillation scheme \cite{deutsch1996quantum} is run to improve the quality of the entanglement.
If it fails, the two links are discarded and the executing nodes restart entanglement generation.
Once both distilled states are ready, the intermediate node swaps the entanglement to achieve end-to-end entanglement. 
We remark that already this simple protocol is rather involved to analyse. 

We begin by modelling the network.
The basic element of NetSquid's modular framework is the ``component''.
It is capable of describing the physical model composition, quantum and classical communication ports,
and, recursively, any subcomponents. All hardware elements, including the network itself, are represented by components. 
For this example we require three remote nodes linked by two quantum 
and two classical connections, the setup of which is shown in Figure~\ref{fig:netsquid_distil_repeat}(a).
In Figure \ref{fig:netsquid_distil_repeat}(b,c) the nested structure of these components is highlighted. 
A selection of physical models is used to describe the loss and delay of the fibre optic channels, the decoherence of the quantum memories, and the errors of quantum gates.

Quantum information in NetSquid is represented at the level of qubits,
which are treated as objects that dynamically share their quantum states.
These internally shared states will automatically merge or ``split'' -- a term we use to mean the separation of a tensor product state into two separately shared sub-states -- as qubits entangle or are measured,
as illustrated by the distillation protocol in Figure~\ref{fig:netsquid_distil_repeat}(e).
The states are tracked internally, i.e.\ hidden from users, for two reasons: to encourage a node-centric approach to programming network protocols, and to allow a seamless switching between different quantum state representations.
The representations offered by NetSquid are ket vectors, density matrices, stabiliser tableaus and graph states with local Cliffords, each with trade-offs in modelling versatility, computation speed and network (memory) scalability (see the subsection \nameref{sec:benchmarking} below and Supplementary Note~\ref{appendix:netsquid}).

Discrete-event simulation,
an established method for simulating classical network systems~\cite{wehrle2010modeling},
is a modelling paradigm that progresses time by stepping through a sequence of events -- see Figure~\ref{fig:example_des} for a visual explanation.
This allows the simulation engine to efficiently handle the control processes and feedback loops characteristic of quantum networking systems, while tracking quantum state decoherence based on the elapsed time between events.
A novel requirement for its application to quantum networks is the need to accurately evolve the state of the quantum information present in a network with time.
This can be achieved by retroactively updating quantum states when the associated qubits are accessed during an event. 
While it is possible to efficiently track 
a density matrix, quantum operations requiring a singular outcome for classical decision making, for instance a quantum measurement, must be probabilistically sampled.
A single simulation run thus consists of a sequence of random choices that forms one of many possible paths.
In Figure~\ref{fig:netsquid_distil_repeat}~(d) we show such a run for the repeater protocol example,
which demonstrates the power of the discrete-event approach for tracking qubit decoherence and handling feedback loops.

\begin{figure}[ht]
\centering
\includegraphics[width=0.49\textwidth]{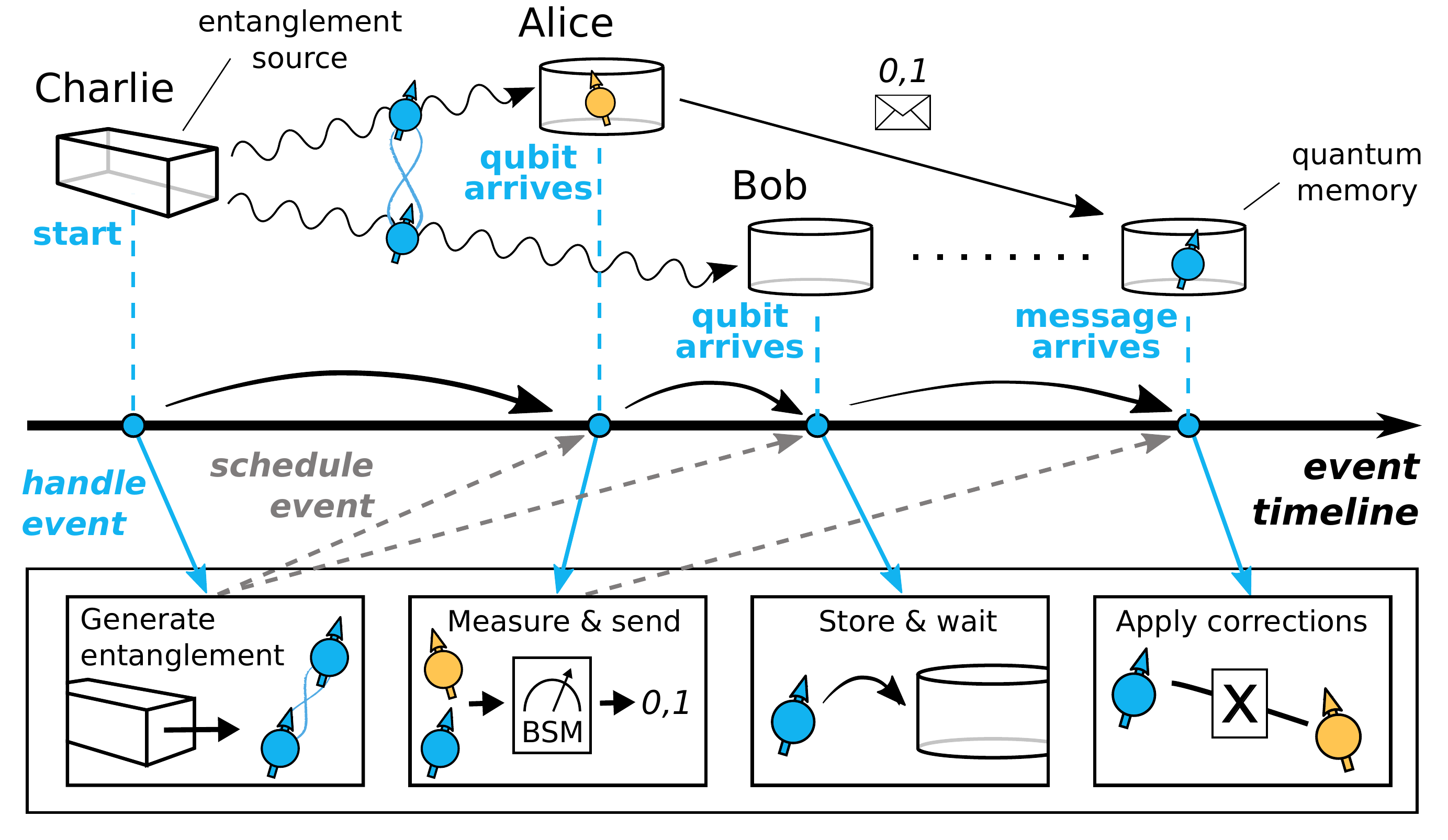}
    \caption{\captionfigexampledes}
\label{fig:example_des}
\end{figure}

The performance metrics of a simulation are determined statistically from many runs.
Due to the independence of each run, simulations can be massively parallelised and thereby efficiently executed on computing clusters.
For the example at hand we choose as metrics the output fidelity and run duration.
In Figure~\ref{fig:netsquid_distil_repeat}~(f) the sampled run from (d), which resulted in perfect fidelity, is plotted in terms of its likelihood and duration together with several other samples, some less successful.
By statistically averaging all of the sampled runs the output fidelity and duration can be estimated. 

In the following sections, we will outline three use cases of NetSquid: first, a quantum switch, followed by simulations of quantum repeaters based on nitrogen-vacancy technology or atomic-ensemble memories.
We will also benchmark NetSquid's scalability in both quantum state size and number of quantum network nodes.
Although the use cases each provide relevant insights into the performance of the studied hardware and protocols, we emphasise that one can use NetSquid to simulate arbitrary network topologies.

\subsection{Simulating a quantum network switch beyond its analytically known regime}\label{sec:results-switch}
As a first use case showcasing the power of NetSquid, we study the control plane
of a recently introduced quantum switch beyond the regime for which analytical results have been obtained, including its performance under time-dependent memory noise.

The switch is a node which is directly connected to each of $k$ users by an optical link.
The communications task is distributing Bell pairs and $n$-partite Greenberger-Horne-Zeilinger (GHZ) states~\cite{greenberger1989going} between $n\leq k$ users.
The switch achieves this by connecting Bell pairs which are generated at random intervals on each link.
See Figure~\ref{fig:quantum_switch}.

Intuitively, the switch can be regarded as a generalisation of a simple repeater performing entanglement swapping with added logic to choose which parties to link. 
Even with a streamlined physical model, it is already rather challenging to analytically characterise the switch use case \cite{vardoyan2019stochastic}. 

In the following, we recover via simulation a selection of the results from Vardoyan et al. \cite{vardoyan2019stochastic}, who studied the switch as the central node in a star network, and extend them in two directions. First, we increase the range of parameters for which we can estimate entanglement rates using the same model as used in the work of Vardoyan et al. Second, simulation enables us to investigate more sophisticated models than the exponentially distributed erasure process from their work, in particular we analyse the behaviour of a switch in the presence of memory dephasing noise.

The protocol for generating the target $n$-partite GHZ states is simple. Entanglement generation is attempted in parallel across all $k$ links. If successful they result in bipartite Bell states that are stored in quantum memories.
The switch waits until $n$ Bell pairs have been generated until performing an $n$-partite GHZ measurement, which converts the pairs into a state locally equivalent to a GHZ state.
An additional constraint is that the switch has a finite buffer $B$ of number of memories dedicated for each user (see Figure~\ref{fig:quantum_switch}).
If the number of pairs stored in a link is $B$ and a new pair is generated, the old one is dropped and the new one is stored.

The protocol can be translated to a Markov chain. The state space is represented by a $k$-length vector where each entry is associated with a link and its value denotes the number of stored links.
The switch's mean capacity, i.e.\ the number of states produced per second, can be derived from the steady-state of the Markov chain \cite{vardoyan2019stochastic}. 

Using NetSquid, it is straightforward to fully reproduce the previous model and study the behaviour of the network without constructing the Markov Chain (details can be found in Supplementary Note~\ref{app:switch}).
In Figure~\ref{fig:switch}(a), we use NetSquid to study the capacity of a switch network serving nine users.
The figure shows the capacity (number of produced GHZ-states per second), which we investigate for three use cases. First we consider a switch network distributing bipartite entanglement.
Second, we consider also a switch-network serving bipartite entanglement but with link generation rates that do not satisfy the stability condition for the Markov Chain if the buffer $B$ is infinitely large, i.e.\ a regime so far intractable.
Third, we consider a switch-network distributing four-partite entanglement where the link generation rates $\mu$ differ per user, a regime not studied so far, and compute the capacity. 

Beyond rate, it is important to understand the quality of the states produced. Answering this question with Markov chain models seems challenging. In order to analyse entanglement quality, we introduce a more sophisticated decoherence model where the memories suffer from decay over time. In particular,
we model decoherence as exponential $T_2$ noise, which impacts the quality of the state, as expressed in its fidelity with the ideal state.
In Figure~\ref{fig:switch}(b), we show the effect of the time-dependent memory noise on the average fidelity.

\begin{figure}[tb]
    \centering
    \includegraphics[width=0.5\textwidth]{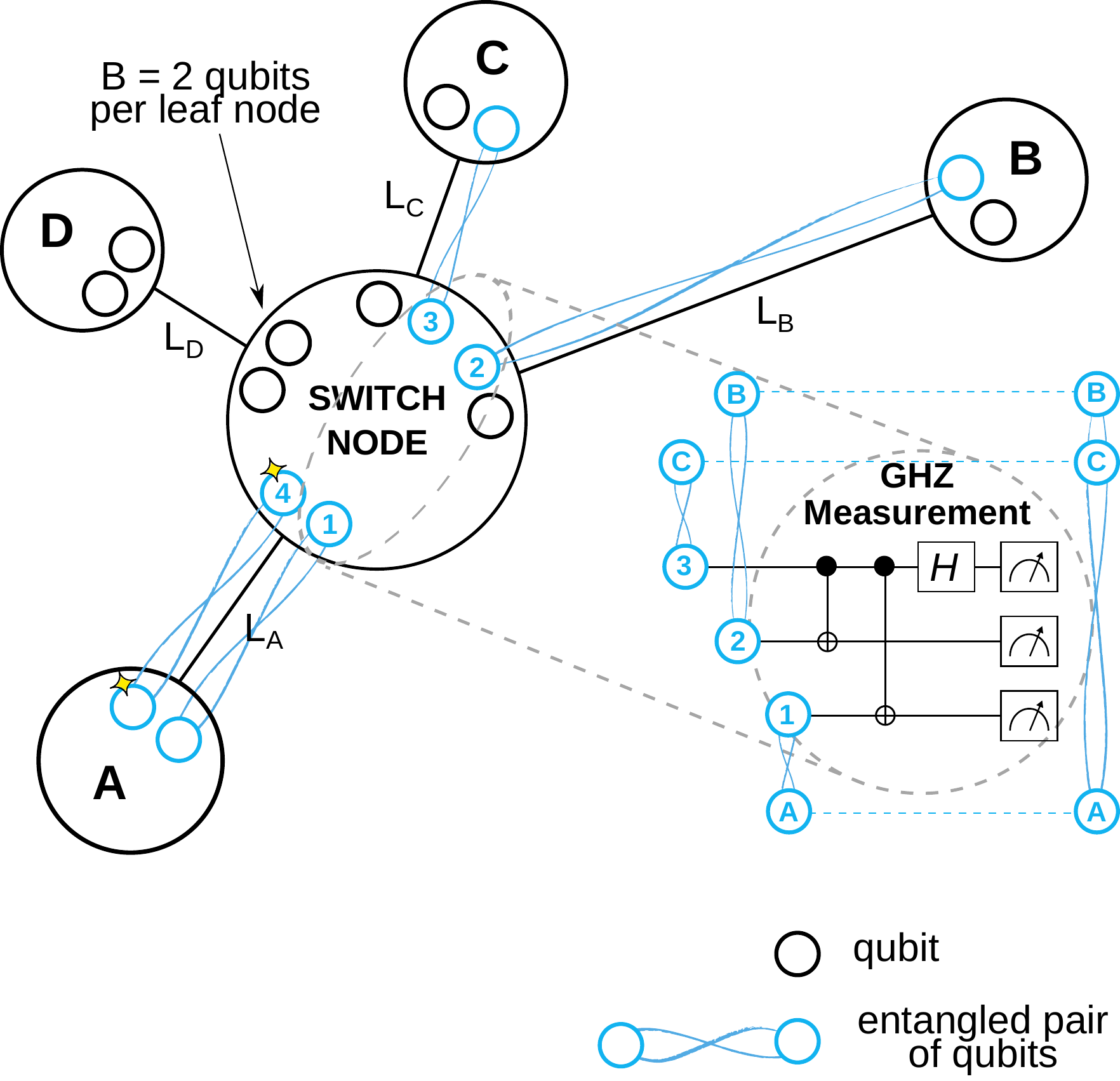}
    \caption{\captionfigswitchtopodiagram}
    \label{fig:quantum_switch}
\end{figure}

\begin{figure*}
  \centering
		\includegraphics[width=1.0\textwidth]{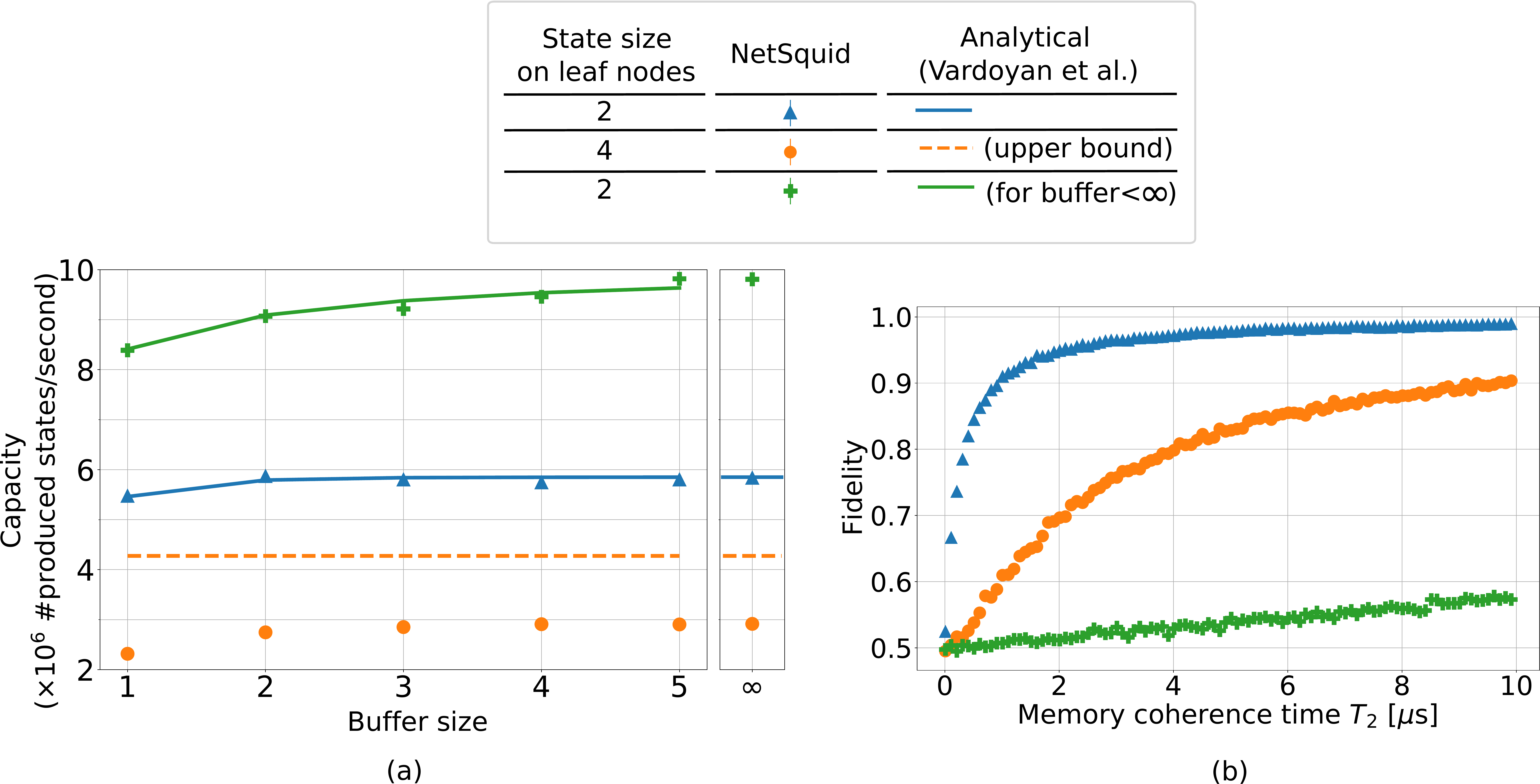}
\caption{\label{fig:switch} \captionfigswitch}
\end{figure*}

\subsection{Sensitivity analysis for the physical modelling of a long range repeater chain}
\label{sec:nv-sensitivity-analysis}
The next use case is the distribution of long-distance entanglement via a chain of quantum repeater nodes \cite{briegel1998quantum, munro2015inside} based on nitrogen-vacancy (NV) centres in diamond \cite{awschalom2018quantum, doherty2013nitrogen}.
This example consists of a more detailed physical model and more complicated control plane logic than the quantum switch or the distillation example presented at the start of this section.
It is also an example of how NetSquid's modularity supports setting up simulations involving many nodes; in this case the node model and the protocol (which runs locally at a node) only need to be specified once,
and can then be assigned to each node in the chain.
Furthermore, the use of a discrete-event engine allows the actions of the individual protocols to be simulated asynchronously, in contrast to the typically sequential execution of quantum computing simulators.

The NV-based quantum processor includes the following three features.
First, the nodes have a single communication qubit, i.e.\ a qubit acting as the optical interface that can be entangled with a remote qubit via photon interference.
This seemingly small restriction has important consequences for the communications protocol. In particular, entanglement can not proceed in parallel with both adjacent nodes. As a consequence, operations need to be scheduled in sequence and the state of the communication qubit transferred onto a storage qubit.
Second, the qubits in a node are connected with a star topology with the communication qubit located in the centre. Two-qubit gates are only possible between the communication qubit and a storage qubit.
Third, communication and storage qubits have unequal coherence times.
Furthermore, the storage qubits suffer additional decoherence when the node attempts to generate entanglement.
Previous repeater-chain analyses, e.g. \cite{nemoto2016photonic, razavi2009quantum, vanmeter2007system}, did not take all three into account simultaneously.

Together with the node model, we consider two protocols: \swapasap~and \withdistill. In \swapasap, 
as soon as adjacent links are generated the entanglement is swapped. 
\withdistill~is a nested protocol \cite{briegel1998quantum} with entanglement distillation at every nesting level.
For a description of the simulation, including the node model and protocols, see Methods, section~\nameref{sec:appprocnode}.

The first question that we investigate is the distance that can be covered by a repeater chain. For this we choose two sets of hardware parameters that we dub near-term and $10\times$ improved (see Supplementary Note~\ref{app:nv-physical-modelling}) and choose two configurations: one without intermediate repeaters and one with three of them. We observe, see Figure~\ref{fig:repchain-A}(a), that the repeater chain performs worse in fidelity than the repeaterless configuration with near-term hardware. For improved hardware, we see two regimes, for short distances the use of repeaters increases rate but lowers fidelity while from 750~km until 1500~km the repeater chain outperforms the no-repeater setup.

The second question that we address is which protocol performs best for a given distance.
We consider seven protocols: no repeater, and repeater chains implementing \swapasap~or \withdistill~over 1, 3 or 7 repeaters.
The latter is motivated by the fact that the \withdistill~protocol is defined for $2^n - 1$ repeaters ($n\geq 1$), and thus 1, 3, and 7 are the first three possible configurations.
In Figure~\ref{fig:repchain-A}(b), we sweep over the hardware parameter space for two distances, where we improve all hardware parameters simultaneously and the improvement is quantified by a number we refer to as "improvement factor" (see section \nameref{sec:improvement-factor} of the Methods).
For 500~km, we observe that the no-repeater configuration achieves larger or equal fidelity for the entire range studied.
However, repeater schemes boost the rate for all parameter values. 
If we increase the distance to 800~km, then we see that the use of repeaters increases both rate and fidelity for the same range of parameters. 
If we focus on the repeater scheme, we observe for both distances that for high hardware quality, the \withdistill~scheme, which includes distillation, is optimal.
In contrast, for lower hardware quality, the best-performing scheme that achieves fidelities larger than the classical bound $0.5$ is the \swapasap~protocol.

We note that beyond 700~km the entanglement rate decreases when the hardware is improved.
This is due to the presence of dark counts, i.e. false signals that a photon has been detected.
At large distances most photons dissipate in the fibre, whereby the majority of detector clicks are dark counts.
Because a dark count is mistakenly counted as a successful entanglement generation attempt, improving (i.e. decreasing) the dark count rate in fact results in a lower number of observed detector clicks, from which the (perceived) entanglement rate plotted in Figure~\ref{fig:repchain-A}(a) is calculated.

Lastly, in Figure~\ref{fig:sensitivity-analysis}, we investigate the sensitivity of the entanglement fidelity for the different hardware parameters. 
We take as the figure of merit the best fidelity achieved with a \swapasap~protocol.
The uniform improvement factor is set to 3, while the following four hardware parameters are varied: a two-qubit gate noise parameter, photon detection probability (excluding transmission), induced storage qubit noise and visibility.
We observe that improving the detection probability yields the largest fidelity increase from $2\times$ to $50\times$ improvement, while this increase is smallest for visibility.
We also see that improving two-qubit gate noise or induced storage qubit noise on top of an increase in detection probability yields only a small additional fidelity improvement, which however boosts fidelity beyond the classical threshold of 0.5.
These observations indicate that detection probability is the most important parameter for realising remote-entanglement generation with the \swapasap~scheme, followed by two-qubit gate noise and induced storage qubit noise. 

\begin{figure*}
  \centering
    \includegraphics[width=1.0\textwidth]{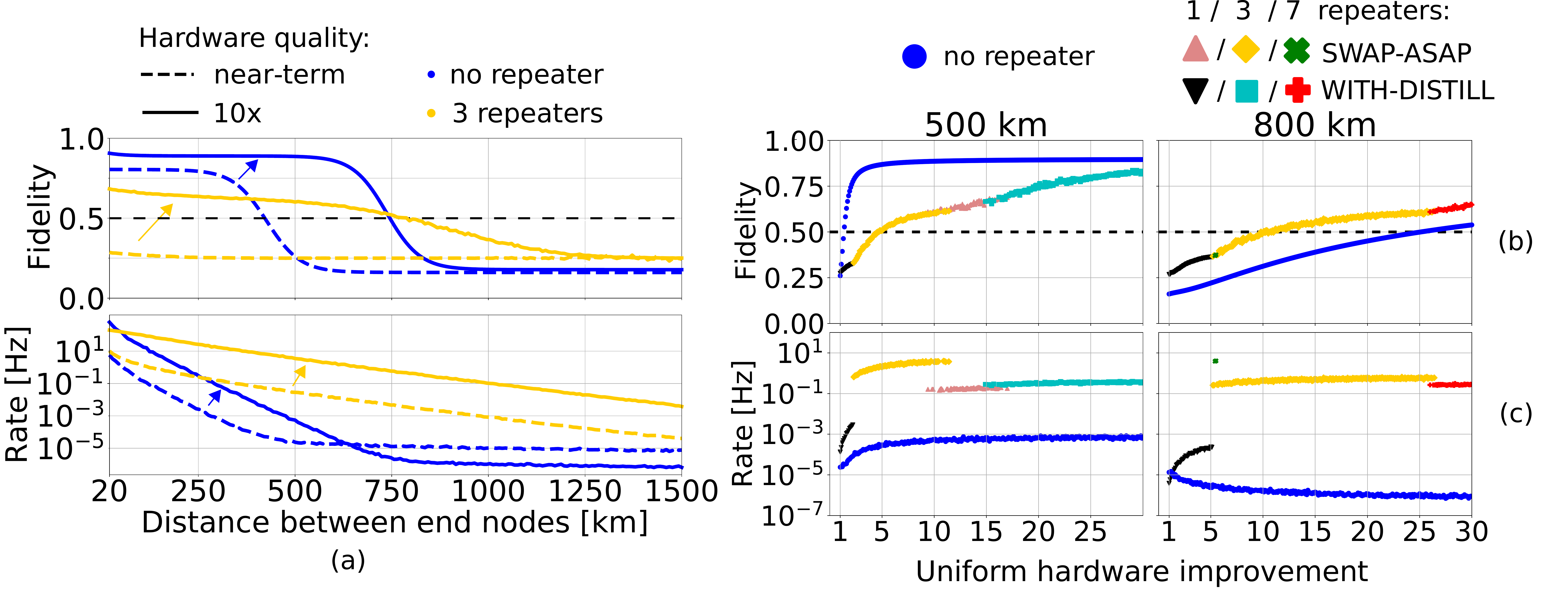}
\caption{\label{fig:repchain-A} \captionfigrepchainA}
\end{figure*}

\begin{figure}
    \centering
	\includegraphics[width=0.5\textwidth]{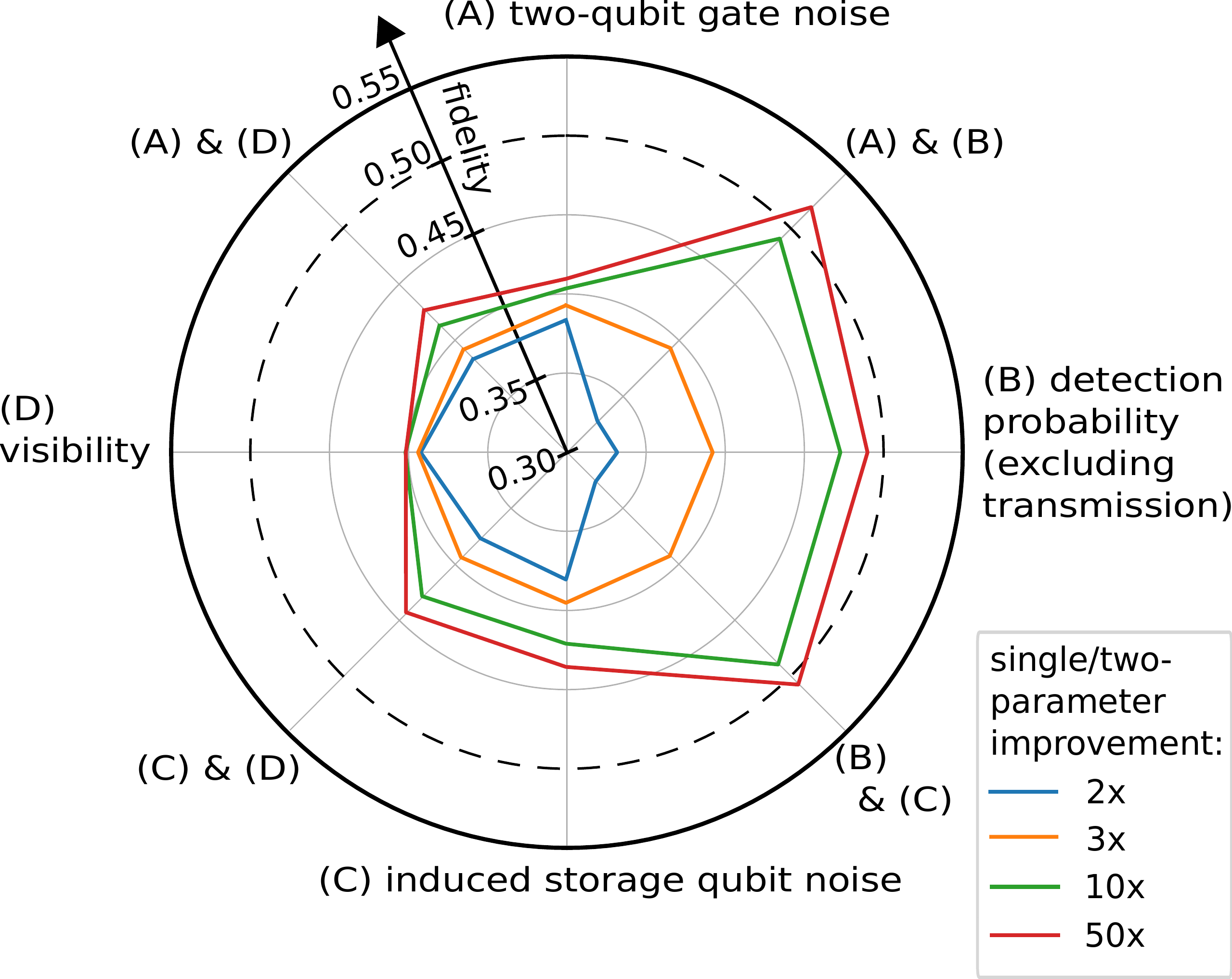}
\caption{\label{fig:sensitivity-analysis} \captionfigsensitivityanalysis}
\end{figure}

\subsection{Performance comparison between two atomic-ensemble memory types through NetSquid's modular design}\label{sec:ae_memory_comp}
Finally, we showcase that NetSquid's modular design greatly reduces the effort of assessing possible hardware development scenarios.
We demonstrate the power of this modularity by simulating point-to-point remote-entanglement generation based on either of two types of atomic-ensemble based quantum memories: atomic frequency combs (AFC) \cite{afzelius2009multimode} and electronically induced transparency (EIT) \cite{fleischhauer2005electromagnetically,lukin2003colloquium} memories. Both simulations are identical except for the choice of a different quantum memory component.

The two types of memories are a promising building block for high-rate remote entanglement generation through quantum repeaters because of their high efficiency (EIT) or their ability for multiplexing (AFC), i.e.\ to perform many attempts at entanglement generation simultaneously without network components having to wait for the arrival of classical messages that herald successful generation.
The first type of memories, AFCs, are based on a photon-echo process, where an absorbed photon is re-emitted after an engineered duration.
In contrast, the second type, EITs, emit the photon after an on-demand interval, due to optical control.
In principle the AFC protocol can be extended to offer such on-demand retrieval as well.
At this point both technologies are promising candidates and it is not yet clear which outperforms the other and under what circumstances.

Atomic-ensemble based repeaters have been analytically and numerically studied before with streamlined physical models \cite{guha2015rate, krovi2016}.
NetSquid's discrete-event based paradigm allows us to go beyond that by concurrently introducing several non-ideal characteristics.
In particular, we include the emission of more than one photon pair, photon distinguishability and time-dependent memory efficiency. Efficiency in this context is the probability that the absorbed photon will be re-emitted.
All these characteristics have a significant impact on the performance of the repeater protocol. 

In order to compare the two memory types, we simulate many rounds of the BB84 quantum key distribution protocol \cite{bennett2014} between two remote nodes, using a single repeater positioned precisely in between them.
Entanglement generation is attempted in synchronised rounds over both segments in parallel.
At the end of each round, the two end nodes measure in the X- or Z-basis, chosen uniformly at random, and the repeater performs a probabilistic linear-optical Bell-state measurement.
Upon a successful outcome, we expect correlation between the measurement outcomes if they were performed in the same basis.
As a figure of merit we choose the asymptotic BB84 secret-key rate.

The results of our simulations are shown in Figure~\ref{fig:afc-memory-comparison}, where the rate at which secret key between the two nodes can be generated is obtained as a function of the distance between the nodes.
For the parameters considered (see Supplementary Note~\ref{app:atomic-ensembles}), we observe that EIT memories outperform AFC memories at short distances. The crossover performance point is reached at $\sim 50$ kilometers, beyond which AFC memories outperform EIT memories.

In the use case above, we showcased NetSquid's modularity by only replacing the memory component.
We emphasise that this modularity also applies to different parts of the simulation.
For example, if the quantum switch should produce a different type of multipartite state than GHZ states, then one only needs to change the circuit at the switch node.
A different example is the NV repeater chain, where one could replace the protocol module (currently either {\sc swap-asap} or {\sc nested-with-distill}).

\begin{figure*}[htb]
\centering
\includegraphics[width=\textwidth]{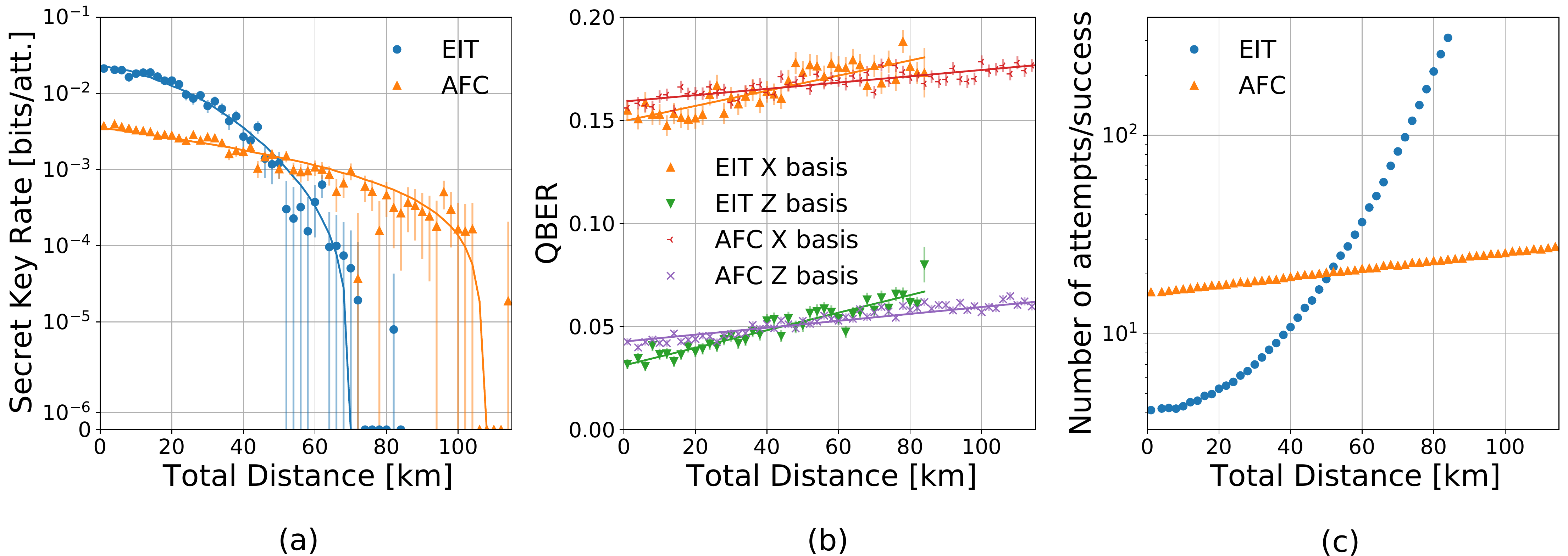}
\caption{\captionfigafcmemorycomparison \label{fig:afc-memory-comparison}}
\end{figure*}

\subsection{Fast and scalable quantum network simulation}
\label{sec:benchmarking}
NetSquid has been designed and optimised to meet several key performance criteria:
to be capable of accurate physical modelling, to be scalable to large networks,
and to be sufficiently fast to support multi-variate design analyses with adequate statistics.
While it is not always possible to jointly satisfy all the criteria for all use cases, NetSquid's design allows the user to prioritise them. 
We proceed to benchmark NetSquid to demonstrate its capabilities and unique strengths for quantum network simulation.

\subsubsection{Benchmarking of quantum computation}
\label{sec:benchmarking-qcomp}
To accurately model physical non-idealities, it is necessary to choose a representation for quantum states that allows a characterisation of general processes such as amplitude damping, general measurements, or arbitrary rotations.
NetSquid provides two representations, or ``formalisms'', that are capable of universal quantum computation: ket state vectors (KET) and density matrices (DM), both stored using dense arrays.
The resource requirements for storage in memory and the computation time associated with applying quantum operations both scale exponentially with the number of qubits.  
While the density matrix scales less favourably, $2^{2n}$ versus $2^{n}$ for $n$ qubits, its ability to represent mixed states makes it more versatile for specific applications.
Given the exponential scaling, these formalisms are most suitable for simulations in which a typical qubit lifetime involves only a limited number of (entangling) interactions.

When scaling to large network simulations it can happen that hundreds of qubits share the same entangled quantum state.
For such use cases, we need a quantum state representation that scales sub-exponentially in time and space.
NetSquid provides two such representations based on the stabiliser state formalism: ``stabiliser tableaus'' (STAB) and ``graph states with local Cliffords'' (GSLC)~\cite{aaronson2004improved,anders2006fast} that the user can select.
Stabiliser states are a subset of quantum states that are closed under the application of Clifford unitaries and single-qubit measurement in the computational basis.
In the context of simulations for quantum networks stabiliser states are particularly interesting because many network protocols consist of only Clifford operations and noise can be well approximated by stochastic application of Pauli gates.
For a theoretical comparison of the STAB and GSLC formalisms see Supplementary Note~\ref{appendix:netsquid}.

The repetitive nature of simulation runs due to the collection of statistics via random sampling allows NetSquid to take advantage of ``memoization'' for expensive quantum operations,
which is a form of caching that stores the outcome of expensive operations and returns them when the same input combinations reoccur to save computation time.
Specifically, the action of a quantum operator onto a quantum state for a specific set of qubit indices and other discrete parameters can be efficiently stored, for instance as a sparse matrix. 
Future matching operator actions can then be reduced to a fast lookup and application, avoiding several expensive computational steps -- see the Methods, section \nameref{sec:methods-qubits} for more details.

In the following we benchmark the performance of the available quantum state formalisms. For this, we first consider the generation of an $n$ qubit entangled GHZ state followed by a measurement of each qubit (see section \nameref{sec:methods-benchmarking} of the Methods).
For a baseline comparison with classical quantum computing simulators we also include the ProjectQ~\cite{steiger2018projectq} package for Python, which uses a quantum state representation equivalent to our ket vector.
We show the average computation time for a single run versus the number of qubits for the different quantum computation libraries in Figure~\ref{fig:qubit-benchmark}(a).
The exponential scaling of the universal formalisms in contrast to the stabiliser formalisms is clearly visible, with the density matrix formalism performing noticeably worse.
For the ket formalism we also show the effect of memoization, which gives a speed-up roughly between two and five.

\begin{figure*}[tb]
    \centering
        \includegraphics[width=1.0\textwidth]{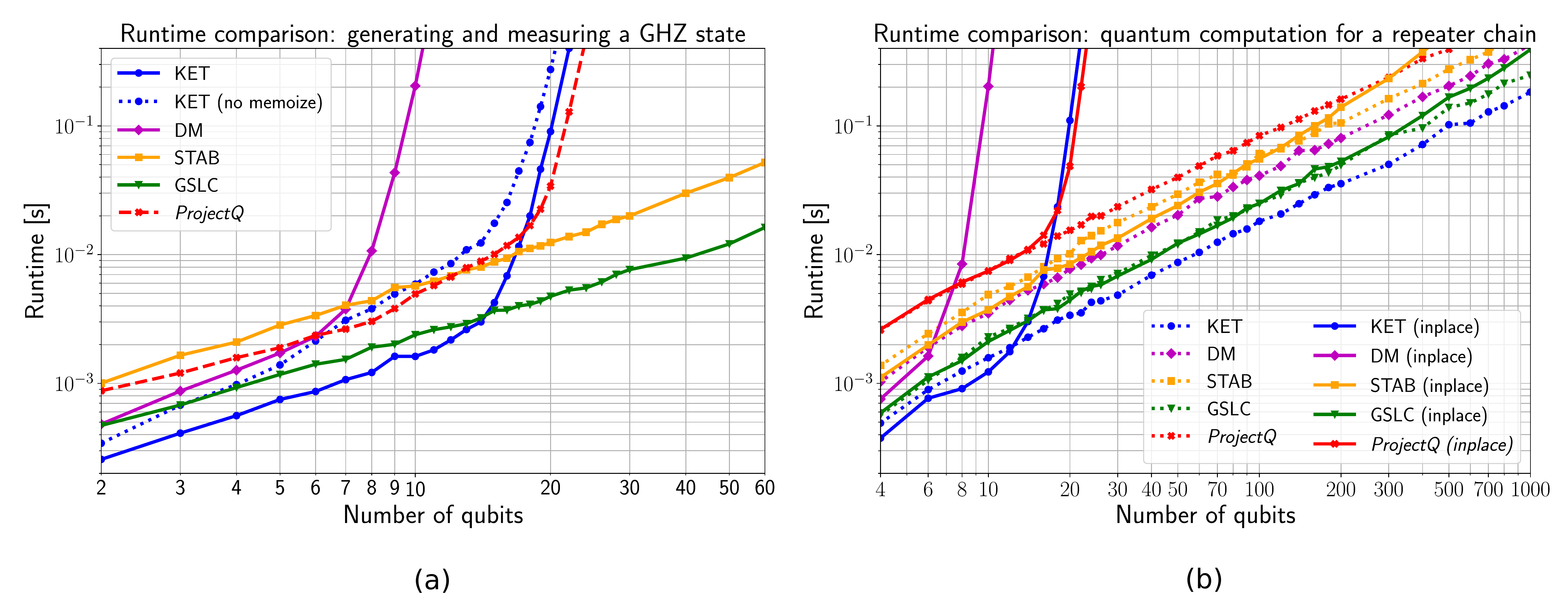}
    \caption{\captionqubitbenchmark \label{fig:qubit-benchmark}}
\end{figure*}

Let us next consider a more involved benchmarking use case: the quantum computation involved in simulating a repeater chain i.e.\ only the manipulation of qubits, postponing all other simulation aspects, such as event processing and component modelling, to the next section.
This benchmark involves the following steps: first the $N-1$ pairs of qubits along an $N$ node repeater chain are entangled, then each qubit experiences depolarising noise, and finally adjacent qubits on all but the end-nodes do an entanglement swap via a Bell state measurement (BSM).
If the measured qubits are split from their shared quantum states after the BSM, then the size of any state is limited to four qubits.

The average computation time for a single run versus the number of qubits in the chain are shown for the different quantum computation libraries in Figure~\ref{fig:qubit-benchmark}(b), where we have again included ProjectQ.
We observe that for the NetSquid formalisms (but not for ProjectQ) keeping qubits ``in-place'' after each measurement is more performant than ``splitting'' them below a certain threshold due to the extra overhead of doing the latter.
The ket vector formalism is seen to be the most efficient for this benchmarking use case if states are split after measurement.
When the measurement operations are performed in-place the GSLC formalism performs the best beyond 15 qubits.

\subsubsection{Benchmarking of event-driven simulations}
\label{sec:benchmarking-profiling}
As explained in the results section, a typical NetSquid simulation involves repeatedly sampling many independent runs.
As such NetSquid is ``embarrassingly parallelisable'': the reduction in runtime scales linearly with the number of processing cores available, assuming there is sufficient memory available.
Nonetheless, given the computational requirements associated with collecting sufficient statistics and analysing large parameter spaces it remains crucial to optimise the runtime performance per core.

Depending on the size of the network, the detail of the physical modelling, and the duration of the protocols under consideration, the number of events processed for a single simulation run can range anywhere from a few thousand to millions.
To efficiently process the dynamic scheduling and handling of events NetSquid uses the discrete-event simulation engine PyDynAA~\cite{de2013model} (see section \nameref{sec:methods-des} of the Methods).
NetSquid aims to schedule events as economically as possible, for instance by streamlining the flow of signals and messages between components using inter-connecting ports. 

To benchmark the performance of an event-driven simulation run in NetSquid we consider a simple network that extends the single repeater (without distillation) shown in Figure~\ref{fig:netsquid_distil_repeat} into an $N$ node chain -- see Supplementary Note~\ref{appendix:benchmarking} for further details on the simulation setup.
For the quantum computation we will use the ket vector formalism based on the benchmarking results from the previous section,
and split qubits from their quantum states after measurement to avoid an exponential scaling with the number of nodes.
In Figure~\ref{fig:repchain_benchmark} we show the average computation time for deterministically generating end-to-end entanglement versus the number of nodes in the chain.
Also shown is a relative breakdown in terms of the time spent in the NetSquid sub-packages involved, as well as the PyDynAA and NumPy packages.
We observe that the biggest contribution to the simulation runtime is the components sub-package, which accounts for 30\% of the total at 1000 nodes.
The relative time spent in each of the NetSquid sub-packages, as well as NumPy and PyDynAA, is seen to remain constant with the number of nodes.
The total runtime of each of the NetSquid sub-packages is the sum of many small contributions, with the costliest function for the components sub-package for a 1000 node chain, for example, contributing only 7\% to the total.

Extending this benchmark simulation with more detailed physical modelling may shift the relative runtime distribution and impact the overall performance.
For example, more time may be spent in calls to the ``components'' and ``components.models'' sub-packages, additional complexity can increase the volume of events processed by the ``pydynaa'' engine,
and extra quantum characteristics can lead to larger quantum states.
In case of the latter, however, the effective splitting of quantum states can still allow such networks to scale if independence among physical elements can be preserved.

\begin{figure}[htb]
\centering
\includegraphics[width=0.5\textwidth]{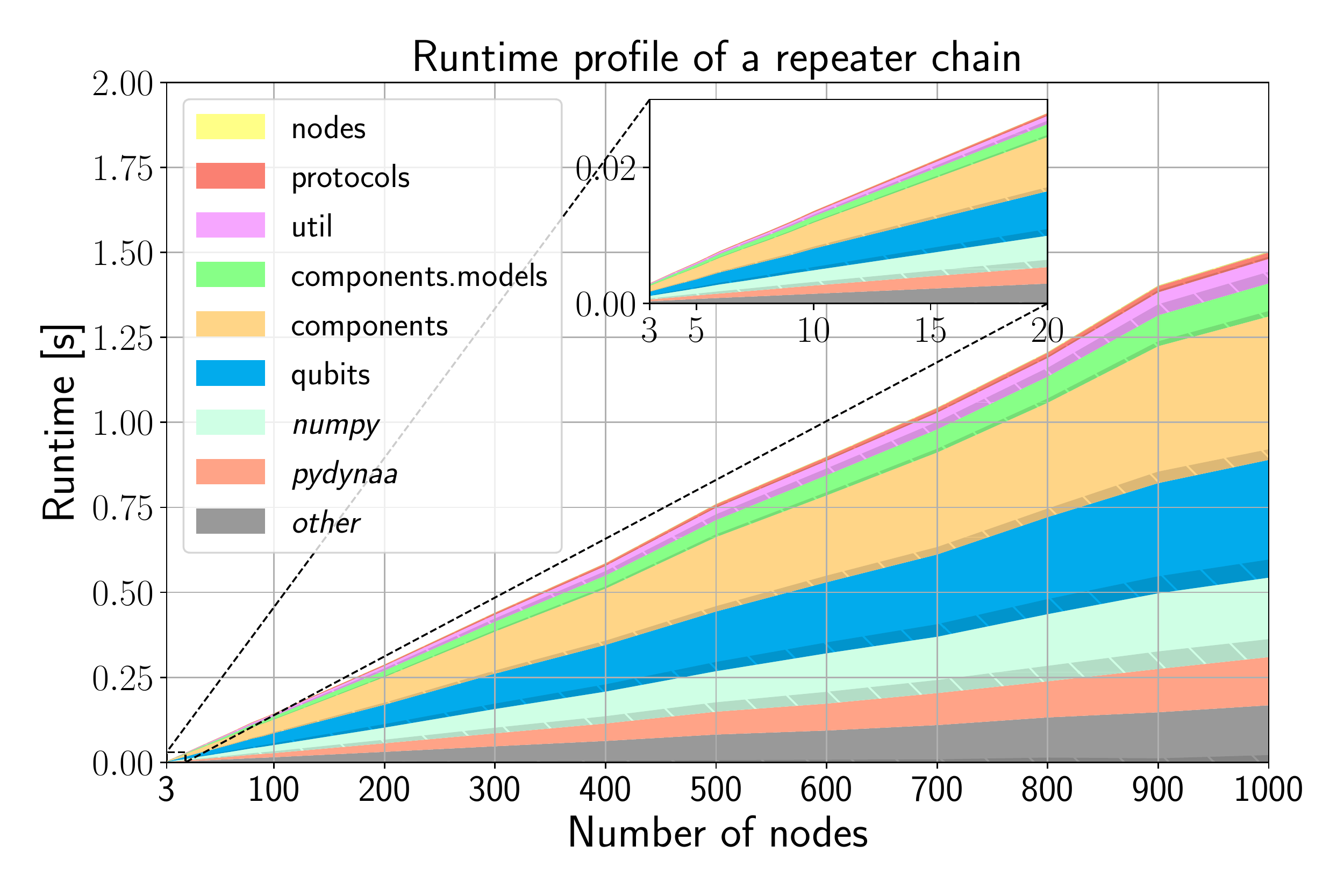}
    \caption{\captionrepchainbenchmark\label{fig:repchain_benchmark}}
\end{figure}

\subsection{Comparison with other quantum network simulators}
Let us compare NetSquid to other existing quantum network simulators. 
First, SimulaQron \cite{dahlberg2018simulaqron} and QuNetSim \cite{diadamo2020qunetsim} are two simulators that do not aim at realistic physical models of channels and devices, or timing control. 
Instead, SimulaQron's main purpose is application development. It is meant to be run in a distributed fashion on physically-distinct classical computers.
QuNetSim focuses on simplifying the development and implementation of quantum network protocols.

In contrast with SimulaQron and QuNetSim, the simulator SQUANCH \cite{bartlett2018distributed} allows for quantum network simulation with configurable error models at the physical layer.
However, SQUANCH, similar to SimulaQron and QuNetSim, does not use a simulation engine that can accurately track time. Accurate tracking is crucial for e.g. studying time-dependent noise such as memory decoherence.

Other than NetSquid, there now exist three discrete-event quantum simulators: the QuISP \cite{matsuo2019simulation}, qkdX \cite{mailloux2014modeling} and SeQUeNCe \cite{wu2020sequence} simulators. With these simulators it is possible to accurately characterise complex timing behaviour, however they differ in goals and scope. 
Similarly to NetSquid, QuISP aims to support the investigation of large networks that consist of too many entangled qubits for full quantum-state tracking.
In contrast to NetSquid, which achieves this by managing the size of the state space, and providing the stabiliser representation as one of its quantum state formalisms, QuISP's approach is to track an error model of the qubits in a network instead of their quantum state.
qkdX, on the other hand, captures the physics more closely through models of the quantum devices but is restricted to the simulation of quantum key distribution protocols. 
Lastly, SeQUeNCe, similar to NetSquid, aims at simulation at the level of hardware, control plane or application. 
It has a fixed control layer consisting of reprogrammable modules.
In contrast, NetSquid's modularity is not tied to a particular network stack design.
Furthermore, it is unclear to us how performant SeQUeNCe's quantum simulation engine is: currently, at most a 9-node network has been simulated, whereas NetSquid's flexibility to choose a quantum state representation enables scalability to simulation of networks of up to 1000 nodes.

\subsection{Conclusions}

In this work we have presented our design of a modular software framework for simulating scalable quantum networks and accurately modelling the non-idealities of real world physical hardware, providing us with a design tool for future quantum networks.
We have showcased its power and also its limitations via example use cases.
Let us recap NetSquid's main features.

First, NetSquid allows the modelling of any physical device in the network that can be mapped to qubits.
To demonstrate this we studied two use cases involving nitrogen-vacancy centres in diamond as well as atomic-ensemble based memories. 

Second, NetSquid is entirely modular,
allowing users to set up large scale simulations of complicated networks and to explore variations in the network design; for example, by comparing how different hardware platforms perform in an otherwise identical network layout. 
Moreover, this modularity makes it possible to explore different control plane protocols for quantum networks in a way that is essentially identical to how such protocols would be executed in the real world.
Control programs can be run on any simulated network node, exchanging classical and quantum communication with other nodes as dictated by the protocol.
That allows users to investigate the intricate interplay between control plane protocols and the physical devices dictating the performance of the combined quantum network system.
As an example, we studied the control plane of a quantum network switch.
NetSquid has also already found use in exploring the interplay between the control plane and the physical layer in~\cite{dahlberg2019linklayer,lee2020quantum, kozlowski2020designing}.

Finally, to allow large scale simulations, the quantum computation library used by NetSquid has been designed to manage the dynamic lifetimes of many qubits across a network.
It offers a seamless choice of quantum state representations to support different modelling use cases, allowing both a fully detailed simulation in terms of wave functions or density matrices, or simplified ones using certain stabiliser formalisms.
As an example use case, we explored the simulation run-time of a repeater chain with up to one thousand nodes. 

In light of the results we have presented, we see a clear application for NetSquid in the broad context of communication networks.
It can be used to predict performance with accurate models, to study the stability of large networks, to validate protocol designs, to guide experiment, etc.
While we have only touched upon it in our discussion of performance benchmarks, NetSquid would also lend itself well to the study of modular quantum computing architectures, where the timing of control plays a crucial role in studying their scalability.
For instance, it might be used to validate the microarchitecture of distributed quantum computers or more generally to simulate different components in modular architectures.

\section{Methods}\label{sec:methods}

\subsection{Design and functionality of NetSquid}
\label{sec:methods-netsquid}
The NetSquid simulator is available as a software package for the Python 3 programming language.
It consists of the sub-packages ``qubits'', ``components'', ``models'', ``nodes'', ``protocols'' and ``util'', which are shown stacked in Figure~\ref{fig:netsquid-arch}.
NetSquid depends on the PyDynAA software library to provide its discrete-event simulation engine~\cite{de2013model}.
Under the hood speed critical routines and classes are written in Cython~\cite{behnel2011cython} to give C-like performance,
including its interfaces to both PyDynAA and the scientific computation packages NumPy and SciPy.
In the following subsections we highlight some of the main design features and functionality of NetSquid;
for a more detailed presentation see Supplementary Note~\ref{appendix:netsquid}.

\begin{figure}[htb]
    \centering
    \includegraphics[width=0.5\textwidth]{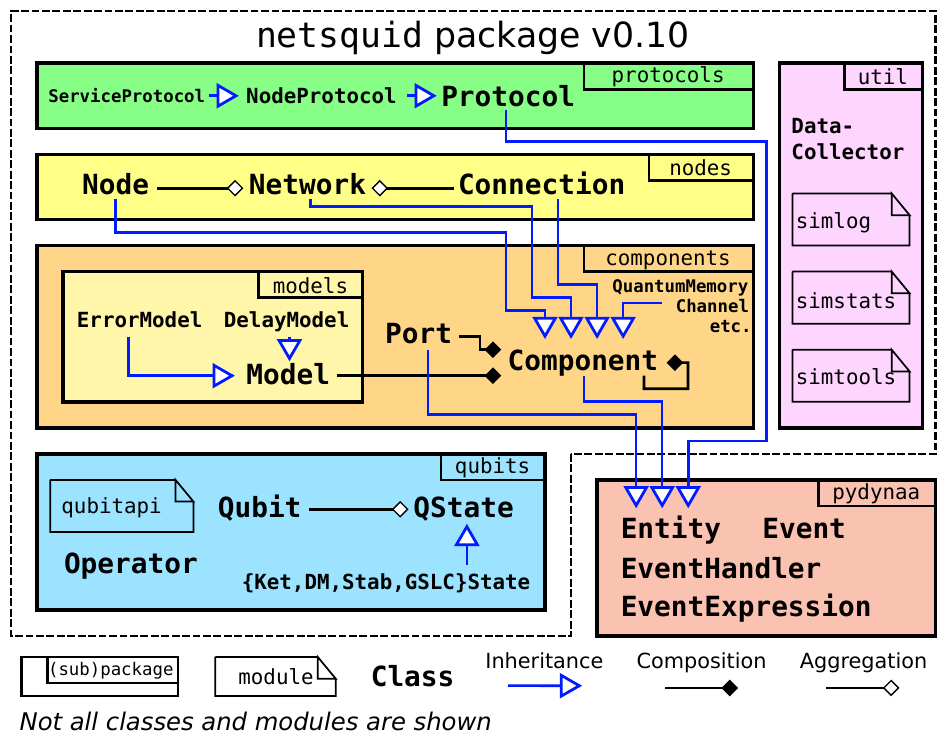}
    \caption{\captionfignetsquidarch}
    \label{fig:netsquid-arch}
\end{figure}

\subsubsection{Discrete event simulation}  
\label{sec:methods-des}
The PyDynAA package provides a fast, powerful, and lightweight discrete-event simulation engine.
It is a C\texttt{++} port of the core engine layer from the DynAA simulation framework~\cite{de2013model}, 
with bindings added for the Python and Cython languages.
DynAA defines a concise set of classes and concepts for modelling event-driven simulations.
The simulation engine manages a timeline of ``events'', which can only be manipulated by objects that are sub-classes of the ``entity'' base class.
Simulation entities can dynamically schedule events on the timeline and react to events by registering an ``event handler'' object to wait for event(s) with a specified type, source entity, or identifier to be triggered.

To deal with the timing complexities encountered in NetSquid simulations, an ``event expression'' class was introduced to PyDynAA to allow entities to also wait on logical combinations of events to occur.
Atomic event expressions, which describe regular wait conditions for standard events, can be combined to form composite expressions using logical ``and'' and ``or'' operators to any depth.
This feature has been used extensively in NetSquid to model both the internal behaviour of hardware components, as well as for programming network protocols.

\subsubsection{Qubits and quantum computation} 
\label{sec:methods-qubits}
The qubits sub-package of NetSquid defines the ``qubit'' object that is used to track the flow of quantum information.
Qubits internally share quantum state (``QState'') objects, which grow and shrink in size as qubits interact or are measured.
The ``QState'' class is an interface that is implemented by a range of different formalisms, as presented in section \nameref{sec:benchmarking-qcomp} of the Results and Discussion.
Via the qubit-centric API, which provides functions to directly manipulate qubits without knowledge of their shared quantum states, users can program simulations in a formalism agnostic way.
Functionality is also provided to automatically convert between quantum states that use different formalisms,
and to sample from a distribution of states, which is useful for instance for pure state formalisms.

The ket and density matrix formalisms use dense arrays (vectors or matrices, respectively) to represent quantum states.
Applying a $k$ qubit operator to an $n$ qubit ket vector state generally involves the computationally expensive task of performing $2^{n-k}$ matrix multiplications on $2^k$ temporary sub-vectors and aggregating the result (only in special cases can this be done in-place)~\cite{DERAEDT2007121,haner2017projectq}.
The analogous application of an operator to a density matrix is more expensive due to the extra dimension involved.
However, as discussed in section \nameref{sec:benchmarking} of the Results and Discussion, the repetitive nature of NetSquid simulations allows us to take advantage of operators frequently being applied to the same qubit indices for states of a given size.
For these operators, we compute a $2^n\times 2^n$ dimensional sparse matrix representation of the $k$ qubit operator via tensor products with the identity and memoize this result for the specific indices and size.
When the memoization is applicable the computational cost of applying a quantum operator can then be reduced to just sparse matrix multiplication onto a dense vector or matrix.
Memoization is similarly applicable to general Clifford operators in the stabiliser tableau formalism.
To use memoization on operators that depend on a continuous parameter, such as arbitrary rotations,
the parameter can be discretised i.e. rounded to some limited precision.

\subsubsection{Physical modelling of network components} 
All physical devices in a quantum network are modelled by a ``component'' object, and are thereby also all simulation entities, as shown in Figure~\ref{fig:netsquid-arch}.
Components can be composed of subcomponents, which makes setting up networks in NetSquid modular.
The network itself, for instance, can be modelled as a composite component containing ``node'' and ``connection'' components;
these composite components can in turn contain components such as quantum memories, quantum and classical channels, quantum sources, etc., as illustrated in Figure~\ref{fig:netsquid_distil_repeat}.
The physical behaviour of a component is described by composing it of ``models'', which can specify physical characteristics such as transmission delays or noise such as photon loss or decoherence.
Communication between components is facilitated by their ``ports'', which can be connected together to automatically pass on messages.

NetSquid also allows precise modelling of quantum computation capable devices. For this it provides the ``quantum processor'' component, a subclass of the quantum memory.
This component is capable of executing ``quantum programs'' i.e.\ sequences of ``instructions'' that describe operations such as quantum gates and measurements or physical processes such as photon emission.
Quantum programs fully support conditional and iterative statements, as well as parallelisation if the modelled device supports it.
When a program is executed its instructions are mapped to the physical instructions on the processor, which model the physical duration and errors associated to carrying out the operation.
A physical instruction can be assigned to all memory positions or only to a specific position, as well as directionally between specific memory positions in the case of multi-qubit instructions.

\subsubsection{Asynchronous framework for programming protocols} 
NetSquid provides a ``protocol'' class to describe the network protocols and classical control plane logic running on a quantum network.
Similarly to the component class, a protocol is a simulation entity and can thereby directly interact with the event timeline.
Protocols can be nested inside other protocols and may describe both local or remote behaviour across a network.
The ``node protocol'' subclass is specifically restricted to only operating locally on a single node.
Inter-protocol communication is possible via a signalling mechanism and a request and response interface defined by the ``service protocol'' class.
Protocols can be programmed using both the standard callback functionality of PyDynAA and a tailored asynchronous framework that allows the suspension of a routine conditioned on an ``event expression'';
for example, to wait for input to arrive on a port, a quantum program to finish, or to pause for a fixed duration.

The ``util'' sub-package shown in Figure~\ref{fig:netsquid-arch} provides a range of utilities for running, recording and interacting with simulations.
Functions to control the simulation are defined in the ``simtools'' module, including functions for inspecting and diagnosing the timeline. 
A ``data collector'' class supports the event-driven collection of data during a simulation,
which has priority over other event handlers to react to events.
The ``simstats'' module is responsible for collecting a range of statistics during a simulation run, such as the number of events and callbacks processed, the maximum and average size of manipulated quantum states, and a count of all the quantum operations performed.
Finally, the ``simlog'' module allows fine grained logging of the various modules for debugging purposes.

\subsubsection{Benchmarking} 
\label{sec:methods-benchmarking}
To perform the benchmarking described in section \nameref{sec:benchmarking} of the Results and Discussion we used computing nodes with two 2.6 GHz Intel Xeon E5-2690 v3 (Haswell) 12 core processors and 64 GB of memory.
Because each process only requires a single core, care was taken to ensure sufficient cores and memory were available when running jobs in parallel.
The computation time of a process is the arithmetic average of a number of successive iterations; to avoid fluctuations due to interfering CPU processes the reported time is a minimum of five such repeated averages.
To perform the simulation profiling
the Cython extension modules of both NetSquid and PyDynAA were compiled with profiling on,
which adds some runtime overhead.
Version 0.10.0 and 0.3.5 of NetSquid and PyDynAA were benchmarked.
We benchmarked against ProjectQ version 0.4.2 using its ``MainEngine'' backend.
See Supplementary Note~\ref{appendix:benchmarking} for further details.

Using the same machine, simulations for Figure~\ref{fig:repchain-A}(b-c) were run, which took almost 260 core hours wallclock time in total, while simulations for Figure~\ref{fig:afc-memory-comparison} took roughly 625 core hours.
For Figure~\ref{fig:switch} ($\approx$10 hours in total), Figure~\ref{fig:repchain-A}(a) ($\approx$90 minutes) and Figure~\ref{fig:sensitivity-analysis} ($\approx$30 minutes), a single core Intel Xeon Gold 6230 processor (3.9GHz) with 192 GB RAM was used.

\subsection{Implementing a processing-node repeater chain in NetSquid}
\label{sec:appprocnode}
Here, we explain the details of the most complex of our three use cases, namely the repeater chain of Nitrogen-Vacancy-based processing nodes
from section~\nameref{sec:nv-sensitivity-analysis} of the Results and Discussion (see Supplementary Notes~\ref{app:switch} and \ref{app:atomic-ensembles} for details on the other two use cases).
We first describe how we modelled the NV hardware, followed by the repeater protocols used.
With regard to the physical modelling, let us emphasise that this is well established (see e.g. \cite{rozpedek2018near-term}); the main goal here is to explain how we used this model in a NetSquid implementation.

In our simulations the following NetSquid components model the physical repeater chain:
``nodes'', each holding a single ``quantum processor'' modelling the NV centre, and ``classical channels'' that connect adjacent nodes and are modelled as fibres with a constant transmission time.
We choose equal spacing between the nodes.
If we were to simulate individual attempts at entanglement generation, we would also need components for transmitting and detecting qubits such as was used in previous NetSquid simulations of NV centres~\cite{dahlberg2019linklayer}.
However, in order to speed up simulations we insert the entangled state between remote NVs using a model.
We designed two types of protocols to run on each node of this network
that differ in whether they implement a scheme with or without distillation.

In the remainder of this section, we describe the components modelling.
More detailed descriptions of the hardware parameters and their values used in our simulation can be found in Supplementary Note~\ref{app:nv-physical-modelling}.

\subsubsection{Modelling a nitrogen-vacancy centre in diamond}
\label{sec:nv-qprocessor}
In NetSquid, the NV centre is modelled by a quantum processor component, which holds a single communication qubit (electronic spin-1 system) and multiple storage qubits (${}^{13}$C nuclear spins).
The decay of the state held by a communication qubit or storage qubit is implemented using a noise model, which is based on the relaxation time $T_1$ and the dephasing time $T_2$.
If a spin is acted upon after having been idle for time $\Delta t$, then to its state $\rho$ we first apply a quantum channel
\[
	\rho \mapsto E_0 \rho E_0^{\dagger} + E_1 \rho E_1^{\dagger}
	\]
where
\[
	E_0 = \dyad{0} + \sqrt{1 - p} \dyad{1}, E_1 = \sqrt{p} \dyad{0}{1}
	\]
and $p=1 - e^{-\Delta t / T_1}$.
Subsequently, we apply a dephasing channel
\def\Ndeph{\mathcal{N}^{\text{deph}}}
\begin{equation}
\label{eq:dephasing-channel}
\Ndeph_p : \rho \mapsto (1 - p) \rho + p Z\rho Z
\end{equation}
where $Z = \dyad{0} - \dyad{1}$ and the dephasing probability equals
\[
p=\frac{1}{2} \left(1 - e^{-\Delta t / T_2}\cdot e^{\Delta t/(2T_1)}\right)
.
\]
The electron and nuclear spins have different $T_1$ and $T_2$ times.

We allow the quantum processor to perform the following operations on the electron spin:
initialisation (setting the state to $\ket{0}$), readout (measurement in the $\{\ket{0}, \ket{1}\}$ basis) and arbitrary single-qubit rotation. 
In particular, the latter includes Pauli rotations
	\begin{equation}
		R_P(\theta) = \cos(\theta / 2) \unit_2 - i \sin(\theta / 2) P
	\end{equation}
		where $\theta$ is the rotation angle, $P\in \{X, Y, Z\}$ and
    \mbox{$\unit_2 = \dyad{0} + \dyad{1}$},
    \mbox{$X = |0\rangle\langle 1| + |1\rangle\langle 0|$},
    \mbox{$Y = -i|0\rangle\langle 1| + i |1\rangle\langle 0|$} and
    \mbox{$Z = \dyad{0} - \dyad{1}$}
	are the single-qubit Pauli operators.	

For the nuclear spin, we have only initialisation and rotations $R_Z(\theta)$ for arbitrary rotation angle $\theta$.
In addition, we allow the two-qubit controlled-$R_X(\pm\theta)$ gate between an electron ($e$) and a nuclear ($n$) spin:
\[
	\dyad{0}_e \otimes R_X(\theta)_n + \dyad{1}_e \otimes R_X(-\theta)_n
	.
\]

We model each noisy operation $O_{\text{noisy}}$ as the perfect operation $O_{\text{perfect}}$ followed by a noise channel $\mathcal{N}$:
\[
O_{\text{noisy}} = \mathcal{N} \circ O_{\text{perfect}}
.
\]
If $O$ is a single-qubit rotation, then $\mathcal{N}$ is the depolarising channel:
\def\Ndepol{\mathcal{N}^{\text{depol}}}
\begin{equation}
\label{eq:depolarizing-channel}
\Ndepol_p : 
	\rho \mapsto \left(1 - \frac{3p}{4}\right) \rho + \frac{p}{4} \left(X\rho X + Y\rho Y + Z\rho Z\right)
\end{equation}
with parameter $p = 4(1 - F)/3$ with $F$ the fidelity of the operation.

If $O$ is single-qubit initialisation, $\mathcal{N} = \Ndepol_p$ with parameter \mbox{$p=2(1 - F)$}.
The noise map of the controlled-$R_X$ gate is an identical single-qubit depolarising channel on both involved qubits, i.e.\ $\mathcal{N} = \mathcal{N}^{\text{depol}}_p \otimes \mathcal{N}^{\text{depol}}_p$.

Finally, we model electron spin readout by a POVM measurement with the Kraus operators
\begin{equation}
M_0 = 
\begin{pmatrix}
\sqrt{f_0} & 0\\
0 & \sqrt{1 - f_1}
\end{pmatrix}
,\quad
M_1 =
\begin{pmatrix}
\sqrt{1 - f_0} & 0\\
0 & \sqrt{f_1}
\end{pmatrix}
	\label{eq:nv-readout}
\end{equation}
where $1-f_0$ ($1-f_1$) is the probability that a measurement outcome $0$ ($1$) is flipped to $1$ ($0$).

\subsubsection{Simulation speedup via state insertion}
\label{sec:generation-magic}

For generating entanglement between the electron spins of two remote NVs, we simulate a scheme based on single-photon detection, following its experimental implementation in~\cite{humphreys2018deterministic}.
NetSquid was used previously to simulate each generation attempt of this scheme, which includes the emission of a single photon by each NV, the transmission of the photons to the midpoint through a noisy and lossy channel, the application of imperfect measurement operators at the midpoint, and the transmission of the measurement outcome back to the two involved nodes~\cite{dahlberg2019linklayer}.
For larger internode distances, simulating each attempt requires unfeasibly long simulation times due to the exponential decrease in attempt success rate.
To speed up our simulations in the examples studied here, we generate the produced state between adjacent nodes from a model which has shown good agreement with experimental results \cite{humphreys2018deterministic}.
This procedure includes a random duration and noise induced on the storage qubits, as we describe below.

Let us define
\begin{eqnarray*}
	p_{00} &=& \alpha^2 [
	2\pdet (1 - \pdet) (1 - \pdc)
	\\&&
	\qquad
	+
	2\pdc  (1 - \pdc) (1 - \pdet)^2 
	\\&&
	\qquad
	+
	\pdet^2 (1 - \pdc) \cdot \frac{1}{2} (1 + V)
	]
	\\
	p_{10} &=& \alpha (1 - \alpha) \cdot
	[
	(1 - \pdc) \cdot \pdet
	\\&&
	+
	2 \pdc (1 - \pdc) (1 - \pdet)
	]
	\\
p_{01} &=& p_{01}\\
	p_{11} &=& (1 - \alpha)^2 \cdot \pdc
\end{eqnarray*}
where $\pdet$ is the detection probability, $\pdc$ the dark count probability, $V$ denotes photon indistinguishability and $\alpha$ is the bright-state parameter (see Supplementary Note~\ref{app:nv-physical-modelling} for parameter descriptions).
We follow the model of the produced entangled state from the experimental work of ~\cite{humphreys2018deterministic}, whose setup consists of a beam splitter with two detectors located between the two adjacent nodes.
In their model, the unnormalised state is given by
\[
\rho =
\begin{pmatrix}
    p_{00} & 0 & 0 & 0\\
0
&
p_{01}
&
\pm \sqrt{V p_{01} p_{10}}
&
0\\
0
&
\pm \sqrt{V p_{01} p_{10}}
&
p_{10}
&
0\\
    0 & 0 & 0 & p_{11}
\end{pmatrix}
\]
where $\pm$ denotes which of the two detectors detected a photon (each occurring with probability $\frac{1}{2}$).
We also follow the model of~\cite{humphreys2018deterministic} for double-excitation noise and optical phase uncertainty, by applying a dephasing channel to both qubits with parameter $p=\pde/2$, followed by a dephasing channel of one of the qubits, respectively.

The success probability of a single attempt is
\[
	p_{\text{succ}} = p_{00} + p_{01} + p_{10} + p_{11}
.
\]

The time elapsed until the fresh state is put on the electron spins is $ (k - 1) \cdot \Delta t$ with \mbox{$\Delta t := (t_{\text{emission}} + L/c)$}, where $t_{\text{emission}}$ is the delay until the NV centre emits a photon, $L$ the internode distance and $c$ the speed of light in fibre.
Here, $k$ is the number of attempts up to and including successful entanglement generation and is computed by drawing a random sample from the geometric distribution $\Pr(k) = p_{\text{succ}} \cdot (1 - p_{\text{succ}})^{k-1}$.
After the successful generation, we wait for another time $\Delta t$ to mimic the photon travel delay and midpoint heralding message delay.

Every entanglement generation attempt induces dephasing noise on the storage qubits in the same NV system.
We apply the dephasing channel (eq.~\eqref{eq:dephasing-channel}) at the end of the successful entanglement generation, where the accumulated dephasing probability is
\begin{equation}
\label{eq:n-fold-application-nuclear-dephasing}
	\frac{1 - (1 - 2p_{\text{single}})^k}{2}
\end{equation}
where $p_{\text{single}}$ is the single-attempt dephasing probability (see eq.~\eqref{eq:nuclear-dephasing-single-attempt-number-of-attempts} in Supplementary Note \ref{app:nv-physical-modelling}).

\subsubsection{How we choose improved hardware parameters}
\label{sec:improvement-factor}
Here, we explain how we choose `improved' hardware parameters.
Let us emphasise that this choice is independent of the setup of our NetSquid simulations and only serves the purpose of showcasing that NetSquid can assess the performance of hardware with a given quality.

By `near-term' hardware, we mean values for the above defined parameters as expected to be achieved in the near future by NV hardware.
If we say that an error probability is improved by an improvement factor $k$, we mean that its corresponding no-error probability equals $\sqrt[\leftroot{-2}\uproot{2}k]{p_{\text{ne}}}$, where $p_{\text{ne}}$ is the no-error probability of the near-term hardware.
For example, visibility $V$ is improved as 
$\sqrt[\leftroot{-2}\uproot{2}k]{V}$
while the probability of dephasing $p$ of a gate is improved as 
$1 - \sqrt[\leftroot{-2}\uproot{2}k]{1 - p}$.
A factor $k=1$ thus corresponds to `near-term' hardware.
By `uniform hardware improvement by $k$', we mean that all hardware parameters are improved by a factor $k$.
We do not improve the duration of local operations or the fibre attenuation.
The near-term parameter values as well as the individual improvement functions for each parameter can be found in Supplementary Note~\ref{app:nv-physical-modelling}.

\subsubsection{NV repeater chain protocols}
\label{sec:nv-protocols}
For the NV repeater chain, we simulated two protocols: \swapasap~and \withdistill.
Both protocols are composed of five building blocks: \entgen, \move, \unmove, \distill~and \swap.
By \entgen, we denote the simulation of the entanglement generation protocol based on the description in the previous subsection: two nodes wait until a classical message signals that their respective electron spins hold an entangled pair.
In reality, such functionality would be achieved by a link layer protocol~\cite{dahlberg2019linklayer}.
\move~is the mapping of the electron spin state onto a free nuclear spin, and \unmove~is the reverse operation.
The \distill~block implements entanglement distillation between two remote NVs for probabilistically improving the quality of entanglement between two nuclear spins (one at each NV), at the cost of reading out entanglement between the two electron spins.
It consists of local operations followed by classical communication to determine whether distillation succeeded.
The entanglement swap (\swap) converts two short-distance entangled qubit pairs $A-M$ and $M-B$ into a single long-distance one $A-B$, where $A, B$ and $M$ are nodes.
It consists of local operations at $M$, including spin readout, and communicating the measurement outcomes to $A$ and $B$, followed by $A$ and $B$ updating their knowledge of the precise state $A-B$ they hold in the perfect case.
We opt for such tracking as opposed to applying a correction operator to bring $A-B$ back to a canonical state since the correction operator generally cannot be applied to the nuclear spins directly.
Details of the tracking are given in Supplementary Note~\ref{app:nv-pauli-frame}.
The circuit implementations for the building blocks, ``quantum programs" in NetSquid, are given in Supplementary Note~\ref{app:nv-protocols}.

Let us explain the \swapasap~and \withdistill~protocols in spirit; the exact protocols run asynchronously on each node and can be found in Supplementary Note~\ref{app:nv-protocols}.
In the \swapasap~protocol, a repeater node performs \entgen~with both its neighbours, followed by \swap~as soon as it holds the two entangled pairs.
Next, \withdistill~is a nested protocol on $2^n + 1$ nodes (integer $n\geq 0$) with distillation at each nesting level which is based on the BDCZ protocol~\cite{briegel1998quantum}.
For nesting level $n=0$, there are no repeaters and the two nodes only perform \entgen~once.
For nesting level $n>0$, the chain is divided into a left part and a right part of $2^{n-1} + 1$ nodes, and the middle node (included in both parts) in the chain generates twice an entangled pair with the left end node following the $(n-1)$-level protocol; \move~is applied in between to free the electron spin.
Subsequently, \distill~is performed with the two pairs as input (restart if distillation fails), after which the same procedure is performed on the right.
Once the right part has finished, the middle node performs \swap~to connect the end nodes.
If needed, \move~and \unmove~are applied prior to \distill~and \swap~in order achieve the desired configuration of qubits in the quantum processor, e.g. for \distill~to ensure that the two involved NVs hold an electron-electron and nuclear-nuclear pair of qubits, instead of electron-nuclear for both entangled pairs.

\section{Data availability}
The data presented in this paper have been made available at \url{https://doi.org/10.34894/URV169} \cite{netsquidpaperdata}.

\section{Code availability}
The NetSquid-based simulation code that was used for the simulations in this paper has been made available at \url{https://doi.org/10.34894/DU3FTS} \cite{netsquidpapercode}.

\section*{Acknowledgements} 
This work was supported by the Dutch Research Cooperation Funds (SMO), 
the European Research Council through a Starting Grant (S.W.), the QIA project (funded by European Union's Horizon 2020, Grant Agreement No. 820445) and the Netherlands Organisation for Scientific Research (NWO/OCW), as part of the Quantum Software Consortium program (project number 024.003.037/3368). 
The authors would like to thank Francisco Ferreira da Silva, Wojciech Kozlowski and Gayane Vardoyan for critical reading of the manuscript.
The authors would like to thank Gustavo Amaral, Guus Avis, Conor Bradley, Chris Elenbaas, Francisco Ferreira da Silva, Sophie Hermans, Roeland ter Hoeven, Hana Jirovsk\'a, Wojciech Kozlowski, Matteo Pompili, Arian Stolk and Gayane Vardoyan for useful discussions.

\section*{Author Contributions}  
T.C. realised the NV repeater chain and the quantum switch simulations.
R.K., L.W. realised the benchmarking simulations.
D.M., J.R. realised the atomic ensembles simulations.
R.K. and J.O. designed NetSquid's software architecture and R.K. led its software development.
T.C, A.D, R.K., D.M., L.N., J.O., M.P., F.R, J.R., M.S., A.T., L.W., and S.W designed use case driven architectures, and contributed to the development of NetSquid and the modelling libraries used in the simulations.
W.J., D.P., A.T. contributed to the optimal execution of simulations on computing clusters.
T.C., D.E., R.K., D.M. and S.W. wrote the manuscript.
All authors revised the manuscript.
D.E. and S.W. conceived and supervised the project.

\section*{Competing Interests statement}
The authors declare no competing interests.

\def\bibsection{}
\section*{References}
\bigbreak

\bibliographystyle{ieeetr}
\bibliography{bibliography}

\onecolumngrid
\appendix
\renewcommand{\thesection}{\arabic{section}}
\renewcommand{\thesubsection}{\Alph{subsection}}
\section{Anatomy of the NetSquid Simulator}
\label{appendix:netsquid}

This section supplements the Methods, section~\nameref{sec:methods-netsquid}, by going into more depth on specific details of NetSquid's design.
The version of NetSquid that we consider is 0.10.
For up-to-date documentation of the latest NetSquid version, including a detailed user tutorial, code examples, and its application programming interface, please visit the NetSquid website: \url{https://netsquid.org}~\cite{netsquid-website}.

\subsection{Qubits and their quantum state formalisms}
\label{sec:qstate-formalism}

The \textit{qubits} sub-package of NetSquid, shown in Figure~\ref{fig:netsquid-arch} (main text), provides a specialised quantum computation library for tracking the lifetimes of many qubits across a quantum network.
A class diagram of the main classes present in this sub-package is shown in in Supplementary Figure~\ref{fig:qubits-uml}.
Rather than assigning a single quantum state for a predefined number of qubits,
both the number of qubits and the quantum states describing them are managed dynamically during a simulation run.
Every \textit{Qubit} (\texttt{Qubit}) object references a \textit{shared quantum state} (\texttt{QState}) object,
which varies in size according to the number of qubits sharing it.
When two or more qubits interact, for instance via a multi-qubit operation, their respective shared quantum states are merged together.
On the other hand, when a qubit is projectively measured or discarded it can be split from the quantum state it's sharing and optionally be assigned a new single-qubit state.

The \texttt{QState} class is an interface for shared quantum states that NetSquid implements for four different \textit{quantum state formalisms} -- described in more detail below.
To allow simulations to seamlessly switch between formalisms NetSquid offers a formalism agnostic API,
which is defined in the \textit{qubitapi} module.
The functions in this API take as their primary input parameters the qubits to manipulate and the \textit{operators} (\texttt{Operator}) describing a quantum operation to perform, if applicable.
The merging and splitting of shared quantum states is handled automatically under the hood, as are conversions between states using different formalisms (where this is possible).
This allows users to program in a ``qubit-centric'' way, by for instance applying local operations to qubits at a network node without knowledge of their positions within a quantum state representation or any entanglement they may have across the network.

\begin{figure}[bt]
    \centering
    \includegraphics[width=1.0\textwidth]{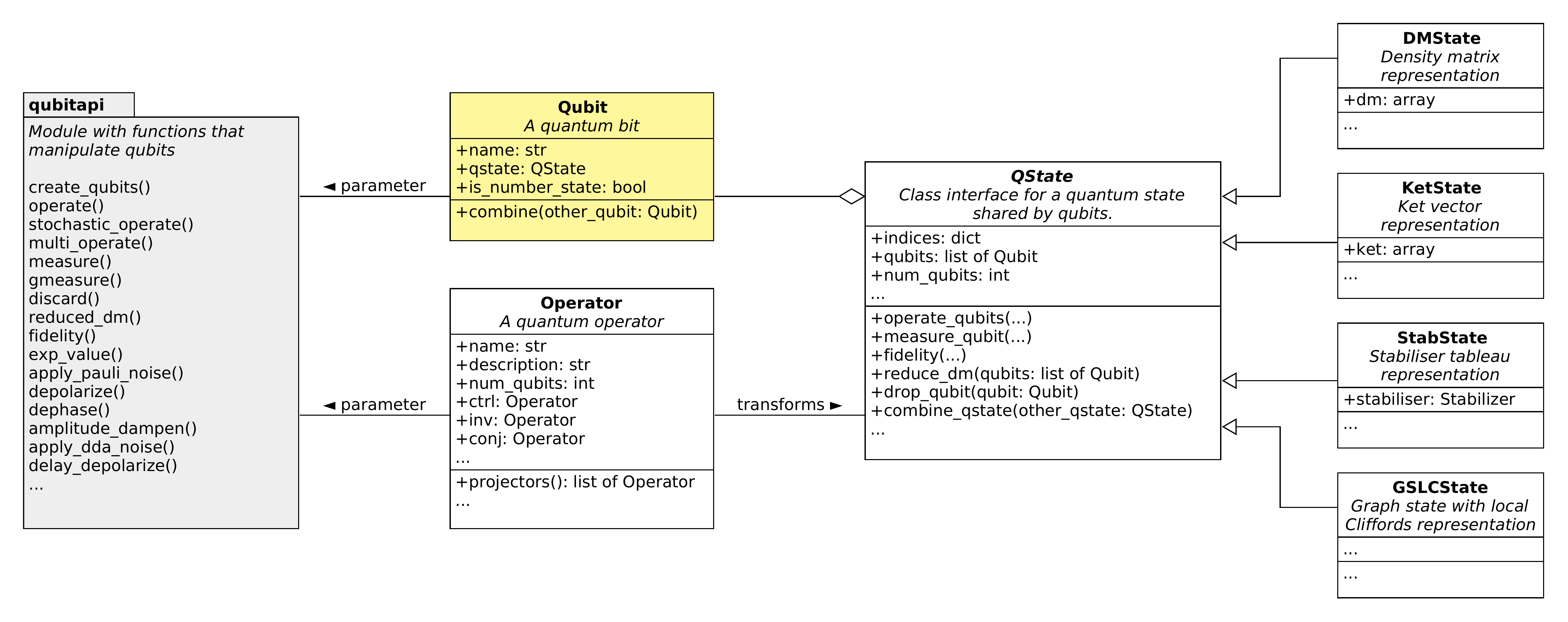}
    \caption{\textbf{Design overview of NetSquid's qubits sub-package.} The main classes and module of the \textit{netsquid.qubits} sub-package. \texttt{Qubit} objects can be manipulated,
        as described for instance by \texttt{Operator} objects, using the functions of the \textit{qubitapi} module.
        Under the hood the qubits share a specific sub-class of the \texttt{QState} interface.
        Ellipses indicate that not all of a class's public variables and methods are listed.}
    \label{fig:qubits-uml}
\end{figure}

We proceed to give a high-level description of the available quantum state formalisms.
The first two formalisms are ket state vectors (KET) and density matrices (DM), which both enable universal quantum computation.
A ket state vector represents a quantum pure state, while a density matrix can represent statistical ensembles of pure states.
The stabiliser formalism (STAB) \cite{gottesman1998heisenberg,aaronson2004improved} and graph states with local Cliffords formalisms (GSLC) \cite{anders2006fast} can only represent stabiliser states.
Stabiliser states form a subset of all quantum states that are closed under the application of: 
\begin{itemize}
    \item{\textit{Clifford gates}. Each Clifford gate can be written as circuit consisting of the following three gates only: the Hadamard gate $H$ (eq.~\eqref{eq:hadamard}), the phase gate $\dyad{0} + i\dyad{1}$ and the CNOT gate $\dyad{00} + \dyad{01} + \dyad{10}{01} + \dyad{01}{10}$. Not all unitaries are Clifford gates;}
    \item{single-qubit measurements in the standard ($\ket{0},\ket{1}$) basis.}
\end{itemize}
As such, for the STAB and GSLC formalisms quantum operations are limited to these two procedures.
The runtime complexity trade-off between GSLC and STAB is nontrivial, since the former is faster on single-qubit unitaries, where the latter outperforms in two-qubit gates.
An overview of the four formalisms and their runtime complexities can be found in Supplementary Table~\ref{netsquid-formalisms-table}.

\begin{table}[tb]
\begin{tabular}{lm{2cm}m{2cm}m{3cm}m{5cm}}
\hline
\hline
 & Density \newline Matrix (DM) & Ket state \newline vector (KET) & Stabiliser \newline tableau (STAB) & Graph state with \newline local Cliffords (GSLC)\\
\hline
Is universal & Yes & Yes & No\footnotemark[1] & No\footnotemark[1] \\
Supports mixed states & Yes & No & No & No \\
Memory (bits) & $128 \times 2^{2n}$ & $128 \times 2^n$ & $2n^2 + n$ & $\mathcal{O}(nd + n)$ \\
Operating complexity &$\mathcal{O}(2^{3n})$\footnotemark[2] &$\mathcal{O}(2^{2n})$\footnotemark[2]&$\mathcal{O}(n)$&Single qubit gates: $\mathcal{O}(1)$ \newline Two-qubit gates: $\mathcal{O}(d^2 + 1)$\\
Measurement complexity & $\mathcal{O}(2^{3n})$\footnotemark[2]& $\mathcal{O}(2^{2n})$\footnotemark[2] & $\mathcal{O}(n^3)$ & $\mathcal{O}(d^2 + 1)$ \\
\hline
\hline
\end{tabular}
\footnotetext[1]{Can only represent stabiliser states. The only operators that can operate on these states are Clifford operators.}
\footnotetext[2]{A stricter upper bound exists, depending on the matrix multiplication implementation.}
\caption{
    \textbf{The four different quantum state formalisms implemented in NetSquid.}
    Where $n$ is the amount of qubits in the quantum states and $d$ is the average amount of edges per vertex in the GSLC formalism with $0 \leq d < n$.}
\label{netsquid-formalisms-table}
\end{table}

Now, let us describe for each of the formalisms how a quantum state is represented.
An example of the different representations of the same quantum state is given in Fig.~\ref{formalism-comparison}.

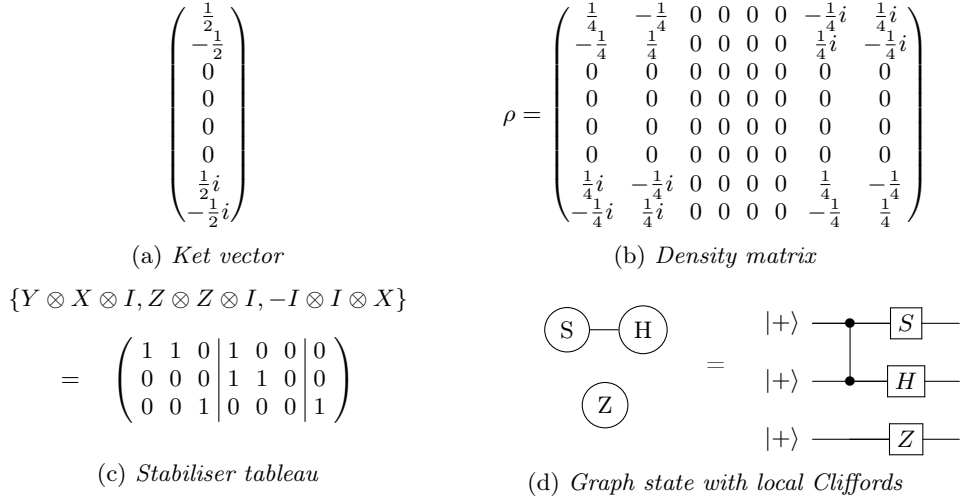
\begin{figure}[tb]
    \begin{subfigure}{0.4\textwidth}
        $\begin{pmatrix}
        \frac{1}{2} \\ -\frac{1}{2} \\ 0 \\ 0 \\ 0 \\ 0 \\ \frac{1}{2}i \\ -\frac{1}{2}i
        \end{pmatrix}$
        \caption{\textit{Ket vector}}
    \end{subfigure}
    \begin{subfigure}{0.4\textwidth}
        $\rho = \begin{pmatrix}
        \frac{1}{4} & -\frac{1}{4} & 0 & 0 & 0 & 0 & -\frac{1}{4}i & \frac{1}{4}i \\
        -\frac{1}{4} & \frac{1}{4} & 0 & 0 & 0 & 0 & \frac{1}{4}i & -\frac{1}{4}i \\
        0 & 0 & 0 & 0 & 0 & 0 & 0 & 0 \\
        0 & 0 & 0 & 0 & 0 & 0 & 0 & 0 \\
        0 & 0 & 0 & 0 & 0 & 0 & 0 & 0 \\
        0 & 0 & 0 & 0 & 0 & 0 & 0 & 0 \\
        \frac{1}{4}i & -\frac{1}{4}i & 0 & 0 & 0 & 0 & \frac{1}{4} & -\frac{1}{4} \\
        -\frac{1}{4}i & \frac{1}{4}i & 0 & 0 & 0 & 0 & -\frac{1}{4} & \frac{1}{4} \\
        \end{pmatrix}$
        \caption{\textit{Density matrix}}
    \end{subfigure}
    \begin{subfigure}{0.4\textwidth}
        $\{Y \otimes X \otimes I, Z \otimes Z \otimes I, -I \otimes I \otimes X\}$
        \begin{equation*}
        =\quad \left(\begin{array}{ccc|ccc|c}
        1 & 1 & 0 & 1 & 0 & 0 & 0 \\
        0 & 0 & 0 & 1 & 1 & 0 & 0 \\
        0 & 0 & 1 & 0 & 0 & 0 & 1 \\
        \end{array}\right)
        \end{equation*}
        \caption{\textit{Stabiliser tableau}}
    \end{subfigure}
    \begin{subfigure}{0.4\textwidth}
        \minipage{0.15\textwidth}
        \begin{tikzpicture}
        \node[shape=circle,draw=black,radius=1] (A) at (0, 0) {S};
        \node[shape=circle,draw=black] (B) at (1, 0) {H};
        \node[shape=circle,draw=black] (C) at (0.5, -1) {Z};
        \path [-] (A) edge (B); 
        \end{tikzpicture}
        \endminipage\hspace{1cm}
        \minipage{0.05\textwidth}
        =
        \endminipage\hspace{1cm}
        \minipage{0.15\textwidth}
        \[
        \Qcircuit @C=1.4em @R=1.1em {
            \lstick{\ket{+}} & \ctrl{1} \qw & \gate{S} \qw & \qw \\
            \lstick{\ket{+}} & \control \qw & \gate{H} \qw & \qw \\
            \lstick{\ket{+}} & \qw & \gate{Z} \qw & \qw \\
        }
        \]
        \endminipage
        \caption{\textit{Graph state with local Cliffords}}
    \end{subfigure}
    \caption{\textbf{Quantum state representations available in NetSquid.} Four different representations of the same quantum state $\ket{\psi} = \frac{1}{\sqrt{2}}(\ket{00} + i\ket{11})\ket{-}$. Each representation type is supported by NetSquid and has different trade-offs (see text of section~\ref{sec:qstate-formalism} in Supplementary Note~\ref{appendix:netsquid}).}
	\label{formalism-comparison}
\end{figure}

\subsubsection*{Ket vectors (KET)}

In the KET formalism, an $n$-qubit pure state $\ket{\psi} = \sum_{k=1}^{2^n} c_k \ket{k}$ is stored as a vector of length $2^n$ containing the complex amplitudes $c_k$. Here, $\ket{k}$ denotes the product state of the binary representation of $k$, e.g. $\ket{5} = \ket{1} \otimes \ket{0} \otimes \ket{1}$.

\subsubsection*{Density matrices (DM)}

The density matrix of a pure state $\ket{\psi}$ is $\dyad{\psi} = \ket{\psi} \cdot \left(\ket{\psi}\right)^{\dagger}$, where $\cdot$ denotes matrix multiplication and $(.)^{\dagger}$ refers to complex transposition.
An $n$-qubit mixed state is a statistical ensemble of $n$-qubit pure states and can be represented as
$$\sum_{k=1}^m p_k\dyad{\psi_k}$$
where $\ket{\psi_1}, \dots, \ket{\psi_m}$ are $n$-qubit pure states (with $1\leq m \leq n$) and the $p_k$ are probabilities that sum to 1.
In DM, the density matrix of a pure or mixed state is represented as a matrix of dimension $2^n \times 2^n$ with complex entries.

\subsubsection*{Stabiliser tableaus (STAB)}

In the stabiliser formalism~\cite{gottesman1998heisenberg}, one tracks the generators of the stabiliser group of a state.
We briefly explain the concept here; for a more accessible introduction to the topic, we refer to \cite{nielsen2000quantum}.
In order to define a stabiliser group, let us give the Pauli group, which consists of strings of Pauli operators with multiplicative phases $\pm 1, \pm i$:
\[
	\{\beta \cdot \bigotimes_{k=1}^n P_k \mid \text{$P_k \in \{\unit_2, X, Y, Z\}$} \text{ and } \beta\in \{\pm 1, \pm i\}\}
	.
	\]
A stabiliser group is a subgroup of the Pauli group which is commutative (i.e. any two elements $A$ and $B$ satisfy $A\cdot B = B\cdot A$) and moreover does not contain the element $-\unit_2 \otimes \unit_2 \otimes \dots\otimes \unit_2$.
In case the stabiliser group contains $2^n$ elements, there is a unique quantum state $\ket{\psi}$ for which each element $A$ from the stabiliser group stabilises $\ket{\psi}$, i.e. $A\ket{\psi} = \ket{\psi}$.
Not all quantum states have such a corresponding stabiliser group; those that do are called stabiliser states.
The intuition behind the stabiliser state formalism is that one tracks how the stabiliser group is altered by Clifford operations and $\ket{0}/\ket{1}$-basis measurements.
Since the stabiliser state belonging to a stabiliser group is unique, one could in principle always convert the group back to any other formalism, such as KET.
Concrete examples of stabiliser groups and their corresponding stabiliser states are:
\begin{itemize}
	\item{the stabiliser group $\{\unit_2, Z\}$, which corresponds to the state $\ket{0}$;}
	\item{the stabiliser group $\{\unit_2 \otimes \unit_2, \unit_2 \otimes Z, Z\otimes \unit_2, Z\otimes Z\}$, which corresponds to the state $\ket{0} \otimes \ket{0}$;}
	\item{the stabiliser group $\{\unit_2 \otimes \unit_2, X \otimes X, Z\otimes Z, -Y\otimes Y\}$, which corresponds to the state $(\ket{00} + \ket{11})/\sqrt{2}.$}
\end{itemize}
Rather than tracking the entire $2^n$-sized stabiliser group, it suffices to track a generating set, i.e. a set of $n$ Pauli strings whose $2^n$ product combinations yield precisely the $2^n$ elements of the stabiliser group.
The choice of generators is not unique.
For the examples given above, example sets of stabiliser generators are:
\begin{itemize}
	\item{for $\ket{0}$, the stabiliser group is generated by the single element $Z$, since $Z^2 = \unit_2$}
	\item{for $\ket{00}$, the stabiliser group is generated by $\{Z\otimes \unit_2, \unit_2 \otimes Z\}$, since squaring any of these two yields $\unit_2\otimes \unit_2$, while multiplying them yields $Z\otimes Z$;}
	\item{for the state $(\ket{00} + \ket{11})/\sqrt{2}$, one possible set of of generators is $\{X\otimes X, Z\otimes Z\}$.}
\end{itemize}

In NetSquid we store generators as a stabiliser tableau:
	$$\begin{vmatrix}X & Z & P \\\end{vmatrix} = 
	\begin{vmatrix}
	x_{11} & \dots & x_{1n} & z_{11} & \dots & z_{1n} & p_1\\
	\vdots & \ddots & \vdots & \vdots & \ddots & \vdots & \vdots\\
	x_{n1} & \dots & x_{nn} & z_{n1} & \dots & z_{nn} & p_n\\
	\end{vmatrix} \text{ where } p_k, x_{jk}, z_{jk} \in \{0, 1\}, 0 < j, k \leq n
	$$
	The $k$-th generator corresponds to the $k$-th row of this tableau and is given by
	 $$(-1)^{p_k}\bigotimes_{j=1}^nX^{x_{jk}}Z^{z_{jk}}$$
For updating the stabiliser tableau after the application of a Clifford gate or a $\ket{0}/\ket{1}$-basis measurement, NetSquid uses the algorithms by \cite{gottesman1998heisenberg} and \cite{aaronson2004improved}.
The runtime performance of stabiliser tableau algorithms is a direct function of the number of qubits:
linear for applying single- or two-qubit Clifford unitaries, which any Clifford can be composed into, and cubic for single-qubit measurement~\cite{gottesman1998heisenberg}.

\subsubsection*{Graph states with local Cliffords (GSLC)}

The last formalism is GSLC: graph states with local Cliffords \cite{anders2006fast}.
Graph states are a subset of all stabiliser states (see~\cite{hein2006entanglement} for a review) and an $n$-qubit graph state $\ket{\psi}$ can be written as
\begin{equation}
	\ket{\psi} = \prod_{(j, k) \in E} Z_{jk} \ket{+}^{\otimes n}
	\label{eq:graph-state}
\end{equation}
	where $Z_{jk}$ indicates a controlled-$Z$ gate $\dyad{00} + \dyad{01} + \dyad{10} - \dyad{11}$ between qubits $j$ and $k$, and we have denoted $\ket{+} = (\ket{0} + \ket{1})/ \sqrt{2}$.
As such, a graph state is completely determined by the set of qubit index pairs $(j, k)$ at which a controlled-$Z$ operation is performed.
These indices can be captured in a graph with undirected edges; in eq.~\eqref{eq:graph-state}, the edge set is $E$.
Each stabiliser state can be written as a graph state, followed by the application of single-qubit Clifford operations.
Thus, a stabiliser state in the GSLC formalism is represented by a set of edges $E$ and a list of $n$ single qubit Cliffords.
There exist 24 single-qubit Cliffords, so the Clifford list only requires $\mathcal{O}(n)$ space.
For updating the graph and the list of single-qubit Cliffords after the application of a Clifford gate or a $\ket{0}/\ket{1}$-basis measurement, NetSquid uses the algorithms by \cite{anders2006fast}.
The runtime scaling of the graph-state-based formalism depends on the edge degree $d$ of the vertices involved in the operation
-- constant-time for single-qubit Cliffords, quadratic in $d$ for two-qubit Cliffords and measurement -- 
and thus scales favourably if the graph is sparse.

\subsection{The PyDynAA simulation engine}

The discrete-event modelling framework used by NetSquid is provided by the Python package PyDynAA,
which is based on the core engine layer of DynAA, a system analysis and design tool~\cite{de2013model}.
This foundation provides a simple yet powerful language for describing large and complex system architectures.
To realise PyDynAA, the simulation engine core was written in C\texttt{++} for increased performance, and bindings to Python were added using Cython.
NetSquid takes advantage of the Cython headers exposed by PyDynAA to efficiently integrate the engine into its own compiled C extension libraries.

Several of NetSquid's sub-packages depend and build on the classes provided by PyDynAA, as illustrated in Figure~\ref{fig:netsquid-arch} (main text).
In Supplementary Figure~\ref{fig:pydynaa-uml} we highlight several of these key classes and how they interact with the simulation timeline in more detail, namely:
the simulation engine (\texttt{SimulationEngine}), events (\texttt{Event} and \texttt{EventType}), simulation entities (\texttt{Entity}), and event handlers (\texttt{EventHandler}).
We proceed to describe the concepts these classes represent in more detail.

Simulation \textit{entities} represent anything in the simulation world capable of generating or responding to events.
They may be dynamically added or removed during a simulation.
The \texttt{Entity} superclass provides methods for scheduling events to the timeline at specific instances and waiting for them to trigger.
The intended use is that users subclass the \texttt{Entity} class to implement their own entities.
The \textit{simulation engine} efficiently handles the scheduling of events at arbitrary (future) times by storing them in a self-balancing binary search tree.
Events may only be scheduled by entities, which ensures that events always have a source entity.
If an entity is removed during a simulation, then any future events it had scheduled will no longer trigger.

An entity responds to events by registering an event handler object with a callback function.
Responses can be associated to a specific type, source, and id (including wildcard combinations).
The simulation engine runs by stepping sequentially from event to event in a discrete fashion and checking if any event handlers in its registry match.
A hash table together with an efficient hashing algorithm ensure efficient lookups of the event handlers in the registry.

PyDynAA implements an \textit{event expression} class to allow entities to wait on logical combinations of events.
Atomic event expressions, which describe regular wait conditions for standard events, can be combined to form composite expressions using logical \textit{and} and \textit{or} operators to any depth.
Event expressions enable NetSquid simulations to deal with timing complexities.
This feature has been used extensively in NetSquid to model both the internal behaviour of hardware components, as well as for programming network protocols.
As example, consider DEJMPS entanglement distillation \cite{deutsch1996quantum}: two nodes participate in this protocol and a node can only decide whether the distillation succeeded or failed when both its local quantum operations have finished and it has received the measurement outcome from the remote node.
Thus, the node waits for the logical \textit{and} of the receive-event and the event that the local operations have finished.

\begin{figure}[tb]
    \centering
    \includegraphics[width=0.9\textwidth]{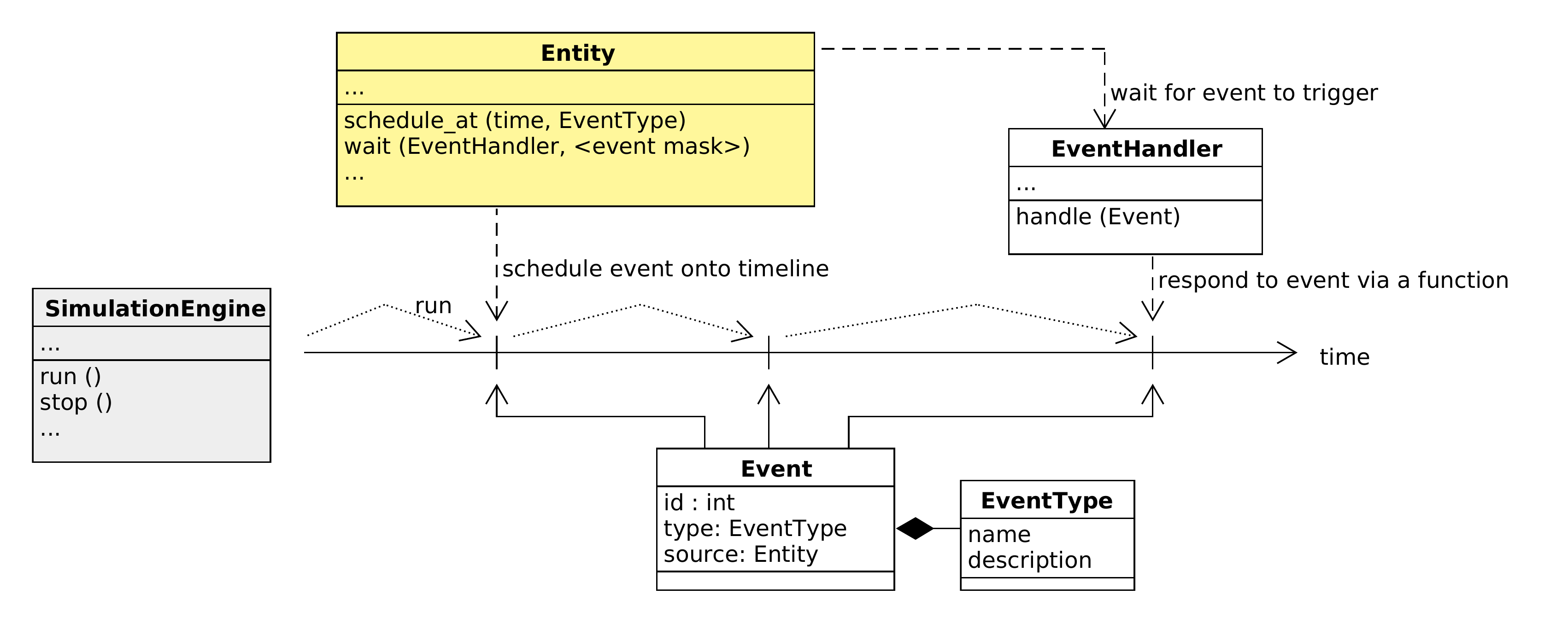}
    \caption{\textbf{Design overview of the PyDynAA package.}
    Schematic overview of key classes defined by the PyDynAA package, the discrete-event simulation engine used by NetSquid.
    Also shown is the relation of each class to the simulation timeline. Events are scheduled onto the simulation timeline by \texttt{Entity} objects. Entities wait for events to trigger by registering \texttt{EventHandler}s,
    which respond to an event by passing it as input to a specified callback function.
    The events to wait for can be specified by their type, id, and source entity. 
    Ellipses indicate that not all of a class's public variables and methods are listed.
    Omitted from this class diagram is the \texttt{EventExpression} class -- see the text for more details.}
    \label{fig:pydynaa-uml}
\end{figure}

\subsection{The modular component modelling framework}

The physical modelling of network devices is provided by several NetSquid sub-packages: \textit{components}, \textit{models} and \textit{nodes}, which are shown stacked with relation the NetSquid package in Figure~\ref{fig:netsquid-arch} (main text).
The pivotal base class connecting all them is the \textit{component} (\texttt{Component}), which is used to model all hardware devices.
Specifically, it represents all physical entities in the simulation, and as such sub-classes the \textit{entity} (\texttt{Entity}),which enables it to interact with the event timeline.
In Supplementary Figure~\ref{fig:components-uml} we show a class diagram of the component class and its relationships to other classes from these sub-packages.

\begin{figure}[htb]
    \centering
    \includegraphics[width=1.0\textwidth]{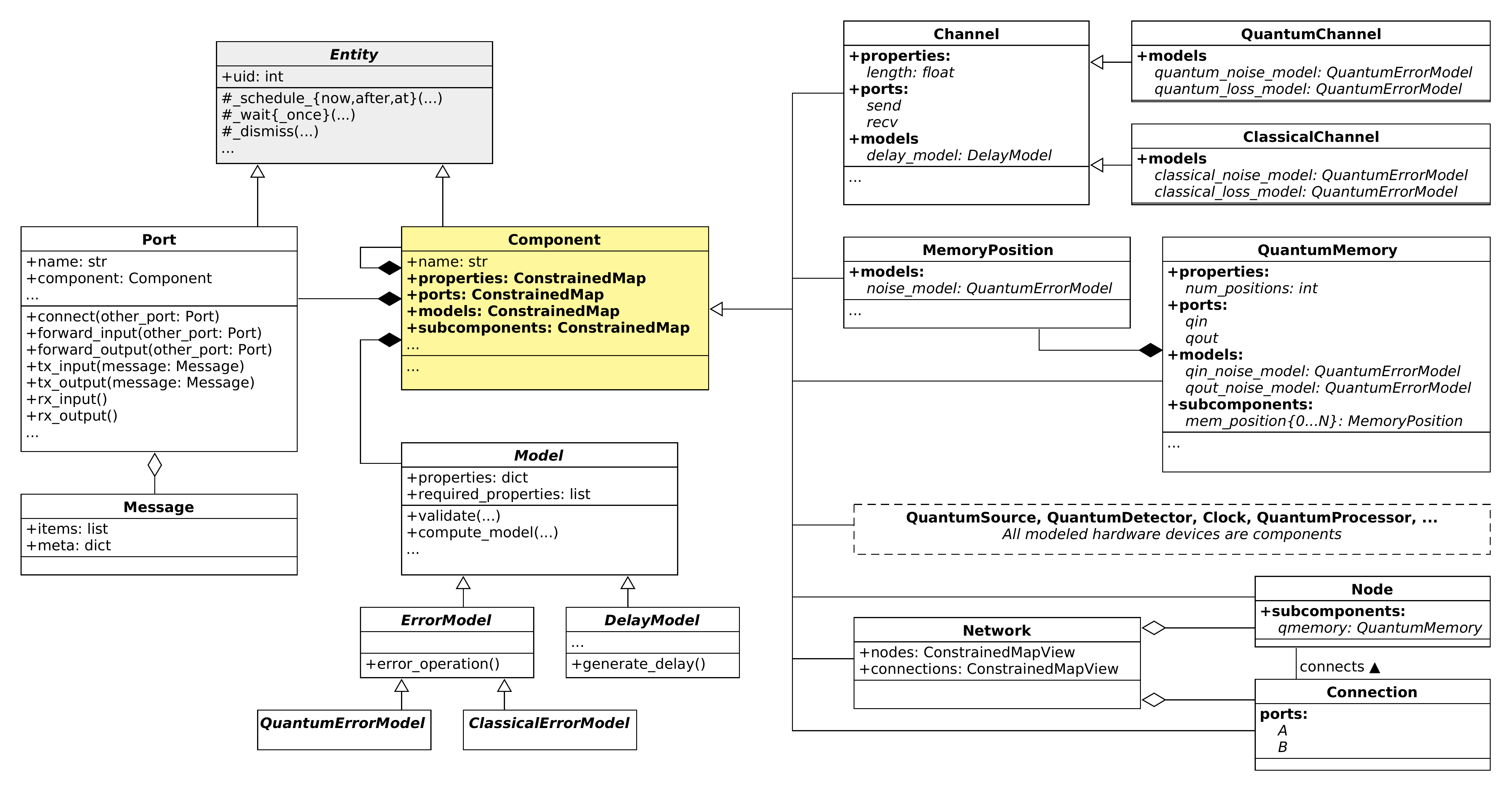}
    \caption{\textbf{Design overview of components in NetSquid.}
        Class diagram for the \texttt{Component} class, a simulation entity that is used to model all network hardware devices, including composite components such as nodes, connections and the network itself. A component is shown to be composed of \textit{properties}, \textit{ports}, \textit{models} and \textit{subcomponents}.
    Ellipses indicate that not all of a class's public variables and methods are listed.}
    \label{fig:components-uml}
\end{figure}

The modularity of NetSquid's modelling framework is achieved by the composition of components in terms of \textit{properties}, \textit{models}, communication \textit{ports} and \textit{subcomponents}.
A component's \textit{properties} are values that physically characterise it, such as the length of a channel or the frequency of a source.
A special \textit{constrained map} (\texttt{ConstrainedMap}) container is used to store the properties (as well as the other composed objects) to give control of the expected types and immutability of properties during a simulation.
\textit{Models} (\texttt{Model}) are used to describe the physical behaviour of a component, such as the transmission delay of a channel, or the quantum decoherence of a qubit in memory.
Model objects are essentially elaborate functions and generally do not store any state;
when a model is called it is passed its component's properties, in addition to any modelling specific input,
such as, in the case of a quantum noise model, the qubit to apply noise and the time the qubit has been waiting on a memory.
Components can be composed of other \textit{subcomponents}, which allows networks to be pieced together in a very modular fashion.
For instance, a complete network can be represented by a single component, which is composed of node and connection sub-components,
which in turn are composed of devices such as channels, sources, memories, etc.
To streamline and automate the communication between components, including to and from sub-components, components can be linked using \textit{ports} (\texttt{Port}) that can send, receive and forward both quantum and classical \textit{messages} (\texttt{Message}).

While the component base class defines a modular interface for modelling all kinds of hardware, it doesn't internally implement any event-driven behaviour itself.
That behaviour is implemented by a library of base classes that sub-class \texttt{Component}.
The right half of Supplementary Figure~\ref{fig:components-uml} shows the sub-classing hierarchy of the provided components, ranging from quantum and classical channels, quantum memory and processing devices, sources, detectors, clocks, to nodes, connections, and networks.

The \textit{quantum processor} (\texttt{QuantumProcessor}) is a component from the base class library used for modelling general quantum processing devices.
It sub-classes the \textit{quantum memory} (\texttt{QuantumMemory}) component, from which it inherits a collection of \textit{quantum memory positions} (\texttt{MemoryPosition}) for tracking the quantum noise of stored qubits.
The processor can assign a set of \textit{physical instructions} to these positions to describe the operations possible for manipulating their stored qubits,
such as quantum gates and measurements, or initialisation, absorption, and emission processes.
The physical instructions map to general device-independent instructions, for which they specify physical models such as duration and error models specific to the modelled device.
This mapping allows users to write \textit{quantum programs} in terms of device-independent instructions and re-use them across devices.
The quantum programs can include classical conditional logic, make use of parallel execution (if supported by the device), and import other programs.

\subsection{Asynchronous programming networks using protocols}

While components are entities in the simulation describing physical hardware, \textit{protocols} -- represented by the \texttt{Protocol} base class as shown in Supplementary Figure~\ref{fig:protocol-uml} -- are entities that describe the intended virtual behaviour of a simulation.
In other words, the protocol base class is used to model the various layers of software running on top of the components at the various nodes and connections of a network.
That can include, for instance, any automated control software at the physical or link layers of a quantum network stack, up to higher-level programs written at the application layer.

Protocols in NetSquid can be likened to background processes: they can be started, stopped, as well as reset to clear any state.
They can also be nested i.e. a protocol can manage the execution of \textit{sub-protocols} under its control.
To communicate changes of state, such as a successful or failed run, protocols can use a \textit{signalling} mechanism (\texttt{Signal}).

NetSquid defines several sub-classes of the protocol base class that add extra restrictions or functionality.
To restrict the influence of a protocol to only a local set of nodes the \textit{local protocol} (\texttt{LocalProtocol}) can be used.
Similarly, to restrict a protocol to executing on only a single node, which is a typical use case, a \textit{node protocol} (\texttt{NodeProtocol}) is available.
The \textit{service protocol} (\texttt{ServiceProtocol}) describes a protocol interface in terms of the types of requests and responses they support.
Lastly, a \textit{data node protocol} adds functionality to process data arriving from a port linked to a connection, and the \textit{timed node protocol} supports carrying out actions at regularly timed intervals.

Programming a protocol involves waiting for and responding to events,
which is achieved in the simulation engine by defining event handlers that wrap callback functions.
As the complexity of a protocol grows, typically the flow and dependencies of the callback calls do too.
To make the asynchronous interaction between protocol and component entities easier and more intuitive to program and read, the main execution function of a protocol (the \texttt{run()} method) can be suspended mid-function to wait for certain combinations of events to trigger.
This is implemented in Python using the \texttt{yield} statement, which takes as its argument an event expression.
Several helper methods have been defined that generate useful event expressions a protocol can \textit{await}, for instance: \texttt{await\_port\_input()} to wait for a message to arrive on a port, or \texttt{await\_timer()} to have the protocol sleep for some time.

\begin{figure}[htb]
    \centering
    \includegraphics[width=1.0\textwidth]{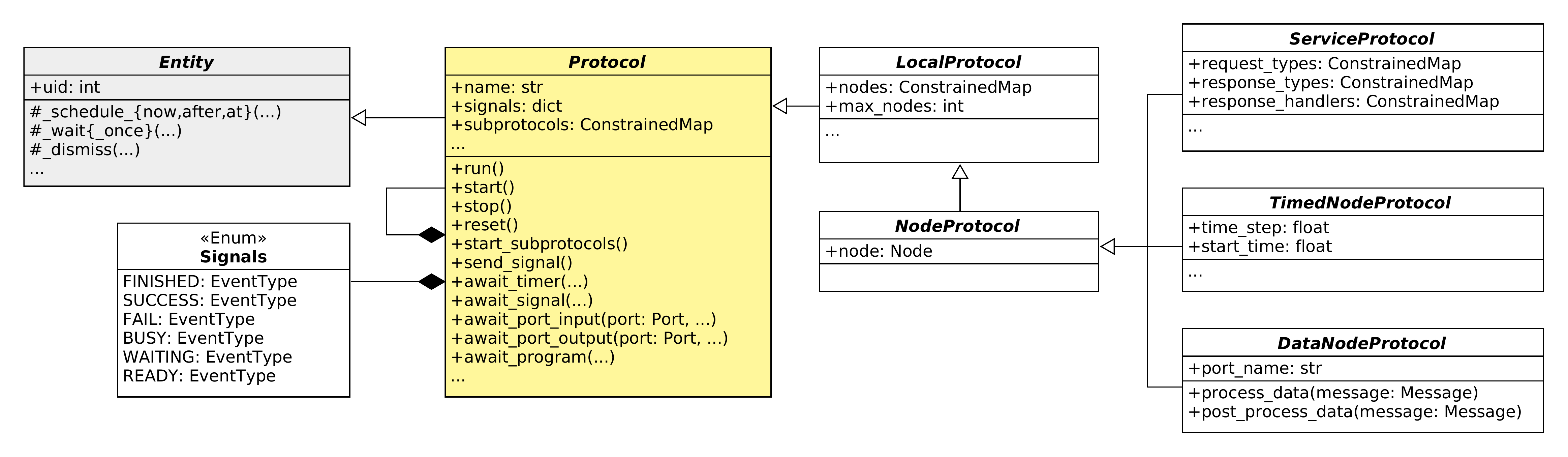}
    \caption{\textbf{Design overview of protocols in NetSquid.}
        Class diagram of the \texttt{Protocol} class and its subclasses.
    Ellipses indicate that not all of a class's public variables and methods are listed.}
    \label{fig:protocol-uml}
\end{figure}

\section{Quantum circuits and network setups for benchmarking}
\label{appendix:benchmarking}

In this section we extend the Methods, section~\nameref{sec:methods-benchmarking}, to provide additional details on the benchmarking simulations presented in the Results, section~\nameref{sec:benchmarking}.

\subsection{Benchmarking of quantum computation runtime}

The quantum circuit used to benchmark the runtime for generating an $n$ qubit GHZ state is shown in Supplementary Figure~\ref{fig:ghz-circuit}.
The $n$ qubits are created in NetSquid with independent quantum states and are combined into the larger state via the CNOT operation.
The measurement operations at the end of the circuit are performed sequentially and each split the measured qubit from its shared quantum state.
Unless otherwise specified the KET and DM formalisms utilise memoization (see Methods, section~\nameref{sec:methods-qubits}).
Memoization is effective because the circuit is successively iterated 30 times.
The reported runtime is the mean runtime of the iterations.
For the baseline comparison with the ProjectQ simulator we set up the circuit in an analogous way to NetSquid,
and its default \texttt{MainEngine} was used with no special settings applied.
Qubits are similarly added sequentially to the growing state via the CNOT operation, and also the measurements are performed sequentially with the measured qubit directly deallocated afterwards.

The quantum circuit used to benchmark the runtime of only the quantum computation involved for a simple repeater chain involving $n$ qubits is shown in Supplementary Figure~\ref{fig:repchain-circuit}.
It is implemented for NetSquid and ProjectQ similarly to the GHZ benchmark, with qubits only combining their quantum states when a multi-qubit gate is performed.
An option has been added to keep qubits \textit{inplace} after measurements i.e. they are not split from their shared quantum states -- in ProjectQ this is achieved by keeping a reference to prevent deallocation.
Noise is applied to each qubit after entanglement by selecting a Pauli gate to mimic depolarising noise, which is done deterministically for convenience.
For this process the runtime is also determined as the mean of 30 successive iterations.

To benchmark the runtimes of quantum computation circuits the processes were timed in isolation from any setup code using the Python \textit{timeit} package.
Python garbage collection is disabled during the timing of each process.
To avoid fluctuations due to interfering CPU processes the reported time is a minimum of five repetitions.

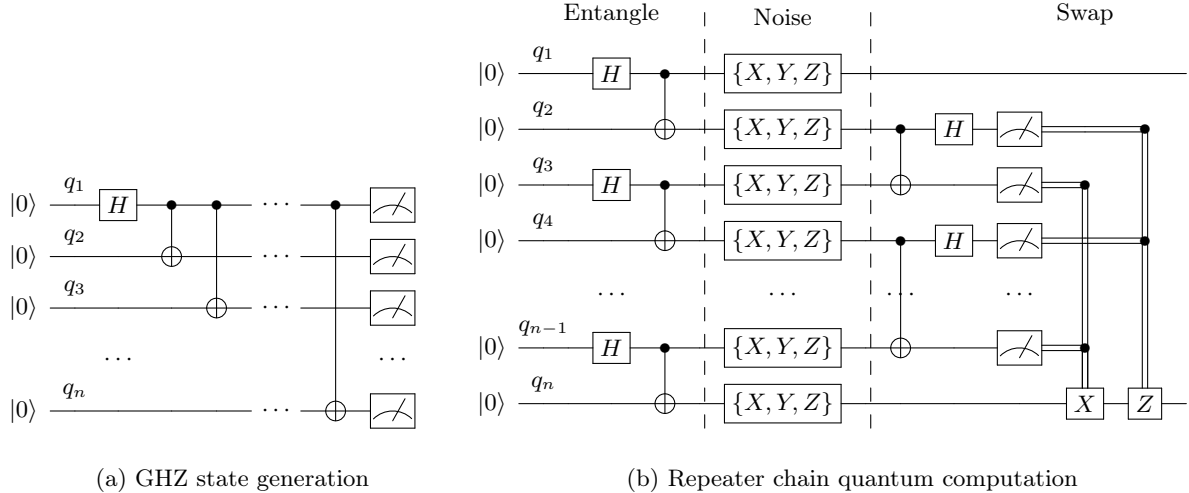
\begin{figure}[tb]
\centering
\begin{subfigure}[b]{0.39\textwidth}
\centering
\begin{equation*}
\Qcircuit @C=1em @R=.7em {
\lstick{\ket{0}} & \ustick{q_1}\qw & \gate{H} & \ctrl{1} & \ctrl{2} & \qw & \cdots & & \ctrl{6} & \meter \\
\lstick{\ket{0}} & \ustick{q_2}\qw & \qw      & \targ    & \qw      & \qw & \cdots & & \qw      & \meter \\
\lstick{\ket{0}} & \ustick{q_3}\qw & \qw      & \qw      & \targ    & \qw & \cdots & & \qw      & \meter \\
                 &     &          &          &          &     &        & &          & \\
                 &     & \cdots   &          &          &     &        & &          & \cdots \\
                 &     &          &          &          &     &        & &          & \\
\lstick{\ket{0}} & \ustick{q_n}\qw & \qw      & \qw      & \qw      & \qw & \cdots & & \targ    & \meter \\
}
\end{equation*}
\caption{GHZ state generation}
\label{fig:ghz-circuit}
\end{subfigure}
\begin{subfigure}[b]{0.59\textwidth}
\centering
\begin{equation*}
\Qcircuit @C=1em @R=.7em {
& & & \ustick{\text{Entangle}} & \barrier[0.2em]{9} & & \ustick{\text{Noise}} \barrier[0.2em]{9} & &
& & & \ustick{\text{Swap}} & & \\
\lstick{\ket{0}} & \ustick{q_1}\qw & \qw & \gate{H} & \ctrl{1} & \qw & \gate{\{X,Y,Z\}} & \qw & \qw      & \qw      & \qw    & \qw       & \qw       & \qw \\
\lstick{\ket{0}} & \ustick{q_2}\qw & \qw & \qw      & \targ    & \qw & \gate{\{X,Y,Z\}} & \qw & \ctrl{1} & \gate{H} & \meter & \cw       & \cctrl{2} & \\
\lstick{\ket{0}} & \ustick{q_3}\qw & \qw & \gate{H} & \ctrl{1} & \qw & \gate{\{X,Y,Z\}} & \qw & \targ    & \qw      & \meter & \cctrl{6} &           & \\
\lstick{\ket{0}} & \ustick{q_4}\qw & \qw & \qw      & \targ    & \qw & \gate{\{X,Y,Z\}} & \qw & \ctrl{4} & \gate{H} & \meter & \cw       & \cctrl{5} & \\
                 &                 &     &          &          &     &                  &     &          &          & \\
                 &                 &     & \cdots   &          &     & \cdots           &     & \cdots   &          & \cdots & \\
                 &                 &     &          &          &     &                  &     &          &          &        & \\
\lstick{\ket{0}} & \ustick{q_{n-1}}\qw & \qw & \gate{H} & \ctrl{1} & \qw & \gate{\{X,Y,Z\}} & \qw & \targ    & \qw      & \meter & \cctrl{1} & \\
\lstick{\ket{0}} & \ustick{q_n}\qw & \qw & \qw & \targ    & \qw & \gate{\{X,Y,Z\}} & \qw & \qw      & \qw      & \qw    & \gate{X}  & \gate{Z}   & \qw \\
}
\end{equation*}
\caption{Repeater chain quantum computation}
\label{fig:repchain-circuit}
\end{subfigure}
\caption{
    \textbf{Circuits used to benchmark quantum computation in the Results, section~\nameref{sec:benchmarking-qcomp}, for $n$ qubits.}
For panel (b) the CNOT control line crossing the ellipses represents multiple lines for $n > 6$ qubits,
following the pattern of $q_2$ and $q_3$. Similarly, the classical control lines represent an AND of the measurement results for $q_3$, $q_5$, \ldots, $q_{n-1}$ and $q_2$, $q_4$, \ldots, $q_{n-2}$ to determine the control of the $X$ and $Z$ gates, respectively.
The noise gates denoted by \{X,Y,Z\} cycle through the Pauli gates (see main text).
Note that this circuit always requires an even number of qubits. }
\label{fig:benchmarking-circuits}
\end{figure}

\subsection{Runtime profiling of a repeater chain simulation}
\label{appendix:repchain-benchmark}

The runtime profiling of NetSquid presented in the Results, section~\nameref{sec:benchmarking-profiling}, is performed for a simple repeater chain.
The network setup of this simulation extends the single repeater presented in Supplementary Figure~\ref{fig:netsquid_distil_repeat} to a chain of nodes by adding the entangling connection shown between each pair of neighbouring nodes.
Direct classical connections are connected between each node and one of the end-nodes, rather than between neighbouring nodes,
and are used to transmit the swapping corrections.
The chosen configuration for this network does not need to be physically realistic;
it suffices for it to be representative of the typical computational complexities.
The nodes are placed at 20km intervals and the channels transmit messages at the speed of light in fibre.
The entanglement sources, assumed to be perfect, are all synchronised and operate at a frequency of 100~kHz.
Physical non-idealities are represented by adding time-dependent depolarising noise to both the quantum channels and quantum memories, as well as dephasing noise to quantum gates.
The corresponding depolarising and dephasing times are $0.1$~s and $0.04$~s, which correspond to the $T_1$ and $T_2$ times presented in section~\nameref{sec:nv-qprocessor} of the Methods.

In a simulation run entanglement is created once between the end-nodes by performing entanglement swaps along the chain.
Protocols are assigned to all but the end-nodes to perform entanglement swaps after each round of entanglement generation,
and send their measurement results as corrections to the same end-node.
A protocol running on the end-node collects these corrections, and applies them if needed.

The runtime of this simulation is profiled to determine the distribution of time spent in the functions of NetSquid's sub-packages,
as well as its dependency packages NumPy and PyDynAA.
To perform this profiling the \texttt{cProfile} package is used.
The reported runtime for a given number of nodes is the mean of 400 successive simulation runs.

\section{Quantum switch: physical network and protocol}
\label{app:switch}

Here, we provide the details of the quantum switch simulations, whose results are presented in section~\nameref{sec:results-switch} of the Results.

We implement the model of Vardoyan et al. \cite{vardoyan2019stochastic}, for which the parameters of the simulation are:
\begin{itemize}
	\item the number of leaf nodes $k$;
	\item the desired size $n$ of the shared entanglement on the leaf nodes;
	\item for each leaf node: the rate $\mu$ at which bipartite entanglement is generated between leaf node and switch;
	\item $B$: the buffer size, i.e. the number of dedicated qubits per leaf node at the switch.
\end{itemize}
In addition, we include $T_2$, the memory coherence time.

\subsection{Physical network}
In the scenario we study, the quantum switch is the centre node of a star-topology network, with $k\geq 2$ leaf nodes.
Each leaf node individually is connected to the switch by a \textit{connection}, which consists of a \textit{source} producing perfect bipartite entangled states $(\ket{00} + \ket{11}) / \sqrt{2}$ on a randomised \textit{clock} and two \textit{quantum connections}, from the source to the leaf and switch node, respectively, for transporting the two produced qubits.
The interval $\Delta t$ between clock triggers is randomly sampled from an exponential distribution with probability $\mu \cdot e^{-\mu\cdot \Delta t}$ where $\mu$ is the rate of the source.
We set the delay of the quantum channels to zero.

Each node holds a single \textit{quantum processor} with enough quantum memory positions for the total duration of our runs.
Each memory position has a $T_2$ \textit{noise model}: if a qubit is acted upon after having resided in memory for time $\Delta t$, then a dephasing map (eq.~\eqref{eq:dephasing-channel}) is applied with dephasing probability $p = \frac{1}{2}\left( 1- e^{-\Delta t/T_2}\right)$.
Each quantum processor can perform any unitary operation or single-qubit measurement; these operations are noiseless and take no time.

\subsection{Protocol of the switch node}
The switch node continuously waits for incoming qubits.
Upon arrival of a qubit from leaf node $\ell$, the switch first checks whether it shares more entangled pairs of qubits with $\ell$ than the pre-specified buffer size $B$; if so, it discards the oldest of those pairs.
Then, it checks whether it holds entangled pairs with at least $n$ different leaves.
If so, then it performs and $n$-qubit GHZ-basis measurement (see below) on its qubits of those pairs.
If multiple groups of $n$ qubits from $n$ distinct nodes are available, then it chooses the oldest pairs.

Directly after completion of the GHZ-basis measurement, we register the measurement outcomes and obtain the resulting $n$-partite entangled state $\ket{\psi}$ on the leaf nodes.
From these, the fidelity $|\langle \psi | \phi_{\text{ideal}}\rangle|^2$ with the ideal target GHZ state $\ket{\phi_{\text{ideal}}}$ is computed.

The $n$-qubit GHZ states are
\begin{equation}
	\label{eq:ghz}
	\Big(\ket{0} \otimes \ket{b_2} \otimes \ket{b_3} \otimes \dots \otimes \ket{b_{n}}
	+ (-1)^{b_1}
	\ket{1} \otimes \ket{\overline{b_2}}\otimes \ket{\overline{b_3}} \otimes \dots \otimes \ket{\overline{b_{n}}}\Big)
	\Big/ \sqrt{2}
\end{equation}
where $b_j \in \{0, 1\}$ and we have denoted $\overline{b} = 1 - b$.
The $n$-qubit \textit{quantum program} that the switch node applies for performing a measurement in the $n$-qubit GHZ basis is as follows: first, a CNOT operation on qubits $1$ and $j$ (1 is the control qubit) is applied for all $j=2, 3, \dots, n$, followed by a Hadamard operation (eq.~\ref{eq:hadamard}) on qubit $1$.
Then, all qubits are measured in the $\ket{0}/\ket{1}$-basis.
If we denote the outcome of qubit $j$ as $b_j$, the GHZ-state that is measured is precisely the one in eq.~\eqref{eq:ghz}.

\section{Hardware parameters for the NV repeater chain}
\label{app:nv-physical-modelling}

Here, we provide the values for the hardware parameters of the nitrogen-vacancy setup used in our simulations.
An overview of all parameters is provided in Supplementary Table~\ref{table:parameters}, including two example sets of improved parameters following the approach in section~\nameref{sec:improvement-factor} of the Methods.

\begin{sidewaystable}
        \begin{tabular}{| l | c | c| c| c|c|}
\hline
		&Noise parameter & Duration/time & Probability &
		\multicolumn{2}{c|}{Improved noise param.}
		\\
		& (`near-term') && of no-error& $3\times$ & $10\times$
		\\\hline
Probability of double excitation $\pde$ (\ref{sec:double-excitation})
& 0.06
& -
&
$\pde$
&
$0.01$
&
$0.003$
\\\hline
Transmission loss $\gamma$ (dB/km, \ref{sec:detection})
& 0.2
& -
& $\cross$
&
$\cross$
&
$\cross$
\\\hline
Dark count probability $\pdc$ (\ref{sec:dark-counts})
&
		$2.5 \cdot 10^{-8}$
& -
& $1 - \pdc$
&

		$8.3 \cdot 10^{-8}$
&
		$2.5 \cdot 10^{-9}$
\\\hline
		Probability of photon detection (\ref{sec:detection}) &  &&&&\\
for zero-length fibre $\pdetnofibre$
& 0.0046
& -
&
$\pdetnofibre$
&
0.16
&
0.58
\\\hline
Interferometric phase uncertainty
		& $0.35$
& -
&
            $1 - p_{\text{phase}}$
		&
		0.20
		&
		0.11
\\
 $\sigma_{phase}$ (rad, \ref{sec:phase-uncertainty}) 
&&& (eq.~\eqref{eq:phase-dephasing-prob})&&\\\hline
Photon visibility $V$ (\ref{sec:visibility})
& 0.9
& -
&
$V$
&
0.97
&
0.99
\\\hline
$N_{1/e}$: indicates nuclear dephasing
&
1400
& -
&
$p$ from eq.~\eqref{eq:prob-carbon-dephasing}
&
4206
&
14006
\\
during electron initialization (\ref{sec:nuclear-dephasing})
&
&
&
&
&
\\\hline
Electron $T_1$
(\ref{app:nv-parameters-quantum-processor})
& -
& 1h
&
$e^{-1/T_1}$
&
2.8h
&
10h
\\\hline
Electron $T_2^*$
(\ref{app:nv-parameters-quantum-processor})
& -
& 1.46 s
& $e^{-1/T_2^*}$
&
4.4s
&
14.6s
\\\hline
Carbon $T_1$
(\ref{app:nv-parameters-quantum-processor})
& -
& 10h
&
$e^{-1/T_1}$
&
27h
&
100h
\\\hline
Carbon $T_2$
(\ref{app:nv-parameters-quantum-processor})
&
&
1s
&
$e^{-1/T_2}$
&
3s
&
10s
\\\hline
Carbon initialization to $\ket{0}$
(\ref{app:nv-parameters-quantum-processor})
&
$F$=0.997
&
310 $\mu$s&
$2F - 1$
&
$F=0.999$
&
$F=0.9997$
\\\hline
Carbon $Z$-rotation gate
(\ref{app:nv-parameters-quantum-processor})
&
$F$=0.999
&
20 $\mu$s
&
$4(F - 1)/3$
& 
$F>0.9999$
&
$F>0.9999$
\\\hline
E-C controlled-$R_X$-gate
&
$F_{EC}$=0.97
& 500 $\mu$s
& $(4\sqrt{F_{EC}} -1)/3$ 
&
$F_{EC}=0.990$
&
$F_{EC}=0.997$
\\
(electron=control)
(\ref{app:nv-parameters-quantum-processor})
&&&&&\\
\hline
Electron initialization to $\ket{0}$
(\ref{app:nv-parameters-quantum-processor})
&
$F$=0.99
&
2 $\mu$s
&
$2F - 1$
&
$F=0.997$
&
$F=0.999$
\\\hline
Electron single-qubit gate
(\ref{app:nv-parameters-quantum-processor})
&
$F$=1& 5 ns
&
$(4F-1)/3$
&
$F=1$
&
$F=1$
\\\hline
Electron readout (eq.~\eqref{eq:nv-readout} and sec.~\ref{app:nv-parameters-quantum-processor}
) &
$0.95/0.995 (f_0/f_1)$ &
3.7 $\mu$s 
& $f_x$
&
$0.983/0.9983$
&
$0.995/0.9995$
\\\hline
    \end{tabular}
    \caption{
        \textbf{Physical parameters dealing with elementary link generation, memory coherence times and duration and fidelities ($F$) of the gates.}
	    Depicted are both parameters of the dataset `near-term' and two examples of improved parameter sets (see Methods, section~\nameref{sec:improvement-factor}), for $3$ times and $10$ times improved, respectively, together with the function to convert the parameter to a `probability of no-error' to compute the improved parameter value for other factors.
	    The `near-term' values correspond to $1\times$ improvement.
The transmission loss parameter $\gamma$ is not changed by the improvement procedure and equals $\gamma = 0.2$ dB/km during any of our simulations.
\label{table:parameters}
}
\end{sidewaystable}

\subsection{Parameters for elementary link generation}
\label{app:nv-parameters-elementary-link-generation}

For generating entanglement between the electron spins of two remote NV centres in diamond, we simulate a scheme based on single-photon detection, following its experimental implementation in~\cite{humphreys2018deterministic}.
The setup consists of a middle station which is positioned exactly in between two remote NV centres in diamond.
The middle station is connected to the two NVs by glass fibre and contains a 50:50 beam splitter and two non-number resolving photon detectors.
In the single-photon scheme, each NV performs the following operations in parallel.
First, the electron of each NV system is brought into the state $\sqrt{\alpha} \ket{0} + \sqrt{1 - \alpha} \ket{1}$ by optical and microwave pulses, where $\alpha$ is referred to as the bright-state parameter.
Then, a laser pulse triggers the emission of a photon, yielding the spin-photon state $\sqrt{\alpha} \ket{0}_s \otimes\ket{1}_p + \sqrt{1 - \alpha} \ket{1}_s \otimes \ket{0}_p$, where $\ket{0}$ ($\ket{1}$) denotes absence (presence) of a photon.
We set $\alpha = 0.1$ since for that value, fidelity is approximately maximal at lab-scale distances \cite{humphreys2018deterministic}; optimising over $\alpha$ is out of the scope for this work.
We assume that the delay until emission of the photon is fixed at $3.8$ $\mu$s \cite{hermans2020privatecommunication}.

From each NV centre, the emitted photons are transmitted to the middle station through glass fibre, where a 50:50 beam splitter effectively erases the which-way information of an incoming photon.
An attempt at generating entanglement using this single-click scheme is declared successful if precisely one of the detectors clicks, which happens if either (a) a single photon arrives at the detector and the other does not or (b) both photons arrive (in case (b), only a single detector clicks due to the Hong-Ou-Mandel effect).
Case (a) yields the generation of the spin-spin state $\ket{\phi_{\pm}} = (\ket{01} \pm \ket{10})/\sqrt{2}$, where $\pm$ indicates which of the two detectors clicked, while case (b) results in $\dyad{00}$.
Given that a single photon arrives, the probabilities that the other photon has or has not arrived are respectively $1 - \alpha$ and $\alpha$ (in the absence of loss).
Therefore, a successful attempt results in the generation of the spin-spin state $(1 - \alpha) \dyad{\phi_{\pm}} + \alpha \dyad{00}$.
We refer to \cite{humphreys2018deterministic} for a more in-depth description of the scheme.
We assume that the speed of the photons and of all classical communication equals $c / n_{ri}$, where $c$ is the speed of light in vacuum and $n_{ri} = 1.44$ is the refractive index of glass \cite{paschotta2020fibers}.
\\

In reality, however, several sources of noise affect the produced state, which we treat below.

\subsubsection{Imperfect detection}
\label{sec:detection}

The total probability $\pdet$ that a photon, emitted by the NV, will be detected in the midpoint is given by the product of four probabilities \cite{rozpedek2018near-term}
\begin{itemize}
\item{
		the probability $p_{\text{zero\_phonon}}$ that the photon frequency is in the zero-phonon line \cite{riedel2017deterministic};
}
\item{
the probability $p_{\text{collection}}$ that the photon is collected into the glass fibre;
}
\item{the probability $p_{\text{transmission}}$ that the photon does not dissipate in the fibre during transmission;}
\item{the probability $p_{\text{detection}}$ that the photon is detected, conditioned on the fact that it reaches the detector.}
\end{itemize}
Thus we can write
\begin{equation}
\label{eq:pdet-computation}
\pdet = \pdetnofibre \cdot p_{\text{transmission}}
\end{equation}

where
\begin{equation}
\pdetnofibre = p_{\text{zero\_phonon}} \cdot p_{\text{collection}} \cdot p_{\text{detection}}
.
\end{equation}
The transmission probability is given by
\[
p_{\text{transmission}} = 10^{-(L/2) \cdot \gamma / 10}
\]

where $L$ is the internode distance (i.e.\ $L/2$ is the length of the fibre from NV to middle station) and $\gamma$ is the loss parameter which depends on the photon frequency.
In our simulations, we assume that the photon frequency is converted to the telecom frequency, corresponding to $\gamma = 0.2 $ dB/km.
Also, we assume that the emission in the zero-phonon line is enhanced by an optical cavity from $p_{\text{zero\_phonon}} = 3$\% (without cavity) to $p_{\text{zero\_phonon}} = 46$\% (with cavity) \cite{riedel2017deterministic}.
We set the detection efficiency $p_{\text{detection}}$ to 0.8 \cite{hensen2015loophole}.

What remains is the collection efficiency $p_{\text{collection}}$, which we compute from experimental values (no cavity, no conversion to the telecom frequency) using eq.~\eqref{eq:pdet-computation} with $L=2$m, $\pdet = 0.001$ \cite{hermans2020privatecommunication}
and $\gamma = 5$ dB/km \cite{dahlberg2019linklayer} (for the zero-photon line frequency), yielding $p_{\text{collection}} = 0.042$.  
Since frequency conversion to the telecom frequency is a probabilistic process and only succeeds with probability $30\%$ \cite{zaske2012visible}, we set $p_{\text{collection}} = 0.3 \cdot 0.042$.

\subsubsection{Other sources of noise}
\label{sec:dark-counts}
\label{sec:visibility}
\label{sec:double-excitation}
\label{sec:phase-uncertainty}
Other sources of noise on the freshly generated electron-electron entanglement are

\begin{itemize}
\item{
Dark counts:
		a photon detector falsely registering.
		The dark count probability follows a Poisson distribution $\pdc = 1 - e^{-t_{\text{w}} \cdot \lambda_{\text{dark}}}$ where $t_{\text{w}} = 25$ ns~\cite{humphreys2018deterministic} is the duration of the time window at the midpoint. We set $\lambda_{\text{dark}} = 1$ Hz as the dark count rate.
}
\item{
Imperfect photon indistinguishability.
The generation of entanglement at the middle station is based upon the erasure of the which-way information with respect to the path of the photons.
Only in case the photons are fully indistinguishable, the which-way information is erased perfectly.
		The overlap of the photon states is given by the visibility $V$, which we set to 0.9 \cite{humphreys2018deterministic}.
}
\item{Double excitation of the electron spin.
When triggered to emit a photon by a resonant laser pulse, an NV centre could be excited twice, which results into the emission of two photons.
		We set its occurrence probability to $\pde = 0.06$ \cite{hermans2020privatecommunication}.
		}
\item{Photon phase uncertainty.
The photons which interfere at the midpoint acquired a phase during transmission and a difference of these phases influences the precise entangled state that is produced~\cite{kalb2017entanglement}.
        Given a standard deviation $\sigma_{\text{phase}} = 0.35$ rad \cite{hermans2020privatecommunication}
		of the acquired phase, we compute the dephasing probability as \cite{humphreys2018deterministic}
\begin{equation}
\label{eq:phase-dephasing-prob}
    p_{\text{phase}} = \frac{1}{2} \left(1 - e^{-\sigma_{\text{phase}}^2 / 2}\right)
	.
\end{equation}
		}
\end{itemize}

\subsubsection{Nuclear spin dephasing during entanglement generation}
\label{sec:nuclear-dephasing}

The initialisation of the electron spin state induces dephasing of the carbon spin states through their hyperfine coupling.
Following~\cite{kalb2018dephasing}, we model this uncertainty by a dephasing channel for each attempt with dephasing probability
\begin{equation}
\label{eq:nuclear-dephasing-single-attempt}
	p_{\text{single}} = \frac{1}{2} (1 - \alpha) \cdot (1 - e^{-C_{\text{nucl}}^2 / 2}).
\end{equation}
The parameter $C_{\text{nucl}}$ is the product of the coupling strength between the electron spin and the carbon nuclear spin, and an empirically determined decay constant.
Rather than expressing the dephasing probability as function of $C_{\text{nucl}}$, we express the magnitude of nuclear dephasing as $N_{1/e}$, the number of electron spin pumping cycles after which the Bloch vector length of a nuclear spin in the state $(\ket{0} + \ket{1}) / \sqrt{2}$ in the $X-Y$ plane of the Bloch sphere has shrunk to $1 / e$, when the electron spin state has bright-state parameter $\alpha = 0.5$ (i.e the electron spin is in the state $(\ket{0} + \ket{1})/ \sqrt{2}$).

Let us compute how $p_{\text{single}}$ depends on $N_{1/e}$ instead of on $C_{\text{nucl}}$.
First, we find by direct computation that the equatorial Bloch vector length of a state is shrunk by a factor $1 - 2 p$ after a single application of the single-qubit dephasing channel (eq.~\eqref{eq:dephasing-channel}) with probability $p$ ($p \leq \frac{1}{2}$).

Equating $(1 - 2 p)^{N_{1/e}} = 1 / e$  yields 
\begin{equation}
	\label{eq:prob-carbon-dephasing}
p = \frac{1}{2} \left(1 - e^{-1 / N_{1/e}}\right)
.
\end{equation}
Equating $p_{\text{single}}$ from eq.~\eqref{eq:nuclear-dephasing-single-attempt} with $\alpha = 0.5$ and $p$ from eq.~\eqref{eq:prob-carbon-dephasing}, followed by solving for $C_{\text{nucl}}$ yields
\[
	1 - e^{-C_{\text{nucl}}^2 / 2} = 2 \left(1 - e^{-1 / N_{1/e}}\right)
	.
	\]
Substituting back into eq.~\eqref{eq:nuclear-dephasing-single-attempt} yields an expression for general $\alpha$:
\begin{equation}
\label{eq:nuclear-dephasing-single-attempt-number-of-attempts}
	p_{\text{single}} = (1 - \alpha) \left(1 - e^{-1 / N_{1/e}}\right)
	.
\end{equation}
We set $N_{1/e} = 1400$ \cite{beukers2019masterthesis}.
\\

\subsubsection{Local processing parameters}
\label{app:nv-parameters-quantum-processor}

For the dynamics of the electron spin, we use $T_1 = 1$ hour and $T_2^* = 1.46$s \cite{abobeih2018one}.
For the carbon nuclear spin, we take $T_1 = 10$ hours and $T_2 = 1$ s (experimentally realised: $T_1 =6$m and $T_2 \approx 0.26-25$s \cite{bradley2019tenqubit}).
For the noise of the controlled-$R_X$ gate (Methods, section~\nameref{sec:nv-qprocessor}), we set the depolarising probability $p=0.02$ (denoted as $p_{EC}$ in Supplementary Table~\ref{table:parameters}), since by simulation of the circuit~\cite[Fig.~2a]{kalb2017entanglement}, we find that this value agrees with the experimentally found effective circuit fidelity of $0.95$.
The corresponding fidelity of the gate is $F_{EC} = (1 - 3 p_{EC} / 4)^2 = 0.97$.
The initialisation fidelities of the electron and carbon spins are set at $0.99$ \cite{reiserer2016robust} and $0.997$ \cite{bradley2019tenqubit}.
We use $0.999$ for the carbon $Z$-rotation gate fidelity (experimentally achieved: $1$ \cite{taminiau2014universal}).
The durations of local operations are identical to our earlier simulations (see Appendix D, Table 6 in \cite{dahlberg2019linklayer} and references therein).
We summarise all hardware values in Supplementary Table~\ref{table:parameters}.

\section{Protocols and quantum programs for the NV repeater chain}
\label{app:nv-protocols}

Here, we first elaborate on the sequence of quantum operations and classical communication that the NV protocol building blocks consist of (Methods, section~\nameref{sec:nv-protocols}).
Then, we describe in detail the two repeater chain protocols we simulated.

\subsection{Operations for the building blocks: \move, \unmove, \distill\ and \swap}

\move~is the mapping of the electron spin state onto a free nuclear spin.
The operation requires the nuclear spin state to be $\ket{0}$ and the circuit, given in Supplementary Figure~\ref{fig:nv-circuits}(a), performs the following mapping:
\begin{equation}
\label{eq:move-operation}
\ket{\phi}_e \otimes \ket{0}_n \mapsto \ket{0}_e \otimes \left(H \ket{\phi}_n\right)
\end{equation}
where $\ket{\phi}$ is an arbitrary single-qubit state and 
\begin{equation}
	\label{eq:hadamard}
	H := \frac{1}{\sqrt{2}} \left(\dyad{0} + \dyad{0}{1} + \dyad{1}{0} - \dyad{1}\right)
\end{equation}
is the Hadamard gate.
By \unmove~(Supplementary Figure~\ref{fig:nv-circuits}(b)), we denote the reverse operation,
\[
\ket{0}_e \otimes \ket{\phi}_n \mapsto \left(H \ket{\phi}_e\right) \otimes \ket{0}_n
.
\]

We simulate the specific entanglement distillation protocol (\distill) from Kalb et al.~\cite{kalb2017entanglement}, which acts upon an electron-electron state and a nuclear-nuclear state to probabilistically increase the quality of the nuclear-nuclear state, at the cost of having to read out the electron-electron state.
In the protocol, the two involved nodes each perform a sequence of local operations including a measurement (Supplementary Figure~\ref{fig:nv-circuits}(c)), followed by communicating the measurement outcome from the circuit to each other.
In this work, we only use distillation in one of the two repeater schemes we consider (\withdistill) and in that case, the success condition is as follows: if the nodes are adjacent, then the measurement outcomes should both equal $0$ (i.e. the bright state of the electron), while otherwise the measurement outcomes only need to be equal (i.e. both $0$ or both $1$). 
In the case of failure, the nuclear-nuclear state is considered lost.

The entanglement swap (\swap) converts two short-distance entangled qubit pairs $A-M$ and $M-B$ into a single long-distance one $A-B$, where $A, B$ and $M$ are nodes.
It is equivalent to performing quantum teleportation \cite{bennett1993teleporting} to a qubit which is part of a larger remote-entangled state.
Our entanglement swapping protocol at node $M$ starts by assuming that one of $M$'s qubits which is involved in the entanglement swap is the electron spin.
Then, a series of local operations including measurements (Supplementary Figure~\ref{fig:nv-circuits}(d)) is performed; the measurement outcomes are transferred to both $A$ and $B$.
In the original teleportation proposal, $B$ performs a local operation to correct the state $A-B$ to the expected one.
However, due to the fact that such correction operation is generally not directly possible to perform on the nuclear spin state (see the allowed operations in section~\nameref{sec:nv-qprocessor} of the Methods), we opt for the approach where the correction operation is tracked in a classical database.
Details of this tracking, including how it affects the entanglement distillation and entanglement swap protocols, are given in Supplementary Note \ref{app:nv-pauli-frame}.

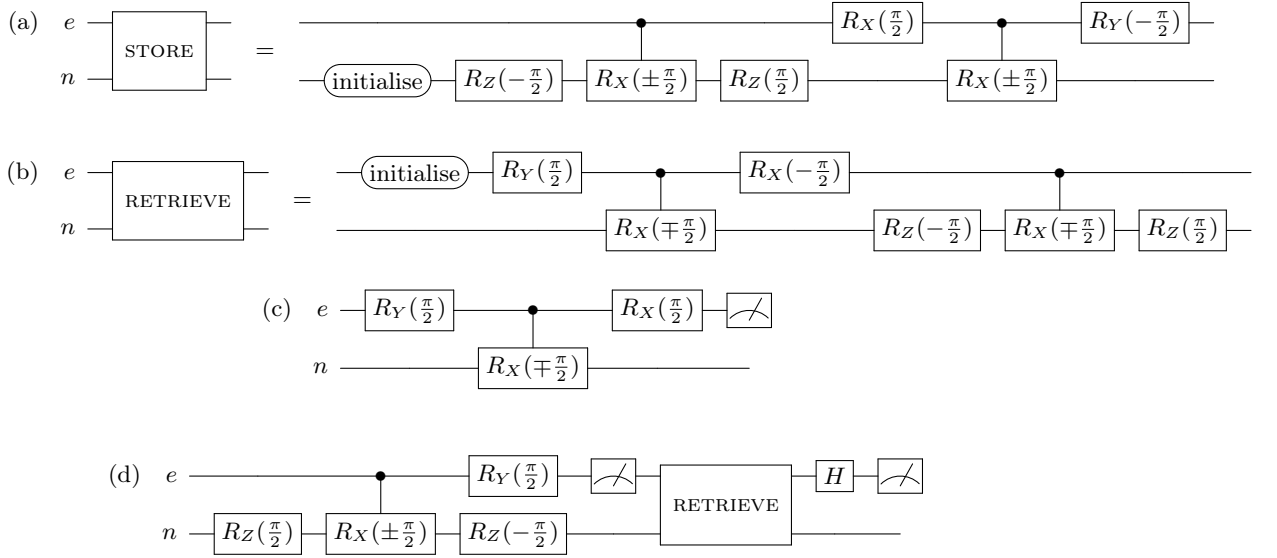
\begin{figure*}
\centering
\begin{subfigure}[b]{0.75\textwidth}
\centering
\[
\Qcircuit @C=1em @R=1.4em {
	\lstick{\textnormal{(a)} \quad e}
	&\multigate{1}{\move} & \qw \\
	\lstick{n}
	&\ghost{\move} & \qw
	}
	\quad
	\raisebox{-2.7ex}[-10pt][-10ex]{$=$}
	\quad
\Qcircuit @C=1em @R=.7em {
	& \qw & \qw          & \ctrl{1} & \qw     & \gate{R_X(\frac{\pi}{2})} & \ctrl{1}  & \Yminuspitwo & \qw\\
	&	\measure{\textnormal{initialise}}
	& \Zminuspitwo & \CXDIR   & \Zpitwo & \qw                       & \CXDIR    & \qw     & \qw
}
\]
\vspace{1em}
\[
\Qcircuit @C=1em @R=1.4em {
	\lstick{\textnormal{(b)} \quad e}
	&\multigate{1}{\unmove} & \qw \\
	\lstick{n}
	&\ghost{\unmove} & \qw
	}
	\quad
	\raisebox{-2.7ex}[-10pt][-10ex]{$=$}
	\quad
\Qcircuit @C=1em @R=.7em {
	&	\measure{\textnormal{initialise}}	& \Ypitwo & \ctrl{1}       & \gate{R_X(-\frac{\pi}{2})} & \qw     & \ctrl{1} & \qw   & \qw\\
	&\qw & \qw     & \CXDIRminus   & \qw                        & \Zminuspitwo & \CXDIRminus  & \Zpitwo & \qw
}
\]
\vspace{1em}
\[
\Qcircuit @C=1em @R=.7em {
	\lstick{\textnormal{(c)} \quad e}
	& \Ypitwo & \ctrl{1}       & \gate{R_X(\frac{\pi}{2})} & \meter\\
	\lstick{n}
	& \qw     & \CXDIRminus   & \qw                        & \qw
}
\]
\vspace{1em}
\newline
\[
\Qcircuit @C=1em @R=.7em {
	\lstick{\textnormal{(d)}\quad e} & \qw     & \ctrl{1}          & \Ypitwo     & \meter{} &  \multigate{1}{\unmove} & \gate{H} & \meter{} \\
\lstick{n} & \Zpitwo & \CXDIR & \Zminuspitwo & \qw                     & \ghost{\unmove} & \qw & \qw
}
\]
\end{subfigure}
\caption{
\label{fig:nv-circuits}
\textbf{Quantum circuits used in simulations of the NV repeater chain, acting on an electron ($e$) and nuclear ($n$) spin.}
Figure depicts the quantum circuit for the NV repeater protocol building blocks: (a) \move~operation (mapping electron spin state onto the nuclear spin), (b) \unmove~(reverse operation to \move), (c) entanglement distillation, (d) entanglement swap.
}
\end{figure*}

\subsection{Repeater chain protocols \swapasap~and \withdistill}
\label{sec:nv-repchain-protocol}

We describe two protocols for the NV repeater chain: \swapasap, a protocol where a node performs an entanglement swap as soon as it holds two entangled pairs, one in each direction of the chain, and \withdistill, a nested protocol with distillation at each nesting level which is based on the BDCZ protocol~\cite{briegel1998quantum}.
Both protocols run asynchronously on each node.

\def\request{{\sc send-request}}
\def\sleep{{\sc sleep}}
\def\waituntilconfirmation{{\sc wait-until-confirmation}}
\def\confirm{{\sc send-confirmation}}
In both protocols, a node remains idle until it is triggered to check whether it should perform an action.
It is triggered at the following three moments: (a) at the start of the simulation, (b) when the node receives a classical message (if a node is busy upon reception, the message is stored and responded to later)
, (c) when its previous action is finished.
A simulation run finishes as soon as the two end nodes share a single entangled pair of qubits.

For the \swapasap~protocol, the sequence of operations that a node performs depends on its index in the chain (start counting from left to right, nodes have indices $1, 2, 3, \dots$).
If the index of the node is even, the node sends a request for \entgen~to its left neighbour, and starts the operation as soon as it has received confirmation. After performing \move~to free the electron spin, it repeats it for its right neighbour.
Odd-indexed nodes remain idle until reception of an \entgen~request, after which they perform \move~if necessary to free the electron, send a confirmation, sleep for the duration of the message transmission, followed by performing \entgen.
Once a node has entanglement with both directions, it performs a \swap~and sends the outcome to the end nodes.

The two end nodes are exceptions to the above.
The left end node (i.e. with index $1$) behaves like an odd-indexed node, but without performing \swap.
The same holds for the rightmost node (i.e. the node with the largest index), unless its index is even, in which case it initiates and performs entanglement generation with the adjacent node on the left.

The \withdistill~protocol is a variant of the BDCZ protocol \cite{briegel1998quantum}, adapted to the fact that an NV cannot perform multiple \entgen~, \distill~or \swap~operations in parallel due to its restricted topology (Methods, section~\nameref{sec:nv-qprocessor}).
In the adapted version, nodes take the role of \textit{initiator} of one of the three main actions (\entgen, \distill, \swap) if the action occurs at the highest nesting level that this node belongs to.
To be precise, we do the following.
In a repeater chain with $2^n + 1$ nodes, denote by $\{0, 1, 2, \dots, 2^n\}$ the indices of the nodes from left to right.
A node (not an end node) with index $k \in \{1, 2, 3, \dots, 2^n - 1\}$ initiates an action only if the entanglement that is involved in the task spans precisely $f_n (k)$ segments, where
\begin{eqnarray*}
	f_1(k) &=& 1 \text{ for all $k$}\\
	f_n (k) &=&
  \begin{cases}
	  f_{n-1} (k) & \text{if } k < 2^{n - 1}, \\
	  f_{n-1} (2^{n} - k) & \text{if } k > 2^{n - 1}, \\
	  2^{n - 1} & \text{if } k = 2^{n - 1}.
  \end{cases}
\end{eqnarray*}

End nodes are never initiators.

When a node (index $k$) is triggered, it performs the following checks in order and performs the first action for which the check holds true:

\begin{enumerate}
	\item{
			If it shares entangled pairs with nodes $k - f_n(k)$ and $k + f_n(k)$, and both are the immediate result of successful distillation: perform \swap~and send the measurement outcomes to the involved nodes
}
\item{
		\label{item:distill-left}
		If it holds two entangled pairs with node $k - f_n(k)$ and neither pair is the result from successful entanglement distillation: send a request to distill to the node, wait for confirmation, followed by performing \distill
}
\item{
Same as~\ref{item:distill-left}, but now for \distill~on the right, i.e. remote node has index $k + f_n(k)$
}
\item{
		\label{item:request-response}
If there are any request-messages that have not been responded to yet: pick the oldest one and act as follows.
Respond to the message with a confirmation message, followed by sleeping for the time that the confirmation takes to arrive at the remote node. Then perform the requested action (\entgen~or \distill).
}
\item{
		\label{item:entgen-left}
If $f_n(k) = 1$ and the node does not hold entanglement with its immediate left neighbour that is the result of successful entanglement distillation: send a request for \entgen~to the node, wait for confirmation, followed by performing \entgen.
}
\item{
		Same as \ref{item:entgen-left} for right adjacent node.
		}
\end{enumerate}

If no action follows from the checks above, then the node remains idle until the next time at which it is triggered.
In the operations above, if necessary \entgen~is preceded by \move~to free the electron by storing its state into a free carbon spin. \distill~is preceded by a combination of \move~and \unmove~to ensure the correct state lives on the electron spin, and so is \swap~in case neither of the two to-be-swapped qubits live on the electron.
Since end nodes are never initiators, they only check~\ref{item:request-response}.

\section{Tracking of correction operations in the NV repeater chain}
\label{app:nv-pauli-frame}

\newcommand{\bellstate}[2]{\ket{\phi[#1, #2]}}
\def\poslose{\pi_{\textnormal{lose}}}
\def\poskeep{\pi_{\textnormal{keep}}}
\def\Plose{P_{\textnormal{lose}}}
\def\Pkeep{P_{\textnormal{keep}}}

Here, we explain how nodes of the NV repeater chain track the precise entangled state they hold.
This is done by associating unitary operations to each qubit, which map the state of two remotely entangled qubits back to $(\ket{01} + \ket{10}) / \sqrt{2}$ in the ideal case.
Tracking these unitaries (gates) in a classical database, instead of performing them on the (imperfect) quantum hardware, has the advantage of avoiding gate noise.
This argument is even stronger for NV centres in case the remote-entangled state is held by a carbon nuclear spin, because direct application of a correction operator to a carbon spin is generally not possible due to the restricted topology of the NV quantum processor (Methods, section~\nameref{sec:nv-qprocessor}).
Thus, performing the correction operator to the nuclear spin requires even more gates, namely the ones to map the nuclear spin to the electron spin (the \unmove~operation, see section~\nameref{sec:nv-protocols} of the Methods), where the correction operator could be applied.

In what follows, we first explain how we track the correction operations.
Then, we describe how the tracking changes the protocol building blocks from section~\nameref{sec:nv-protocols} of the Methods and subsequently prove the correctness of the tracking in the ideal case.

Let us denote the four Bell states as
\begin{align*}
	\bellstate{\pm 1}{1}
	&=
	(\ket{00} \pm \ket{11}) / \sqrt{2},
\\
	\bellstate{\pm 1}{-1}
	&=
	(\ket{01} \pm \ket{10}) / \sqrt{2} 
	.
\end{align*}

To each of the qubits it holds, a node associates a single-qubit Pauli operator $\unit, X, Y$ or $Z$, which are defined as
\[
	\unit = \dyad{0} + \dyad{1}, Z = \dyad{0} - \dyad{1},
	X = \dyad{0}{1} + \dyad{1}{0},
	Y = -i \dyad{0}{1} + i \dyad{1}{0}
	.
	\]
The goal of the tracking is, at any time during the simulation, for any two nodes $A$ and $B$ sharing electron-electron entanglement, that the target electron-electron state equals
\begin{equation}
	\label{eq:tracking-invariance-1}
	\ket{\psi} \equiv (P_A \otimes P_B)\bellstate{1}{-1}
	.
\end{equation}
Here, $P_A$ and $P_B$ denote the Pauli correction operations of node $A$ or $B$, respectively, and $\equiv$ denotes equality modulo a complex number of norm 1.

In what follows, it will be more convenient to use the following equivalent statement to eq.~\eqref{eq:tracking-invariance-1}:
\begin{equation}
	(P_A \otimes P_B) \ket{\psi} \equiv \bellstate{1}{-1}
	\label{eq:tracking-invariance}
	.
\end{equation}

\subsection{Tracking correction operators during the NV repeater chain protocol}
\label{sec:tracking-protocol-blocks}

Here, we explain how each of the four protocol building blocks from section~\nameref{sec:nv-protocols} of the Methods are adjusted to ensure that eq.~\eqref{eq:tracking-invariance} holds after the operations \entgen, \distill~and \swap.
\\

\textit{Entanglement generation.}
Suppose that nodes $A$ and $B$ perform the \entgen~protocol.
In the absence of noise, this protocol (approximately) produces the state $\bellstate{\pm 1}{-1}$, where $\pm$ denotes which detector clicked (Methods, section~\nameref{sec:generation-magic}).
If the $+$-detector clicked, then the state that A and B hold is the desired state $\bellstate{1}{-1}$, so we set $P_A = P_B = \unit$.
If the other detector clicked, then the produced state is $\bellstate{-1}{-1}$.
Therefore, one of the nodes (for example, the one with the higher position index in the chain) sets the correction operator to $Z$, whereas the other sets it to $\unit$, since $(\unit \otimes Z) \bellstate{-1}{-1} = \bellstate{1}{-1}$.
\\

\textit{Storing and retrieving qubits.}
Locally mapping the state of a qubit onto a different memory position by the \move~or \unmove~circuits does not alter the correction Pauli corresponding to that qubit.
\\

\textit{Entanglement distillation.}
Suppose that nodes $A$ and $B$ wish to perform the \distill~protocol, which starts by $A$ and $B$ sharing an electron-electron pair (correction Paulis $P_A^e$ and $P_B^e$ at node $A$ and $B$, respectively) and a nuclear-nuclear pair ($P_A^n$ and $P_B^n$).
In the protocol, first both nodes apply $P^n \cdot P^e$ to their electron spin qubit.
Then, both nodes locally perform the distillation circuit from Supplementary Figure~\ref{fig:nv-circuits}(c), followed by sending both the measurement outcome and $P^n$ to the other node.
The nodes determine whether the distillation succeeded using the condition explained in section~\ref{app:nv-protocols}.
In case of failure, the nuclear-nuclear state is discarded.
In case of success, one of the nodes in the chain (for example, the one with the lower position index in the chain) sets $P^n = \unit$, while the other sets
\begin{equation}
P^n =
\begin{cases}
	Y & \text{if $P_A^n \in \{X, Y\}$ and $P_B^n \in \{\unit, Z\}$}
	\\
	Y & \text{if $P_A^n \in \{\unit, Z\}$ and $P_B^n \in \{X, Y\}$}
	\\
	\unit & \text{otherwise}
	.
\end{cases}
\end{equation}
Below, in section~\ref{sec:tracking-proof-distill} of this Supplementary Note, we show that after this procedure, eq.~\eqref{eq:tracking-invariance} still holds.
\\

\textit{Entanglement swapping.}
Suppose that node $M$ wants to execute the \swap~protocol on shared pairs $A-M$ and $M-B$, with nodes $A$ and $B$ respectively.
We denote $M$'s correction Paulis as $P_M^A$ and $P_M^B$.
First, $M$ performs the Bell-state measurement circuit from Supplementary Figure~\ref{fig:nv-circuits}(d).
Let us denote the individual measurement outcomes of the circuit as $m_{\textnormal{earlier}}$ and $m_{\textnormal{later}}$ (both take values from $\{1, -1\}$), which correspond to the measured Bell state $\bellstate{a}{b}$ with $a=-1\cdot m_{\textnormal{earlier}}\cdot m_{\textnormal{later}}$ and $b=m_{\textnormal{later}}$.
Then, $M$ sends the Pauli $\unit$ to $A$, while to $B$ it sends $P_M^A \cdot P_M^B \cdot Q$, where $Q$ is given by
\begin{equation}
			\begin{tabular}{|c | c|}
				\hline
				$(a, b)$ & $Q$ \\\hline\hline
				$(1, 1)$ & $\X$\\\hline
				$(1, -1)$ & $\unit$\\\hline
				$(-1, 1)$ & $\Y$\\\hline
				$(-1, -1)$ & $\Z$\\\hline
			\end{tabular}
\label{table:pauli-correction}
\end{equation}
Both nodes $A$ and $B$ multiply their local Pauli with the Pauli they received from $M$.
The proof that after the swap, eq.~\eqref{eq:tracking-invariance} still holds can be found below in section~\ref{sec:tracking-proof-swap} of this Supplementary Note.

\subsection{Correctness proof of the correction operator update for \distill}
\label{sec:tracking-proof-distill}

Here, we prove that eq.~\eqref{eq:tracking-invariance} holds for the states that are outputted by the protocols for entanglement distillation and swapping explained above.

Let us start with entanglement distillation.
For this, we denote by `physical nuclear-nuclear state' the joint state of the nuclear spins of node $A$ and $B$.  
By direct computation, one can show the following.

\begin{prop}
	\label{prop:distillation-circuit-output-physical}
		Suppose that nodes $A$ and $B$ share the state $\bellstate{a}{b}$ on the electrons and the physical nuclear-nuclear state $\bellstate{c}{d}$, where $a, b, c, d \in \{1, -1\}$. When both nodes execute the distillation circuit from Supplementary Figure~\ref{fig:nv-circuits}(c), the resulting state on the carbon nuclear spins is 
	\[ \bellstate{c}{-a\cdot c\cdot d}\]
	and the measurement outcome $m_1 \in \{1, -1\}$ on one side is uniformly random, while the outcome of the other node equals $m_2 = m_1 \cdot b \cdot c$.
\end{prop}

We emphasise that using the correction-operator tracking for the \move~and \unmove~operations as described in section~\ref{sec:tracking-protocol-blocks} of this Supplementary Note, the physical nuclear-nuclear state between any two nodes does not satisfy eq.~\eqref{eq:tracking-invariance}.
The reason for this is that the \move~operation maps the electron spin state to the nuclear spin in a rotated basis, where the rotation operator is a Hadamard gate $H$ (eq.~\ref{eq:move-operation}).
However, the correction operators are not updated when the \move~is applied (see `Storing and retrieving qubits' in section~\ref{sec:tracking-protocol-blocks}).
Consequently, if nodes $A$ and $B$ share the physical nuclear-nuclear state $\ket{\psi}$, then mapping $\ket{\psi}$ to the reference Bell state $\bellstate{1}{-1}$ requires first the application of $H\otimes H$, followed by applying $P_A\otimes P_B$.
By `virtual nuclear-nuclear state', we mean the state $\ket{\psi'} = (H\otimes H)\ket{\psi}$, i.e. the state that satisfies eq.~\eqref{eq:tracking-invariance}.
Let us first convert Prop.~\ref{prop:distillation-circuit-output-physical} to a statement with the virtual-virtual nuclear state.

\begin{prop}
	\label{prop:distillation-circuit-output-virtual}
	Suppose nodes $A$ and $B$ share the electron-electron state $\bellstate{a}{b}$ and the virtual nuclear-nuclear state $\bellstate{c}{d}$. Then after the distillation circuit from Supplementary Figure~\ref{fig:nv-circuits}(c), the virtual state on the nuclear spins after performing the distillation equals
	\[
		\bellstate{-a\cdot c \cdot d}{d}
		\]
	and the measurement outcomes are $m_1 \in \{1, -1\}$ (uniformly random) and $m_2 = m_1 \cdot b \cdot d$.
\end{prop}

\begin{proof}
	The virtual nuclear-nuclear state and the physical one are related by $H\otimes H$.
	It is not hard to see that $H\otimes H \bellstate{x}{y} = \bellstate{y}{x}$ for any $x, y \in \{1, -1\}$.
	Applying this to Prop.~\ref{prop:distillation-circuit-output-physical} results in the measurement outcomes $m_1$ (uniformly random) and $m_2 = m_1 \cdot b \cdot d$ and resulting physical nuclear-nuclear state $\bellstate{d}{-acd}$.
	Obtaining the virtual state is done by applying $H\otimes H$ again, which yields $\bellstate{-acd}{d}$.
\end{proof}

Using Prop.~\ref{prop:distillation-circuit-output-virtual}, it is straightforward to check that the output state of the distillation protocol from section~\ref{sec:tracking-protocol-blocks} satisfies eq.~\eqref{eq:tracking-invariance}.

Suppose $A$ and $B$ share the electron-electron state $\bellstate{a}{b}$ and the virtual nuclear-nuclear state $\bellstate{c}{d}$ for some $a, b, c, d \in\{1, -1\}$, with correction Paulis $P_A^e$ ($P_B^e$) and $P_A^n$ ($P_B^n$) for $A$ ($B$).
In the first step of the protocol, $A$ and $B$ apply $P^n \cdot P^e$ to the electron-electron state, resulting in the electron-electron state
\[
	(P_A^n P_A^e \otimes P_B^n P_B^e) \bellstate{a}{b}
	=
	(P_A^n P_A^e \otimes P_B^n P_B^e) (P_A^e \otimes P_B^e) \bellstate{1}{-1}
	=
	(P_A^n \otimes P_B^n ) \bellstate{1}{-1}
	=
	\bellstate{c}{d}
	\]
where we made use of the fact that each Pauli squares to $\unit$.
In case of successful distillation, the virtual nuclear-nuclear state can be found using Prop.~\ref{prop:distillation-circuit-output-virtual} and equals $\bellstate{-ccd}{d} = \bellstate{-d}{d}$.
What remains is to determine the correction operators conditioned on the value of $d$.
If $d=1$, then the correction operators are $\unit$ for one node and $Y$ for the other (since $\unit \otimes Y \bellstate{-1}{-1}$ equals the target Bell state $\bellstate{1}{-1}$), while for $d=-1$ the resulting state is already the target Bell state and both correction operators should be $\unit$.
Determining the value of $d$ can be done by using the fact that eq.~\eqref{eq:tracking-invariance} was satisfied by the pre-distillation virtual nuclear-nuclear state,
\[
	(P_A^n \otimes P_B^n) \bellstate{c}{d} = \bellstate{1}{-1}
	\]
and thus $\bellstate{c}{d} = P_A^n \otimes P_B^n \bellstate{1}{-1}$.
From checking all possible cases of $P_A^n$ and $P_B^n$ we find that $d = 1$ precisely if one of $P_A^n, P_B^n$ equals $X$ or $Y$, while the other equals $\unit$ or $Z$.

\subsection{Correctness proof of the correction operator update for \swap}
\label{sec:tracking-proof-swap}

Here we show that eq.~\eqref{eq:tracking-invariance} holds for the state between nodes $A$ and $B$ after node $M$ has performed an entanglement swap on Bell states $A-M$ and $M-B$.
Let us denote $A$'s ($B$'s) correction operator as $P_A$ ($P_B$) and $M$'s correction operator as $P_M^A$ ($P_M^B$) for the state it shares with node $A$ ($B$).
That is, in the ideal case, the nodes hold the state
\begin{equation}
	(P_A \otimes P_M^A \otimes P_M^B \otimes P_B) (\bellstate{1}{-1}_{AM} \otimes \bellstate{1}{-1}_{MB})
	.
	\label{eq:pre-swap-state}
\end{equation}
We will make use of the fact that
	\begin{equation}
		\label{eq:pauli-ping-pong}
	(P \otimes Q) \bellstate{a}{b} \equiv (\unit \otimes PQ) \bellstate{a}{b}
	\end{equation}
	for single-qubit Paulis $P, Q$ and $a, b\in \{1, -1\}$, where $\equiv$ as before indicates that the two states differ only by a complex factor of norm 1 (in fact, for eq.\eqref{eq:pauli-ping-pong} we can restrict this to a multiplicative factor $\pm 1$).
	Using eq.~\eqref{eq:pauli-ping-pong}, we rewrite eq.~\eqref{eq:pre-swap-state} to
	\begin{equation}
		(P_M^A P_A \otimes \unit \otimes \unit \otimes P_M^B P_B) (\bellstate{1}{-1}_{AM} \otimes \bellstate{1}{-1}_{MB})
.
		\label{eq:pre-swap-state-rewritten}
	\end{equation}
	Eq.~\eqref{eq:pre-swap-state-rewritten} implies that we may assume that $M$'s two correction operators are both $\unit$.
	Thus $M$ only needs to communicate the correction operator that corresponds to having measured one qubit of each pair of the pair $\bellstate{1}{-1}\otimes\bellstate{1}{-1}$.
	The resulting correction operator $Q$ can be straightforwardly worked out in a similar way as in \cite{bennett1993teleporting} and the result is given in~\ref{table:pauli-correction}.

	The state after the entanglement swap is thus
	\[
		(P_M^A P_A \otimes P_M^B P_B Q) \bellstate{1}{-1}_{AB}
		\]
	which we rewrite using eq.~\eqref{eq:pauli-ping-pong} to
	\[
		(P_A \otimes P_M^A P_M^B P_B Q) \bellstate{1}{-1}_{AB}
		.
		\]
		Indeed, $P_A$ and $P_M^A P_M^B P_B Q$ are (modulo possible factor $-1$) the correction operators of node $A$ and $B$, respectively, after finishing the entanglement swapping protocol described above in section~\ref{sec:tracking-protocol-blocks} of this Supplementary Note.

	What remains is to convert the measurement outcomes from the circuit from Supplementary Figure~\ref{fig:nv-circuits}(d) to the measured Bell state.
For this, a direct computation shows that applying the circuit to the electron-nuclear state \mbox{$(\unit_e \otimes H_n) \bellstate{a}{b}$} (the Hadamard gate $H$ is needed since the nuclear qubit lives in a rotated basis, see section~\ref{app:nv-protocols}) yields the measurement outcomes $m_{\textnormal{earlier}} = -ab$ and $m_{\textnormal{later}} = b$.
Rewriting gives $a = -m_{\textnormal{earlier}}m_{\textnormal{later}}$ and $b=m_{\textnormal{later}}$.

\section{Atomic Ensemble physical modelling}
\label{app:atomic-ensembles}
Here, we provide the details of the simulation comparing different memory technologies for atomic-ensemble quantum repeaters, whose results are presented in section~\nameref{sec:ae_memory_comp} of the Result. For a more detailed discussion see \cite{maier2020}.

\subsection{Generating end-to-end entanglement with Atomic Ensembles}

In our experiment we simulate the protocol proposed by Sinclair et al. \cite{sinclair2014} and analysed in detail by Guha et al. \cite{guha2015rate}.
However, in our simulation we go further than any previous analysis by not only including dark counts and detector efficiency but also multi-photon emissions, photon distinguishability and time-dependent memory efficiency.
For a detailed review of different atomic-ensemble protocols see Sangouard et al. \cite{sangouard2011quantum}.
Let us briefly review the main points necessary to understand our results.

\subsubsection{The protocol}

The simulation setup consists of two elementary links connected by a quantum repeater station. 
Each elementary link contains two photon-pair sources sending one half of the produced state to a midpoint station through optical fibres.
We assume that the speed of all photons and classical communication is $c/n_{ri}$, where $c$ is the speed of light in vacuum and $n_{ri} = 1.44$ is the refractive index of glass \cite{paschotta2020fibers}.
The midpoint station contains a 50:50 beam splitter and two photon detectors to perform a linear-optical Bell state measurement (BSM) \cite{sangouard2011quantum}.
Detection of a single photon per time-bin heralds successful elementary link entanglement generation.
By performing the entanglement attempts on the elementary link simultaneously in multiple modes $M$ (e.g. different temporal, spatial or frequency modes) the probability of successfully generating elementary link entanglement can be greatly increased as $p = 1 - (1 - p_{single-mode})^M$.
This is called multiplexing.

The repeater station also contains two quantum memories, which store the second half of the quantum state emitted by the adjacent sources.

If both elementary links herald a successful BSM at their respective midpoint station, the successful modes are extracted from the memories.
Another BSM is then performed on those modes at the repeater station.

The second halves of the quantum state at the end nodes are measured directly in either the X or Z basis without storing them on a memory first.

The photon-pair source generates a time-bin encoded superposition of photon pairs with a joint quantum state given by \cite{krovi2016}

\begin{equation}\label{eq:ae_source_state}
\ket{\psi} = \sum_{n=0}^\infty\sqrt{p(n)}\ket{\psi_n},
\end{equation}

where 

\begin{align}
\ket{\psi_n} & = \frac{1}{\sqrt{n+1}}\sum_{m=0}^n (-1)^m \ket{n-m,m;n-m,m}. \label{eq:spdc_state}
\end{align}

Here $p(n)$ is the probability to generate $n$ pairs and a state $\ket{n-m,m;n-m,m}$ describes $n-m$ photons travelling to the left (/right) side in the early and $m$ in the late time window.

    Photon loss in the optical fibres is modelled following the beam-splitter-attenuator channel \cite{ivan2011operator}, which is also referred to as generalised amplitude damping \cite{chuang1997}.
    Such a channel performs the map $\rho \mapsto \sum_{k=0}^{\infty} A_k \rho A_k^{\dagger}$ where the Kraus operator $A_k$ corresponds to the case where exactly $k$ photons are lost and is given by
\begin{equation*}
A_k = \sum_{n=k}^{\infty} \sqrt{\binom{n}{k}} \sqrt{(1-\gamma)^{n-k} \gamma^k} \ket{n-k}\bra{n} ,
\end{equation*}
where $\gamma$ is the probability of losing a photon.

For the quantum channel it depends on the fibre attenuation $\alpha$, the channel length $L$ and the coupling loss $\gamma_{chan}(0)$ as $\gamma_{chan}(L) = 1 - (1 - \gamma_{chan}(0)) \times 10^{- \alpha L / 10}$.

The linear-optical BSM is modelled as a set of POVMs which include non-unit Hong-Ou-Mandel dip visibility (photon distinguishability), detection efficiency and dark counts.
We calculate them as a set of perfect POVMs $M_{mn}$ and then derive the set of effective POVMs $M'_{kl} = \sum_{mn} p(mn \rightarrow kl) M_{mn}$, where $p(mn \rightarrow kl)$ is the probability that an imperfect BSM measures $m$(/$n$) photons arriving on a detector as $k$(/$l$) due to dark counts and detection inefficiency.
For a complete derivation of the POVM elements see \cite{maier2020}.

The quantum memory is also modelled as a quantum channel with generalised amplitude damping using the same operators $A_k$.
However, here the probability of losing the photon $\gamma$ depends on the time $t$ that the quantum state spent on the memory, the maximum memory efficiency $\eta$ (at $t=0$) and the memory lifetime $\tau$.
The memory efficiency then decays either exponentially (AFC) or Gaussian (EIT) with $t/\tau$ and $\gamma$ behaves as

\begin{align}
\gamma_{AFC}(t) & = 1 - \eta(0) \times e^{- t / \tau} \\
\gamma_{EIT}(t) & = 1 - \eta(0) \times e^{- \frac{1}{2}(t /
\tau)^2} .
\end{align}

For a more detailed discussion of the protocol and the reasoning behind our model please see our upcoming work \cite{maier2020}.

\subsubsection{The figures of merit}
In order to compare the two memory technologies, we also need to define figures of merit we want to investigate.
To this end we choose the secret key rate (SKR) obtainable for quantum key distribution using the BB84 protocol \cite{bennett2014}, the quantum bit error rate (QBER) and the average number of attempts per successful end node measurement.

For the QBER we compare the end node measurement outcomes with the expected correlations in each basis.
The fraction of wrong bits among the measurement outcomes in each basis is then called the QBER and is determined by performing end node measurements in either X or Z basis 

The secret key rate $R_{SK}$ is then computed as
\begin{equation}\label{eq:skr}
R_{SK} = \max(0 , 1 - H(Q_x) - H(Q_z)) / (2 A_S) ,
\end{equation}
where $H(x) = -x\log(x) - (1-x)\log(1-x)$ is the binary entropy, $Q_x, Q_z$ are the QBERs in the X and Z bases and $A_S$ is the average number of attempts per successful end-to-end entanglement generation.
The factor $1/2$ (in eq.~\ref{eq:skr} and eq.~\ref{eq:skr_error}) accounts for the fact that in the BB84 protocol the two end nodes randomly choose between measuring in X or Z basis and therefore only measure in the same basis in $1/2$ of the successful entanglement generations.

Errors for the number of attempt are calculated from the standard deviation of the mean.
The error on the secret key rate is computed as follow from the error on the number of attempts per success $\Delta A_S$ and the errors $\Delta Q_x$ and  $\Delta Q_z$ on the QBER in X and Z basis respectively:

\begin{equation}\label{eq:skr_error}
\Delta R_{SK} = \frac{1}{2 A_S}\sqrt{\bigg(\frac{\partial H(Q_x)}{\partial Q_x} \Delta Q_x \bigg)^2 + \bigg(\frac{\partial H(Q_z)}{\partial Q_z} \Delta Q_z \bigg)^2 + (R_{SK} \times \Delta A_S)^2}.
\end{equation}

\subsection{Atomic Frequency Comb and Electronically Induced Transparency quantum memories}

In this work, we study two promising types of atomic-ensemble memories, both based upon photon absorption: electronically induced transparency (EIT) quantum memories as an example of an optical control protocol and atomic frequency comb (AFC) memories as an example for an engineered absorption protocol.
For a more detailed comparison of these two types of memories see \cite{maier2020}.
AFC memories are very promising due to their great potential for multiplexing, while EIT memories promise superior efficiency with limited multiplexing.
Here we will give a brief overview of both technologies.

\subsubsection{AFC}
AFC memories require an inhomogeneously broadened material (e.g. rare-earth-doped crystals) with equidistant peaks in occupation density created by spectral hole burning \cite{afzelius2009multimode}.
When a photon is absorbed by the memory, the system will be in a collective delocalised state (Dicke state \cite{dicke1954}) which then rapidly dephases with time.
After a fixed time the Dicke state rephases and the absorbed photon is re-emitted.
The retrieval efficiency decreases exponentially with storage time.
Due to the large linewidth of the inhomogeneously broadened state, AFC memories have great potential for massive multiplexing as the number of modes is not limited by the optical depth of the material.
Currently, the efficiencies that have been experimentally achieved are up to 58\% \cite{sabooni2013} reported for a cavity-enhanced AFC protocol and around 8.5\% \cite{seri2019, jobez2016} for most multiplexed memories.
Memory lifetimes of up to 0.53 s \cite{holzapfel2020} have been reported for a dynamically decoupled AFC protocol (however with an efficiency of just 0.5\%).
In our simulation we set a maximum memory efficiency of 45\% (currently optimistic value for multimode memories, but still below the highest demonstrated single-mode storage \cite{sabooni2013}) with 1000 modes, consisting of 50 spectral and 20 temporal modes (15 x 9 already demonstrated \cite{seri2019}, over 1000 total modes also demonstrated \cite{bonarota2011}), and a lifetime of 1 ms (0.542 ms already achieved for multimode on-demand AFC memory \cite{jobez2016}).

\subsubsection{EIT}
EIT \cite{harris1990} is an effect where an opaque sample of atomic-ensembles becomes transparent due to interaction with electromagnetic fields.
This creates a steep dispersion which in turn affects the group velocity of incoming light.
The effect can be used to slow or completely stop and later re-emit incoming photonic states.
Due to the use of a strong resonant control laser to store the incoming light, these kind of memories are subject to background noise in the form of additional photons introduced by the control beam.
These additional photons need to be compensated by filtering.
Another disadvantage is that multiplexing in EIT is limited to using only the spatial degree of freedom, thus making it more difficult to realise a highly multiplexed protocol.
The great advantage of these memories is the high efficiency of 85\% \cite{cao2020} at a memory lifetime of 15 $\mu$s.
In our simulation we set a maximum efficiency of 90\% (within reach of \cite{cao2020}) with 1000 spatial modes (experimentally 2 spatial modes have been demonstrated \cite{vernaz-gris2018}), and a lifetime of 100 $\mu$s (\cite{cao2020} gives 200$\mu$s as a theoretical limit) .
For a review of EIT see \cite{fleischhauer2005electromagnetically}.

\subsection{Parameters}
For the comparison we used the optimistic elementary link parameters taken from \cite{guha2015rate}(Fig 10 (g)).
We then added optimistic values for the the memory parameters (see above) and the additional noise parameters we use, such as the photon visibility.
In line with the analysis of \cite{guha2015rate}, we model the source as an almost perfect single-pair source with a small probability $p(2)$ of emitting two pairs.

We first list the common parameters of both simulations before we compare the parameters for the memory components in Supplementary Table \ref{tab:mem_params}.
It is worth pointing out that using the same number of modes for both technologies slightly biases this comparison as the high multiplexing potential is the main advantage of AFC memories.
\begin{itemize}
	\item Source total number of modes (spectral x temporal x spatial) = 1000
	\item Source emission probabilities: $p(0) = 1 - p(1)-p(2)$, $p(1) = 0.9$, $p(2)=0.013$ (see eq.\ref{eq:ae_source_state})
	\item Fibre coupling loss probability $\gamma_{chan}(0) = 0$
	\item Fibre attenuation $\alpha = 0.15$dB/km 
	\item Visibility = 0.9  
	\item Detector dark count probability = $10^{-6}$
	\item Detector detection efficiency = 90\%
	\item Detectors number-resolving: No
\end{itemize}

\begin{table}[h]
	\begin{tabular}{l|l|l}
		& \textbf{AFC} & \textbf{EIT} \\ \hline
		spectral modes           & 50           & 1            \\ \hline
		temporal modes           & 20           & 1          \\ \hline
		spatial modes            & 1            & 1000           \\ \hline
		maximum efficiency ($t=0$) & 45\%         & 90\%         \\ \hline
		memory lifetime          & 1 ms          & 0.1 ms        \\ \hline
		time dependence          & exponential  & Gaussian    
	\end{tabular}
\caption{
    \textbf{Parameters for the two different atomic-ensemble memory technologies.}
}
\label{tab:mem_params}
\end{table}


\end{document}